\documentclass[letterpaper, DIV=17]{scrartcl}
\usepackage{amssymb, amsthm, mathtools,bbm,bm}
\usepackage[square,sort,comma,numbers]{natbib}

\usepackage[none]{hyphenat}
\usepackage[english]{babel}
\usepackage{enumerate}
\usepackage{amsmath}
\usepackage{amsthm}
\usepackage{amssymb}
\usepackage{dsfont}%
\usepackage{hyperref}
\usepackage{amssymb}
\usepackage{color}
\usepackage[utf8]{inputenc}
\usepackage{graphicx}
\usepackage{tikz}
\usepackage{tikz-cd}
\usepackage{verbatim} 
	\usepackage{mathtools} 
	\usepackage{mathabx} 
	\usepackage{hyperref}
	\usepackage{comment}
	\usetikzlibrary{decorations.markings}
	\usepackage{bbm} 
	\usepackage{xfrac} 
	\usepackage{tensor} 
	\usepackage{listings}  
	\usepackage{matlab-prettifier} 
	
	\usepackage{pgfplots}              
	\pgfplotsset{compat=newest}
	\usepgfplotslibrary{fillbetween}

	\usepackage[title]{appendix}

\usepackage[shortcuts]{extdash}  
	
	\newtheorem{theorem}{Theorem}[section]  
	\newtheorem{cor}[theorem]{Corollary}
	\newtheorem{prop}[theorem]{Proposition}
	\newtheorem{lemma}[theorem]{Lemma}

	\numberwithin{equation}{section}
	\theoremstyle{definition}
	\newtheorem{defn}{Definition}[section] 
	\newtheorem{rem}[theorem]{Remark}

	\newcommand{\pol}{Pol}
	\DeclareMathOperator{\dist}{dist}
	\DeclareMathOperator{\supp}{supp}
	\DeclareMathOperator{\spa}{span}
	\newcommand{\loc}{\mathrm{loc}}
	\newcommand{\irr}{\small \textrm{irreducible}}

\begin{document}
		\title{Fractional Quantum Hall States:\\ Infinite Matrix Product Representation\\ and its Implications}
		
	\author{Severin Schraven$^{1}$, Simone Warzel$^{1,2,3}$ \\
	\small $^1$ Department of Mathematics, TU Munich, Germany \\[-1ex]
	\small  $^2$ Munich Center for Quantum Science and Technology, Munich, Germany \\[-1ex]
	\small  $^3$ Department of Physics, TU Munich, Germany\\
	\small  Corresponding author: severin.schraven@cit.tum.de}
		
		\date{\small\today}	
		\maketitle

		\begin{abstract}
		\noindent
			We present a novel matrix product representation of the Laughlin and related fractional quantum Hall wavefunctions based on a rigorous version of the correlators of a chiral quantum field theory. This representation enables the quantitative control of the coefficients of the Laughlin wavefunction times an arbitrary monomial symmetric polynomial when expanded in a Slater determinant or permanent basis. It renders the properties, such as factorization and the renewal structure, inherent in such fractional quantum Hall wavefunctions transparent. We prove bounds on the correlators of the chiral quantum field theory and utilize this representation to demonstrate the exponential decay of connected correlations and a gap in the entanglement spectrum on a thin cylinder. 	\end{abstract}

\bigskip
	\tableofcontents
\bigskip
%

\section{Introduction and main results}

Models of the fractional quantum Hall effect~\cite{QHEOxford2003,Simon21} start from Laughlin's proposal \cite{PhysRevLett.50.1395} of a highly correlated, quantum wavefunction for $ N $ indistinguishable spinless particles. 
In the planar geometry and up to a normalization, this function takes the Jastrow form 
\begin{equation}\label{eq:Laughlin}
 \prod_{1\leq j< k \leq N} (z_j-z_k)^q \ \times \ \exp\left( -\frac{1}{2}\sum_{j=1}^N |z_j|^2 \right) , 
\end{equation}
i.e., an integer power $ q \in \mathbb{N} $ of the Vandermonde determinant times a Gaussian weight stemming from the fact that the single-particle Hilbert space is the lowest Landau level. The latter is the Bargmann space of entire functions that are square integrable with respect to a Gaussian weight.
Here and in the following, we use physical units in which the magnetic length is set to $1/\sqrt{2} $. The parity of $ q $ determines whether the wavefunction~\eqref{eq:Laughlin} describes bosons
(even) or fermions (odd).

Revealing the intricate structure behind this class of wavefunctions has continued to fascinate many communities. The Coulomb gas perspective regards the wavefunction as the Gibbs-Boltzmann distribution of an interacting gas of charged particles with logarithmic interactions in a uniform background known as Jellium~\cite{di1994laughlin,Aizenman:2010aa,jansen2012fermionic,RSY:2014,Lieb:2019vl,Roug19}.  This point of view reveals that the number of particles per area (measured in terms of the square of the magnetic length)  in the above wavefunction is $ 1/q $.  Important for us will be the connection of
Laughlin's function and generalizations to the correlation functions of primary
fields in a conformal field theory (CFT), which goes back to~\cite{moore1991nonabelions,ReadMoore92}. 
The topic, which we will pursue here, is to investigate the expansion of such wavefunctions with respect to a canonical orthonormal basis, i.e., the Slater determinants or permanents (depending on whether $ q $ is odd or even) of canonical one-particle orbitals.  Such expansions have been studied early on~\cite{dunne1993slater, di1994laughlin}.  It has been pointed out that the expansion coefficients take the form of an 'infinite matrix product' (iMPS) of operators in a CFT  \cite{PhysRevB.86.245305,Bernevig:2009aa,Cirac:2010aa,thomale2011decomposition, estienne2013fractional,schossler2022inner}. Our goal is to establish these results on a sound mathematical foundation and present the first analytical results within this framework. 
Products of the Vandermonde determinant and symmetric polynomials define the Jack polynomials, for which expansions of this form are also investigated in a purely algebraic context \cite{ballantine2012powers,Cai:2014cz}.  

Our results have been inspired by previous works on the subject and in particular~\cite{jansen2009symmetry,estienne2013fractional,nachtergaele2021spectral}. The iMPS representation extends beyond that of \cite{estienne2013fractional} or \cite{Cai:2014cz}, as it enables analytical control of the expansion coefficients and facilitates easy proofs of structural properties, such as factorization and the associated renewal structure, which,  as far as quantitative analytical bounds are concerned, go far beyond previous results in~\cite{jansen2009symmetry,jansen2012fermionic}.

\subsection{Expansions of powers of the Vandermonde times a symmetric polynomial} 
To set the stage, we consider the backbone of the expansion, namely, that of the $q$th power of the Vandermonde determinant in terms of the standard monomial basis. We will prove that 
\begin{equation}\label{eq:VandermondeMPS}
 \prod_{1\leq j< k \leq N} (z_j-z_k)^q  =(-1)^{qN(N-1)/2}  \!\!  \sum_{\lambda \preceq \lambda^{(q)}_N}   \frac{ \langle 0 \vert \mathbb{W}(\lambda) \vert 0 \rangle }{ M(\lambda)!}  \sum_{\sigma\in S_N} \mathrm{sgn}(\sigma)^q   \prod_{j=1}^N z_{\sigma(j)}^{\lambda_j}
\end{equation}
where the first summation on the right side extends over all integer partitions $ \lambda = (\lambda_1, \dots , \lambda_N) \in \mathbb{N}_0^N $, i.e.\ 
$ \lambda_1 \leq \dots \leq \lambda_N $,
which are dominated in the natural order on partitions by the so-called root partition
$
\lambda^{(q)}_N \coloneqq (0, q , 2q, \dots , q (N-1) ) 
$, which, in particular, entails that
$$
 \sum_{j=1}^N \lambda_j = \frac{q}{2} N(N-1) ,
$$
cf.\ Definition~\ref{def:partitions} for these notions. The second summation in~\eqref{eq:VandermondeMPS} is over all permutations of $ N $ elements, and $\mathrm{sgn}(\sigma) $ stands for the signature of the permutation $ \sigma \in \mathcal{S}_N $. The prefactor 
 $$ M(\lambda)! \coloneqq \prod_{k=0}^\infty m(\lambda, k)! $$ involves the occupation numbers $ m(\lambda, k) $ of the $k$th monomial, cf.\ Definition~\ref{def:partitions}. 
The central player in~\eqref{eq:VandermondeMPS} is the product
$\mathbb{W}(\lambda) =  W_{\lambda_N-q(N-1)} W_{\lambda_{N-1}-q(N-2)} \dots W_{\lambda_2 - q} W_{\lambda_1}  $
of operators on the 'virtual' Hilbert space of a chiral CFT.  In bra-ket notation, which we use for the virtual Hilbert space,  the CFT's vacuum is denoted by $ | 0 \rangle $. This 'virtual' Hilbert space is introduced in Section~\ref{sec:iMPS} and the definition of the unbounded operators $ W_{\lambda_j}$, including a domain, can be found in Definition~\ref{def:DW}. \\

The Laughlin wavefunction accommodates the most particles per area. Less compressed fractional quantum Hall states with lower filling fraction are obtained by multiplying~\eqref{eq:VandermondeMPS} with an arbitrary symmetric polynomial~\cite{Simon21}. 
An (algebraic) basis of the ring $ \mathbb{C}[x_1, \dots, x_{N}]^{S_N} $ of symmetric polynomials with complex coefficients is given by products of the  power-sum symmetric polynomials 
$
			p_n(z_1, \dots, z_N) \coloneqq \sum_{j=1}^N z_j^n $, $n \in \mathbb{N}_0 $.
More convenient for us will be the monomial symmetric polynomial, 
		\begin{equation} \label{def monomial symmetric polynomial}
			m_{b}(z_1, \dots, z_N) \coloneqq  \frac{1}{ M(b)!} \sum_{\sigma\in S_N} \prod_{j=1}^N z_{\sigma(j)}^{b_j} ,
		\end{equation}
		which are again enumerated by integer partitions $b =(b_1,\dots, b_N) $. 
		The monomial symmetric polynomials also form a basis of the set of symmetric polynomials. Consequently, for a given partition $ b $ of the integer $ | b | \coloneqq \sum_{j=1}^N b_j $, there are transition polynomials, which allow to express the monomial polynomials in terms of a real polynomial $ \pol_b \in \mathbb{R}[x_1, \dots, x_{\vert b\vert}] $ of the power symmetric polynomials  \cite[Prop. 7.7.1]{stanley2001enumerative}
		\begin{equation} \label{def Fb}
				m_b(z_1, \dots, z_N) = \pol_b(p_1(z_1, \dots, z_N), \dots,   p_{\vert b \vert}(z_1, \dots, z_N)).
			\end{equation}

			The generalization of~\eqref{eq:VandermondeMPS} to fractional quantum Hall states,  which are less compressed than the Laughlin state and uniquely addressed by $ (N,b) $, involves the following notions.
			\begin{defn}\label{def:rootp}
			For $ N \in \mathbb{N} $ and a partition $ b  \in \mathbb{N}_0^N $, which encodes a monomial symmetric polynomial, we call the pair $(N,b) $ the \emph{root}, and 			
			\begin{align*}
				\lambda^{(q)}_N(b )\coloneqq (b_1, q +b_2,  \dots, q(N-1)+b_N) ,
			\end{align*}
			the \emph{root partition}.
			\end{defn}	
			If $ b = 0 $, then $ \lambda^{(q)}_N(0)= \lambda^{(q)}_N $. 
			Given a root $ (N,b) $ and another partition $ \lambda \in \mathbb{N}_0^N $, we define the operator products
			\begin{equation}\label{def:Wwithb}
			\mathbb{W}(\lambda,b) \coloneqq W_{\lambda_N-q(N-1)-b_N} W_{\lambda_{N-1}-q(N-2)-b_{N-1}} \dots W_{\lambda_2 - q-b_2} W_{\lambda_1-b_1}  ,
			\end{equation}
			for which $\mathbb{W}(\lambda,0)\equiv \mathbb{W}(\lambda) $, cf.\ Definition~\ref{def:DW}.  Our main theorem generalizing~\eqref{eq:VandermondeMPS} 
			then reads as follows. 
\begin{theorem} \label{thm:root states}
			Given a root $ (N,b) $, then
			\begin{align} \label{MPS with monomial symmetric}
					&m_{b}(z_1, \dots, z_N) \times \prod_{1\leq j<k\leq N} (z_j-z_k)^q \notag \\
					& = (-1)^{qN(N-1)/2}  \sum_{\lambda \preceq \lambda^{(q)}_N(b)} \frac{  w_b(\lambda) }{ M(\lambda)!}  \sum_{\sigma\in S_N} \mathrm{sgn}(\sigma)^q \prod_{j=1}^N z_{\sigma(j)}^{\lambda_j},  
			\end{align}
			where the first summation on the right side extends over all integer partitions $ \lambda $ dominated by the root partition $ \lambda^{(q)}_N(b ) $ (cf.~Definition~\ref{def:partitions}). The expansion coefficient is given by
				\begin{align} \label{MPS with monomial symmetric_c}
			w_b(\lambda) & = \frac{1}{ M(b)!}  \sum_{\tau \in \mathcal{S}_N}  \langle 0 \vert \mathbb{W}(\lambda,b_\tau)\vert 0 \rangle  \notag \\ 
			& = \langle 0 \vert \pol_{b}\left(\frac{a_1}{\sqrt{q}} , \dots, \frac{a_{\vert b \vert}}{\sqrt{q}} \right) \mathbb{W}(\lambda) \vert 0 \rangle .
				\end{align}
			Here $ b_\tau \coloneqq \big(b_{\tau(1)}, \dots , b_{\tau(N)}\big) $, and $\pol_{b}\in \mathbb{R}[x_1, \dots, x_{\vert b\vert}]$ is the unique polynomial with real coefficients such that~\eqref{def Fb}. 
			For the root partition, we have $ w_b( \lambda^{(q)}_N(b)) = 1 $. 
			\end{theorem}
The proof of this theorem is found in Subsection~\ref{subsec:ProofofRoot}.

\begin{rem}
In the physics literature, the additional annihilation operators described by  $\pol_{b} $ in~\eqref{MPS with monomial symmetric_c} are understood as boundary charges in the CFT. In this language, the second line is a CFT-correlator with out-state $\pol_{b}(a_1^*, \dots, a_{\vert b \vert}^*) \vert 0\rangle$ (see \cite{estienne2013fractional,schossler2022inner} and references therein). In this context, it might be useful to note that the coefficients of the transition polynomial $\pol_b$ depend on the partition $b$, but not explicitly on $N$, i.e.\ if $\tilde{b}=(0, \dots, 0, b_1, \dots, b_N)$ then $\pol_{b}=\pol_{\tilde{b}}$.

In case of  fermions ($ q $ odd), one may restrict the summation in~\eqref{MPS with monomial symmetric} to partitions without double occupancy, i.e.\ $ \lambda_1 < \lambda_2 < \dots < \lambda_N $, for which $ M(\lambda)! = 1 $. 
\end{rem}

\subsection{Quantum Hall geometries, root partitions and iMPS representation}
It is well known that the core of the structure of fractional quantum Hall wavefunctions is geometry-independent \cite{estienne2013fractional,10.1063/1.5046122}.
However, to connect Theorem~\ref{thm:root states} to the wavefunctions used in specific models of the fractional quantum Hall effect, we need to fix a geometry and with it a Hilbert space of one-particle states. 
We describe this in the planar case, and,  in more detail, the case of a cylinder geometry, since much of the analysis, which is presented later, pertains to the latter case.
\subsubsection{Planar geometry}
In the planar geometry, in which the particles roam the complex plane $ \mathbb{C} $, the one-particle Hilbert space is the Bargmann space 
$$ \mathcal{P} \coloneqq \left\{ f:\mathbb{C} \to \mathbb{C} \ : \ \mbox{$f $ is entire and} \; \int |f(z)|^2 e^{- |z|^2} dz < \infty \right\} . $$
An orthonormal basis is provided by the functions
\begin{equation}\label{eq:orbitalsp}
\varphi_k(z) \coloneqq \frac{z^k}{\sqrt{\pi k!} } ,  \quad k \in \mathbb{N}_0 .
\end{equation}
Physically speaking, this is the eigenbasis of the angular momentum operator, and $ k \in \mathbb{N}_0 $ labels the angular-momentum orbitals on the plane. 
The Hilbert space for $ N $ fermions or bosons is the $ N $-fold antisymmetric ($q$ odd) or symmetric ($q$ even) tensor-product Hilbert space $ \mathcal{F}_N \coloneqq \bigotimes_{q}^N \mathcal{P} $. An orthonormal basis in the latter is labeled by partitions $ \lambda = (\lambda_1 , \dots , \lambda_N) \in \mathbb{N}_0^N$:
\begin{equation}\label{eq:Slater}
\Phi_\lambda(z_1,\dots, z_N) \coloneqq \frac{1}{\sqrt{M(\lambda)! \ N!}} \sum_{\sigma\in \mathcal{S}_N} \left(\textrm{sgn}\ \sigma\right)^q \ \prod_{j=1}^N \varphi_{\lambda_j}(z_{\sigma(j)})  . 
\end{equation} 
In other words, the $ \Phi_\lambda $'s are the normalized Slater determinants or permanents of the orthonormal single-particle basis~\eqref{eq:orbitalsp}.
Varying $ N $, the $ \Phi_\lambda $'s make up the occupation number basis corresponding to the angular momentum orbitals of the fermionic or bosonic Fock space $ \mathcal{F} \coloneqq \bigoplus_{N=0}^\infty  \mathcal{F}_N $.

In the planar setting, the fractional quantum Hall model wavefunctions, which include the Laughlin state~\eqref{eq:Laughlin}, are uniquely labeled by roots $ (N,b) $ and given by
\begin{align}\label{def:WF}
 \Psi_{b,N}(z_1,\dots , z_N)  \coloneqq \kappa_{b,N} \ m_{b}(z_1, \dots, z_N) \times \mkern-20mu \prod_{1\leq j<k\leq N} (z_j-z_k)^q
 \end{align} 
 where we choose the normalization
\begin{equation}\label{eq:kappanorm}
	\frac{1}{\kappa_{b,N}} \coloneqq(-1)^{qN(N-1)/2}  \sqrt{N!}  \ \prod_{j=1}^N \sqrt{(q(j-1)+b_j)! \ \pi } . 
\end{equation}
We note that $\Psi_{b,N}$ is not normalized in $\mathcal{F}_N$. The normalizing factor is chosen in a manner that makes the underlying structure more transparent.\\

An immediate corollary of Theorem~\ref{thm:root states} is the following result on the expansion coefficients of these wavefunctions in terms of the occupation basis.
\begin{cor}\label{cor:planar}
Given a root $(N,b) $, the wavefunction~\eqref{def:WF} has the representation
$$ \Psi_{b,N} =  \sum_{ \lambda \preceq \lambda^{(q)}_N(b)  } h_b (\lambda) \ \Phi_\lambda  $$ with expansion coefficients of the form
\begin{equation} \label{eq:hb}
h_b(\lambda) = \frac{g_b(\lambda) \ w_b(\lambda)}{\sqrt{M(\lambda)!}}   , \quad \mbox{with}\quad g_b(\lambda) \coloneqq  \prod_{j=1}^N\sqrt{ \frac{\lambda_j!}{(q(j-1)+b_j)!}} .
\end{equation}
The normalization of $  \Psi_{b,N} $ is such that the coefficient of the root partition $  \lambda^{(q)}_N(b) $ is unity, $ h_b\left(  \lambda^{(q)}_N(b) \right) = 1 $.
\end{cor}
The proof can be found in~Subsection~\ref{subsec:ProofofRoot}. 
The basic structure of the expansion coefficient~\eqref{eq:hb} is similar in other geometries:  $ g_b(\lambda) $ is geometry-dependent, and the rest remains invariant. 

\subsubsection{Cylinder geometry}

In this work, we primarily focus on the cylinder geometry, which is derived from the planar geometry through the conformal transformation $ z \mapsto e^{\gamma z} $ with parameter $ \gamma > 0 $. The single-particle space is then the Bargmann space
$$
 \left\{ f: \mathbb{R} + i \big[0,\tfrac{2\pi}{\gamma}\big)  \to \mathbb{C} \ : \ \mbox{$f $ is entire and} \; \int_{ \mathbb{R}\times\big[0,\tfrac{2\pi}{\gamma}\big) }  |f(x+iy)|^2 e^{- x^2} dx dy < \infty \right\} . 
$$
An orthonormal basis of this Hilbert space is given by
\begin{equation}\label{eq:cylinderONB}
	 \varphi_k(z) \coloneqq \sqrt{\frac{\gamma }{2\pi^{3/2}}} e^{- \gamma^2 k^2/2} \exp\left( \gamma z k \right) , \quad k \in \mathbb{Z} .
\end{equation}
These functions represent orbitals lined up on a cylinder of radius $ \gamma^{-1} $ roughly at $\gamma k $ as their $ x $-coordinate. 

In the following, we will restrict attention to the closed subspace 
$$ \mathcal{B} \coloneqq \overline{\spa\left\{  \varphi_k  \ | \ k \in \mathbb{N}_0\right\}} $$ 
spanned by orbitals with $ k \in \mathbb{N}_0 $. 
An orthonormal basis $ \Phi_{\lambda} $ for the Hilbert space $ \mathcal{F}_N \coloneqq \bigotimes_{q}^N \mathcal{B}$ of $ N $ fermions or bosons is again labeled by partitions $ \lambda = (\lambda_1, \dots , \lambda_N) \in \mathbb{N}_0^N $ and given by 
the Slater determinants or permanents~\eqref{eq:Slater} in which one substitutes the function~\eqref{eq:cylinderONB}. 
By varying $ N $, this basis is naturally identified with the occupation number basis of the associated Fock space $ \mathcal{F}= \bigoplus_{N=0}^\infty  \mathcal{F}_N $. 

Fractional quantum Hall model wavefunctions, which live on the non-negative $ x $-axis, take the form of a symmetric polynomial of $ e^{\gamma z} $ times the corresponding Laughlin function,
\begin{equation}\label{def:WFc}
 \Psi_{b,N}(z_1,\dots , z_N)  = \kappa_{b,N} \ m_{b}\left(e^{\gamma z_1}, \dots, e^{\gamma z_N}\right) \times \mkern-10mu \prod_{1\leq j<k\leq N} \left(e^{\gamma z_j}-e^{\gamma z_k} \right)^q .
 \end{equation}
 These functions are indexed by partitions $ b \in \mathbb{N}_0 $, and
we choose the normalization
\begin{equation}\label{eq:norm2cyl}
	\frac{1}{\kappa_{b,N} } \coloneqq (-1)^{qN(N-1)/2}  \sqrt{N!} \ \sqrt{ \frac{2 \pi^{3/2}}{ \gamma}}^{N} \exp\left( \frac{\gamma^2}{2} \sum_{j=1}^N \left[ q(j-1)+ b_j\right]^2\right) .
\end{equation}
Again, we do not choose $\kappa_{b,N}$ such that $\Psi_{b,N}$ is normalized in $\mathcal{F}_N$, but such that the structure is more transparent. An immediate consequence of Theorem~\ref{thm:root states} is the following representation of these states in terms of the occupation basis. 
\begin{cor}\label{cor:cyl}
Given a root $(N,b) $,  the wavefunction~\eqref{def:WFc} has the representation
$$ \Psi_{b,N} =  \sum_{ \lambda \preceq \lambda^{(q)}_N(b)  } h_b(\lambda) \ \Phi_\lambda  , \qquad h_b(\lambda) \coloneqq \frac{g_b(\lambda) \ w_b(\lambda)}{\sqrt{M(\lambda)!}} $$
 in terms of the orthonormal Slater determinants or permanents $  \Phi_\lambda $ given by~\eqref{eq:Slater} with $ \varphi$ substituted by~\eqref{eq:cylinderONB} and
\begin{equation}
	g_b(\lambda) \coloneqq 
	    \displaystyle \exp\left( - \frac{\gamma^2}{2} \Delta_b(\lambda) \right) , \quad \mbox{where $ \displaystyle \Delta_b(\lambda) \coloneqq \sum_{j=1}^{N} \left[\left( q(j-1)+ b_j\right)^2-\lambda_j^2\right]$.} 
\end{equation}
The normalization of $  \Psi_{b,N}  $ is such that the coefficient of the root partition $  \lambda^{(q)}_N(b) $ is unity, $ h_b\left(  \lambda^{(q)}_N(b) \right) = 1 $.
\end{cor}
The proof can be found in~Subsection~\ref{subsec:ProofofRoot}.

\subsection{Tilings, squeezings and factorization over irreducible segments} \label{sec:factor}
As pointed out in \cite{jansen2009symmetry} for the case $ b = 0 $ in the cylinder geometry, the expansion coefficients $ h_b(\lambda) $ exhibit a renewal structure, i.e.\ they factorize over irreducible segments of the partition $ \lambda  \preceq \lambda^{(q)}_N(b) $.  This is most easily explained by associating the root $ (N,b) $ and its partition $  \lambda^{(q)}_N(b) $ with an occupation configuration $ \mathbf{m}\big(  \lambda^{(q)}_N(b)\big) $ and segmenting this configuration into tiles and elementary blocks.  To do so, we note that for any partition~$ \lambda $, there is a unique occupation configuration $ \mathbf{m}(\lambda) = (m(\lambda,0), m(\lambda,1), \dots ) $ of the non-negative orbitals. In this way, the root partition $  \lambda^{(q)}_N(b) $ corresponds to the configuration which has exactly one particle in each of the orbitals $ \big( \lambda^{(q)}_N(b)\big)_{j } $ with $ j = 1 , \dots , N $. This occupation pattern uniquely corresponds to a root tiling of the orbital set $ \{ 0, 1, 2, \dots , q N + b_N -1 \} $  with
\begin{itemize}
\item 'monomers' of length $ q $ with particle configuration $ 1 0^{q-1} $, and 
\item 'voids' of length $ 1 $ with particle configuration $ 0 $, which are placed in between the monomers. 
\end{itemize}
The $ j$th monomer starts at $ \big( \lambda^{(q)}_N(b)\big)_{j } $, and the total number of voids ahead of the $ j $th monomer equals $ b_j $ with $ j  = 1 , \dots , N $, see Figure~\ref{fig:tiling}. 
\begin{figure}[h]\label{fig:tiling}
	\begin{center}
			\includegraphics[scale=0.4]{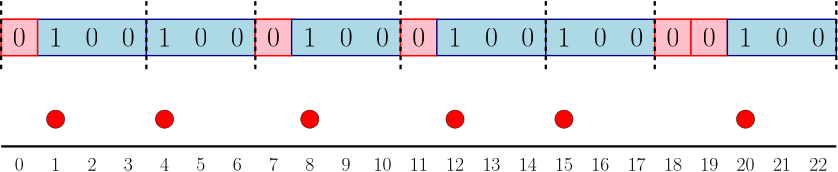}
	\end{center}
\caption{An example of the occupation configuration $ \mathbf{m} $ of a root partition $ \lambda_N^{(q)}(b) $ with $ q = 3 $ and $ N = 6 $ and $ b = (1, 1, 2 , 3, 3 , 5) $. The bottom line illustrates the orbitals and their occupation with red balls symbolizing particles. The top line is the equivalent tiling picture with monomer and void tiles on which the occupations are imprinted. 
The dashed lines mark separations of elementary blocks, which carry exactly one monomer and potentially frontal voids.}\label{fig:tiling}
\end{figure}

\begin{figure}[h]
	\begin{center}
		\includegraphics[scale=0.45]{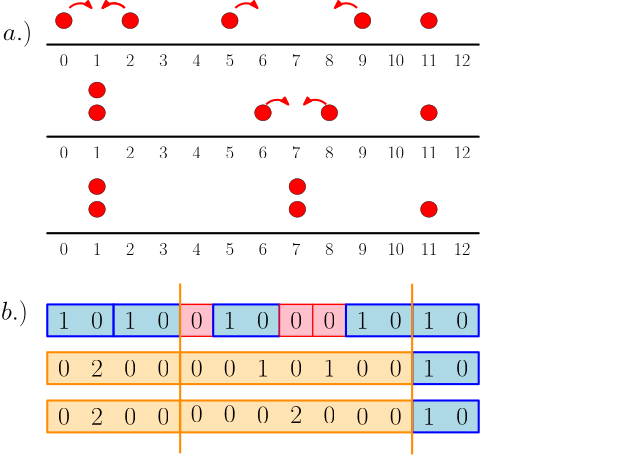}
	\end{center}
\caption{An example of twice two elementary squeezing moves on a root with $ q = 2 $ in $a.)$ the occupation picture with red balls symbolizing particles and $b.)$ in the tiling picture with occupation numbers printed on the tiles. The orange vertical lines mark segmentations, which are consistent with the squeezings. In each picture, the top line represents the root, and the bottom line represents the result after the two squeezing moves.}\label{fig:sqeezing}
\end{figure}
As will be explained in detail in Section~\ref{sec:partitions}, any partition dominated by a root partition, $\lambda \preceq \lambda^{(q)}_N(b) $, can be obtained from that root by a sequence of elementary squeezing operations. In the occupation picture, an elementary squeezing constitutes of two particles moving towards each other, thereby respecting their 'center of mass', see Figure~\ref{fig:sqeezing}. The fact that the expansions~\eqref{MPS with monomial symmetric} and those in Corollaries~\ref{cor:planar} and \ref{cor:cyl}, are supported on partitions dominated by a root partition is hence in accordance with what is known in the physics literature  \cite{Bernevig:2008aa,Bernevig:2009aa,thomale2011decomposition} for fractional quantum Hall wavefunctions.

Through the monomer-void tiling picture, any root partition $   \lambda^{(q)}_N(b) $ can be decomposed into $ N $ elementary blocks, each of which contains exactly one monomer and potentially frontal voids, see Figure~\ref{fig:tiling}. In terms of the composition operation on partitions, this amounts to writing the root partition as 
\begin{equation}\label{eq:segroot0}
\lambda_N^{(q)}(b) = \lambda_{1}^{(q)}\big(b_1\big) \cup  \lambda_{1}^{(q)}\big(b_2-b_1\big)  \cup  \dots  \cup \lambda_{1}^{(q)}\big(b_N- b_{N-1}\big) ,
\end{equation}
see  Definition~\ref{def:composition}. 

To harvest the factorization structure inherent in the expansion coefficients, we will consider arbitrary segmentations of these elementary blocks.
\begin{defn} 
Given a root $ (N,b ) $, 
we call $ (N_1, N_2 , \dots , N_r) \in \mathbb{N}^r $ with $ \sum_{j=1}^r N_j = N $ and 
$ r \in \{ 1, \dots , N \} $,  a \emph{segmentation of the root} and $ r $  the segmentation's length. For any segmentation of length $ r $, we denote by $ b^{(j)} = \big( b^{(j)}_1 , \dots , b^{(j)}_{N_1} \big)  $, $ j \in \{ 1, \dots , r \} $, the unique partitions such that 
\begin{align*} b = \big(b^{(1)} , \dots , b^{(r)} \big) ,
\end{align*}
where we define iteratively for two partitions $ b^{(1)} \in \mathbb{N}_0^{N_1} $, $ b^{(2)} \in \mathbb{N}_0^{N_2} $, their composition
\begin{equation}\label{eq:bconcat}
		\big(b^{(1)},b^{(2)}\big) \coloneqq \big( b^{(1)}_1, \dots ,  b^{(1)}_{N_1} , b^{(1)}_{N_1} +b^{(2)}_1 , \dots , b^{(1)}_{N_1} +b^{(2)}_{N_2} \big) . 
		\end{equation}
Finally, the segments corresponding to $ r = N $ and $ N_1 = \dots = N_N = 1 $ are referred to as \emph{elementary blocks}. 
\end{defn}
In the language of tilings, a segmentation of the root $ (N,b) $ of lengths $ r $ corresponds to marking $ r $ blocks on the tiling. Each segment consists of elementary blocks, which in turn are made up from the void tiles preceding a monomer (in case there are voids) and a monomer tile,  cf.\ Figures~\ref{fig:tiling} and~\ref{fig:irrseg}. In the language of partitions, a segmentation of length $ r $ corresponds to a (de)composition
\begin{equation}\label{eq:segroot}
\lambda_N^{(q)}(b) = \lambda_{N_1}^{(q)}\big(b^{(1)}\big) \cup  \lambda_{N_2}^{(q)}\big(b^{(2)}\big)  \cup  \dots  \cup \lambda_{N_r}^{(q)}\big(b^{(r)}\big) 
\end{equation}
cf.~Definition~\ref{def:composition}.
As is specified in Definition~\ref{def:irred}, any partition which is dominated by a given root partition,  $ \lambda  \preceq \lambda^{(q)}_N(b) $, can be uniquely decomposed into irreducible segments 
\begin{equation}\label{eq:decomppart}
 \lambda =   \lambda^{(1)} \cup  \dots \cup \lambda^{(r)}  , 
 \end{equation}
which are each dominated by the root partition corresponding to the respective segment, $ \lambda^{(j)}  \preceq \lambda_{N_j}^{(q)}\big(b^{(j)}\big) $, $ j \in \{1 , \dots , r \} $.  The segmentation instances $ (N_1, \dots, N_r) $ are then referred to as the renewal points of $ \lambda $, cf.~Definition~\ref{def:irred}. 
For any such partition, the occupation states are products by construction, 
\begin{align}\label{eq:productSlater}
\Phi^{(c)}_{ \lambda^{(1)} \cup  \dots \cup \lambda^{(r)} } 
 = \Phi^{(c)}_{ \lambda^{(1)}} \odot  \Phi^{(c)}_{ \lambda^{(2)}} \odot \cdots \odot  \Phi^{(c)}_{ \lambda^{(r)}} .
 \end{align}
In terms of the occupation numbers $ \mathbf{m}\left(  \lambda^{(1)} \cup  \dots \cup \lambda^{(r)}   \right) $ the operation $ \odot $ corresponds to a simple concatenation of the individual occupations $  \mathbf{m} \left(  \lambda^{(j)}\right) $ of the segments, cf.~Figure~\ref{fig:concatenation}. 
The (non-commutative) product corresponding to $ \lambda =   \lambda^{(1)} \cup  \lambda^{(2)}  $ with $  \lambda^{(j)}  \preceq \lambda_{N_j}^{(q)}\big(b^{(j)}\big) $ is defined by shifting the second segment to the end of the first, 
$$ T_{b^{(1)},N_1 }\lambda^{(2)}  \coloneqq \big(0, \dots , 0 , q (N_1-1) + |b^{(1)} | + \lambda_1^{(2)} , \dots , q (N_1-1) + |b^{(1)} |  +  \lambda_{N_2}^{(2)} \big) , 
$$
and taking the $ q $-symmetrized product, 
\begin{equation}\label{eq:Slaterfact}
  \Phi^{(c)}_{ \lambda^{(1)}} \odot  \Phi^{(c)}_{ \lambda^{(2)}}  \coloneqq \Phi^{(c)}_{ \lambda^{(1)}} \otimes_q \Phi^{(c)}_{ T_{b^{(1)},N_1 }\lambda^{(2)} } . 
\end{equation}
The definition of $ \odot $ on this basis extends by linearity. In this notation, our main factorization result for wavefunctions in the cylinder geometry reads as follows.
\begin{figure}
	\begin{center}
		\includegraphics[scale=0.4]{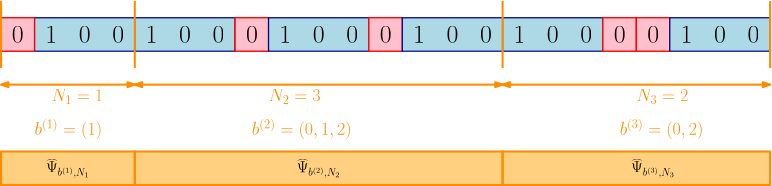}
	\end{center}
	\caption{ The root tiling from Figure~\ref{fig:tiling}  with a particular segmentation of length three with  $N_1=1, N_2=3, N_3=2$. The orange lines indicate the renewal points, and the orange blocks in the last line correspond to the term with $ r = 3 $ and $N_1=1, N_2=3, N_3=2$ in the representation~\eqref{eq:fact}  with irreducible factors indicated on the blocks.}\label{fig:irrseg}
\end{figure}
\begin{theorem}\label{thm:fact}
Given a root $ (N,b) $, the functions~\eqref{def:WFc} have the representation
\begin{equation}\label{eq:fact}
 \Psi_{b,N}  = \sum_{r=1}^N \sum_{ \substack{N_1, \dots , N_r \in \mathbb{N}\\  \sum_{j=1}^r N_j = N } }  \widehat \Psi_{b^{(1)},N_1} \odot  \widehat \Psi_{b^{(2)},N_2} \odot \cdots \odot  \widehat \Psi_{b^{(r)},N_r} ,
\end{equation}
Here $  \big(b^{(1)} , \dots , b^{(r)} \big)  $ is the decomposition of $ b $ corresponding to the segmentation $ (N_1, \dots , N_r )$, and the term for fixed $ N_j \in \mathbb{N} $ and $ b^{(j)} $ is
\begin{equation}\label{eq:irredexp}
	 \widehat \Psi_{b^{(j)},N_j} \coloneqq \mkern-10mu \sum_{\substack{ \lambda  \preceq \lambda^{(q)}_{N_j}(b^{(j)}) \\ \lambda \; \textrm{irreducible}} } h_{b^{(j)} }(\lambda) \ \Phi^{(c)}_\lambda ,
\end{equation}
where the sum is over the irreducible partitions dominated by $ \lambda^{(q)}_{N_j}(b^{(j)}) $  only  (cf.\ Definition~\ref{def:irred} for the notion of irreducibility). Moreover, the vectors in the right side of~\eqref{eq:fact} are orthogonal, and the norm-square of an irreducible contribution is estimated by 
\begin{equation}\label{eq:normseg}
\big\|  \widehat \Psi_{b,N} \big\|^2 \leq  \exp\left( - \left(q(N-1)+b_N-b_1\right) C_q(\gamma) \right)  \leq \exp\left( - q (N-1) C_q(\gamma) \right) ,
\end{equation}
where 
\begin{equation}\label{def:Cgamma}
C_q(\gamma)  \coloneqq  (\gamma^2 - 2(4q+1)) -2  q^{-1} \ln  (1- e^{-(4q+1)} )  -  \pi \sqrt{\frac{2}{3}}  
\end{equation}
is strictly positive for sufficiently large $ \gamma $.
\end{theorem}
The proof is found in Subsection~\ref{sec:factor}. 
\begin{rem}
The fractional quantum Hall wavefunctions also exhibit a renewal structure in other 2D geometries, since the geometry-independent contribution $w_b $ to the coefficients $ h_b = g_b w_b $ determines the irreducible segments in the factorization~\eqref{eq:fact}. However, the geometry-dependent contribution $ g_b $ is more complicated under compositions. E.g.\ in the plane, $ g_b(\lambda^{(1)} \cup \lambda^{(2)} ) $ for $ \lambda^{(j)} \preceq \lambda^{(q)}_N(b^{(j)}) $, $ j = 1,2 $, would factor into orbital-dependent terms.  From an analytical point of view, most important is that, e.g., in the planar situation, we do not have a small parameter $ \gamma^{-1} $, which enables the estimate~\eqref{eq:normseg}
\end{rem}

Repeatedly, we will take advantage of the fact that the wavefunctions $ 	 \widehat \Psi_{b^{(j)},N_j}  $ on each segment $ j $ have a non-fluctuating particle number $ N_j $ and orbital momentum associated with $ \lambda^{(q)}_{N_j}(b^{(j)})  $.

The decomposition~\eqref{eq:decomppart} of a partition $ \lambda $ into segments dominated by root segments, $  \lambda^{(j)}  \preceq \lambda^{(q)}_{N_j}(b^{(j)}) $, ensures that the squeezing operations, by which one obtains $ \lambda $ from the root partition $ \lambda_N^{(q)}(b) $, are restricted to the individual segments. In other words, no squeezing occurred in between different segments. The representation~\eqref{eq:fact} can hence be interpreted as an expansion along the partial order of squeezing operators restricted to all possible segmentations, cf.~Figure~\ref{fig:sqeezing} and~\ref{fig:irrseg}. It is non-trivial to realize that~\eqref{eq:fact}  is an orthogonal expansion. \\

Due to orthogonality, the contribution of the irreducible segments to the total norm-square is
\begin{equation}\label{eq:normPsi}
\left\|  \Psi_{b,N}  \right\|^2 = \sum_{r=1}^N \sum_{ \substack{N_1, \dots , N_r \in \mathbb{N}\\  \sum_{j=1}^r N_j = N } } \prod_{j=1}^r  \left\|   \widehat \Psi_{b^{(j)},N_j} \right\|^2 .
\end{equation}
For sufficiently large $\gamma $, i.e.\ on thin cylinders, the norm bound~\eqref {eq:normseg} is exponential in the length of the root partition $ \lambda_N^{(q)}(b) $ as defined by the number of orbitals its tiles cover.  
This should be compared with the fact that irreducible segments with only one occupancy carry unit weight,  $ \|  \widehat \Psi_{b,1}  \|^2 = |h_b(\lambda_1^{(q)}(b)) |^2 = 1 $ irrespective of the number $ b \in \mathbb{N}_0 $ of voids by Corollary~\ref{cor:cyl}. 
The renewal process governing the expansion~\eqref{eq:fact} hence exponentially favors short segments. This observation is the key behind our results on the entanglement gap and exponential clustering on thin cylinders.

An immediate consequence of~\eqref{eq:normPsi}  is the super-multiplicativity of the square of the norm, i.e.\ for any  $ N_1 , N_2 \in \mathbb{N} $ and $ b^{(1)} \in \mathbb{N}_0^{N_1} $, $  b^{(2)} \in \mathbb{N}_0^{N_2} $ with $ N = N_1 + N_2 $ and $ b= \big(b^{(1)}, b^{(2)} \big) $, we have:
\begin{equation}\label{eq:supermult}
\left\|  \Psi_{b,N}  \right\|^2 \geq \left\|  \Psi_{b^{(1)},N_1}  \right\|^2  \left\|  \Psi_{b^{(2)},N_2}  \right\|^2 \geq  \left\|  \Psi_{b^{(1)},N_1}  \right\|^2 .
\end{equation}
Here the second equality follows from the fact that $ \left\|  \Psi_{b,1}  \right\|^2 = 1 $ and hence $ \left\|  \Psi_{b,N}  \right\|^2  \geq 1 $. \\

More can be said for the norm-squared~\eqref{eq:normPsi}  of the Laughlin wavefunction $ b= 0 $. As is well known~\cite{di1994laughlin,jansen2009symmetry}, the quantity $ C_N \coloneqq \left\|  \Psi_{0,N}  \right\|^2 $ can be interpreted as the partition function of a Coulomb gas, and
$
\lim_{N\to \infty}N^{-1} \ln C_N
$ is the corresponding pressure, which exists thanks to super-multiplicativity and Fekete's lemma. The finite-$ N $ corrections to this infinite-volume pressure are expected to be suppressed. Since the cylinder is essentially a one-dimensional system, this goes along with the breaking of the translational symmetry, first established in~\cite{jansen2009symmetry} using renewal theory, and later in~\cite{Aizenman:2010aa} with the help of the Aizenman-Martin argument. In Appendix~\ref{app:new}, we show how our estimate \eqref{eq:normseg} can be used to establish quantitative versions of finite-$N$  corrections to the pressure. 
\subsection{Entanglement gap and exponential clustering on thin cylinder}

In the cylinder geometry, it is natural to investigate the entanglement properties of the wavefunctions $  \Psi_{b,N}  $ concerning a left-right bipartition of the orbitals. In view of Theorem~\ref{thm:fact}, 
it is not surprising that these wavefunctions approximately factorize with respect to bipartitions of their root partitions. This is expressed in the following 
\begin{theorem}\label{thm:EGap}
Let $ N_1 , N_2 \in \mathbb{N} $ and $ b^{(1)} \in \mathbb{N}_0^{N_1} $, $  b^{(2)} \in \mathbb{N}_0^{N_2} $ be partitions. If $ N = N_1 + N_2 $ and $ b= \big(b^{(1)}, b^{(2)} \big) $, then in case  $ \gamma $ is large enough such that the constant from~\eqref{def:Cgamma} is strictly positive, $ C_q(\gamma) > 0 $,   the functions~\eqref{def:WFc} approximately factorize: 
$$
\left\| \frac{\Psi_{b,N}}{\left\|\Psi_{b,N}\right\|}- \frac{ \Psi_{b^{(1)},N_1}  }{\left\| \Psi_{b^{(1)},N_1} \right\|}\odot  \frac{\Psi_{b^{(2)},N_2} }{\left\| \Psi_{b^{(2)},N_2} \right\|} \right\| \leq \frac{2 e^{-C_q(\gamma)q/2}}{1-e^{-C_q(\gamma)q}}  .
$$  
\end{theorem}
The proof is found in Subsection~\ref{sec:Egap}.

Theorem~\ref{thm:EGap} implies that the largest eigenvalue of the left and right reduced state of the (normalized) state associated with $ \Psi_{b,N}  $  is larger or equal to $ 1 -  4\sqrt{2}/ \sqrt{e^{C_q(\gamma)\ q } - 1}   $. This follows from standard estimates, e.g.\ using Uhlmann's bound~\cite{Uhlmann:1970aa}, see e.g.~\cite[Eq. (3.5)]{Aizenman:2025aa}.  In the thin cylinder limit, the largest eigenvalue is hence close to its maximal value of one. Since all eigenvalues sum to one, the other eigenvalues must be small.  This explains the common notion of the 'entanglement gap' in the so-called entanglement spectrum, i.e.\ the negative logarithm of the eigenvalues of the reduced state of such a bipartition.

In the physics literature, the entanglement spectrum of the Laughlin state is well studied numerically~\cite{li2008entanglement,Thomale:2010aa,Lauchli:2010aa}. Li and Haldane have proposed \cite{li2008entanglement} that the part of the entanglement spectrum, which consists of the small eigenvalues, exhibits universal features predicted by the CFT. Theorem~\ref{thm:EGap} does not address this conjecture. \\

Our last result concerns the exponential decay of correlations of local observables in each of the states $ \Psi_{b,N} $. Any local observable is composed of the sums of products of 
the canonical annihilation and creation operators, $ c_k $, $ c_k^* $, with $ k \in \mathbb{N}_0 $, associated with the one-particle orbitals~\eqref{eq:cylinderONB} on the non-negative half-cylinder. They obey the canonical anticommutation or commutation rules,
$$
[ c_k , c^*_l ]_q \coloneqq c_k c_l^* - (-1)^q c^*_l c_ k = \delta_{k,l} , \quad [c_k, c_l ]_q =  [c_k^*, c_l^* ]_q = 0 ,
$$
where $ [ \cdot, \cdot ]_q $ stands for the anti-commutator or commutator, depending on whether we deal with fermions ($q$ odd) or bosons ($q$ even). Subsequently, we denote for a partition $ L = (L_1, L_2 , \dots , L_r) \in \mathbb{N}_0^r $,
$$
c_L \coloneqq  c_{L_1} c_{L_2} \dots c_{L_r} 
$$
In this notation, the algebra of local observables is generated by such products and their adjoints, and we abbreviate their building blocks as
$$
\mathcal{O}_\loc \coloneqq \left\{ c^*_L c_{L'} \ : \  L \in \mathbb{N}_0^r , \,  L' \in \mathbb{N}_0^s  \; \mbox{for some $ r,s \in \mathbb{N}$} \right\} . 
$$
We will refer to the orbitals in the partitions $ L, L' $ as the support, $ \supp O $, of the observable  $O = c^*_L c_{L'}$.
In the case of fermions ($q $ odd), it is straightforward that the expectation values  
$$  \langle O \rangle_{N,b} \coloneqq \frac{\langle \Psi_{b,N} ,O \   \Psi_{b,N}\rangle}{\left\| \Psi_{b,N}\right\|^2}  $$
are well defined for any $ O \in \mathcal{O}_\loc $ and any root $ (N,b) $. Since the occupation number of bosons is a priori unbounded on any orbital, already the finiteness of these expectation values, which we will establish in Lemma~\ref{lm:norms observables}, is non-trivial for bosons ($q $ even).
Our next theorem shows that, on sufficiently thin cylinders, any such expectation value exhibits exponential decay of the correlations of local observables in the distance, $ \dist(\supp A,\supp B)  \coloneqq \min \{ |x-y| \ | \ x \in \supp A , \;  y \in \supp B \} $, of their supports.
\begin{theorem} \label{thm:clustering}
For all $\gamma > 0 $ large enough such that the constant from~\eqref{def:Cgamma} is strictly positive, $ C_q(\gamma) > \ln(4) $, there is some $ d_q(\gamma) > 0 $ such that for all local observables $ A, B \in \mathcal{O}_\loc $ with $\max \supp A < \min \supp B$ and $ \dist(\supp A,\supp B) \geq d_q(\gamma) $, and   all roots $ (N,b) $:	
			\begin{equation} \label{exponential clustering}
				\vert \langle A B\rangle_{N,b} - \langle A  \rangle_{N,b} \langle B  \rangle_{N,b} \vert
				\leq C \exp\left(- \frac{C_q(\gamma)-\ln(4)}{6} \ \dist(\supp A,\supp B) \right) . 
			\end{equation}
			with a constant $ C < \infty $, which only depends on the cardinalities of the support of $ A $ and $ B $.
		\end{theorem}
The proof and explicit values for the constants $ c_q(\gamma)  , C $ are found in Subsection~\ref{sec:cluster}. 

Exponential clustering for the Laughlin state $ b = 0 $ has been shown in~\cite{jansen2009symmetry}. However, their proof method, which is based on the implicit function theorem and renewal theory, yields no explicit threshold for the inverse radius $\gamma$, at which exponential clustering applies. The follow-up~\cite{jansen2012fermionic} to \cite{jansen2009symmetry} uses Coulomb-gas methods for an attempt to improve estimates, however, without reaching a proof of quantitative clustering.

\subsection{Outlook}  
The set of fractional quantum Hall wavefunctions $ \Psi_{b,N} $ with arbitrary root $(N,b) $ is known to span a dense subspace of the kernel of the Haldane pseudopotential corresponding to filling fraction $ 1/q $~\cite{Rezayi:1994aa,Mazaheri:2015aa}. Haldane pseudopotentials are many-body operators modelling quantum Hall physics in rather general 2D geometries~\cite{PhysRevLett.51.605,QHEOxford2003,Ortiz:2013aa,Bandyopadhyay:2020aa,schossler2022inner}.  They emerge from a generic interacting Landau Hamiltonian in the scaling limit of short-range interactions~\cite{seiringer:2020}.  
For filling fraction $ 1/q $,  the Haldane pair potentials are frustration-free and non-negative, and hence 'parent Hamiltonians' for $  \Psi_{b,N} $.  

Pseudo-potentials are conjectured to have a uniform spectral gap above their ground-state -- a feature, which is responsible for the incompressibility of the quantum fluid \cite{Roug19, nachtergaele2021spectral} as well as the quantization of the Hall conductance~\cite{A7_hastings:2015, A7_Bachmann:2018lb, A7_Bachmann:2021dp}. Proving this so-called Haldane-gap conjecture has remained elusive. Rigorous results 
so far address only truncated versions of pseudopotentials~\cite{nachtergaele2021spectral,warze2022spectral,warzel2023bulk,A7_Y24}. 
Notably, the identification of the root $ (N,b) $ with a void-monomer tiling (cf.~Figure~\ref{fig:tiling}) proves that the set of  fractional quantum Hall wavefunctions $ \Psi_{b,N} $ 
 is in one-to-one correspondence with the set of VMD states introduced in~\cite{nachtergaele2021spectral,nachtergaele2020low}. The latter are an orthonormal basis of the kernel of the truncated Haldane Hamiltonian~\cite{Bergholtz:2005pl,nachtergaele2021spectral}, and the analogue of Theorem~\ref{thm:clustering} for the VMD states has also been established in~\cite{nachtergaele2021spectral}.   In contrast to VMD states, however, the wavefunctions $ \Psi_{b,N} $ corresponding to the same particle number $ N $ and total momentum $ |  \lambda^{(q)}_N(b)  | = \sum_{j=1}^N \big( \lambda^{(q)}_N(b)\big)_{j }  $, but yet different $ b $, are not necessarily orthogonal. 
 The correspondence raises the question and hope of an adiabatic connection of truncated and untruncated pseudopotentials. If this could be achieved adiabatically without closing the gap via a stability analysis \cite{BHM10,A7_NSY20}, Haldane's gap-conjecture would be proven. Other known methods for many-body spectral gaps~\cite{Young25}, such as the martingale method, Knabe's method, or the induction method~\cite{LNWY25}, also rely on a good understanding of the structure of the ground state.  With this work and the techniques developed in its body, we hope to provide a first non-trivial step -- at least in the cylinder geometry.

\section{Infinite MPS representation}\label{sec:iMPS}

	A virtual Hilbert space and a set of operators characterize MPS representations. This section sets the stage by defining these fundamental quantities and by exploring their basic properties. 
	Our iMPS representation adopts the framework of the free chiral conformal field theory (CFT) for quantum Hall states, as described in the physics literature. It is inspired, yet slightly different from predecessors~\cite{estienne2013fractional}. The relation is explained in Appendix~\ref{app:B}.

	\subsection{Virtual Hilbert space and momentum eigenspaces in a chiral CFT}
	  For the virtual Hilbert space, we take  the Fock space $\mathcal{H}$ generated from the vacuum~$ \vert 0 \rangle$ and a set of the bosonic annihilation and creation operators $(a_j)_{j\in \mathbb{N}}, (a_j^*)_{j\in \mathbb{N}}$, which satisfy the commutation relations
		\begin{equation} \label{commutation relation}
				[a_n, a_m] =0 = [a_n^*, a_m^*], \qquad [a_n, a_m^*] = n \delta_{n,m}, \qquad \text{for all } n,m \in \mathbb{N}.
		\end{equation}
		We will denote by $\mathcal{H}_0$ the subspace
		\begin{equation}
			\mathcal{H}_0 = \text{span} \left\{ \left(\prod_{j=1}^M (a_j^*)^{n_j} \right) \vert 0 \rangle \ : \ M\in \mathbb{N}_0, n_1, \dots, n_M\in \mathbb{N}_0 \right\}, 
		\end{equation}
		which is dense in $\mathcal{H}$ by construction.  
		The scalar product on $\mathcal{H}$ is chosen in such a way that $a_j^*$ is the adjoint of $a_j$ and such that $\vert 0\rangle$ is normalized. In accordance with the interpretation of conformal field theory, we will subsequently refer to  $a_j^*$ as the creation of moment $ j \in \mathbb{N} $, and to the operator 
		\begin{equation}\label{def:L0}
			L_0 = \sum_{n=1}^\infty a_n^* a_n , 
		\end{equation}
		which is  non-negative and symmetric on $\mathcal{H}_0$, 
		as the operator of total momentum. By Friedrich's extension theorem, it hence uniquely extends to a self-adjoint operator on $ \mathcal{H} $. 
		A straightforward computation using~\eqref{commutation relation} shows that 
		\begin{equation} \label{commutation relation L0 aj}
				L_0 a_j = a_j (L_0-j), \qquad a_j^* L_0 = (L_0-j) a_j^*
			\end{equation}
		for all $ j \in \mathbb{N} $, and hence
		\begin{equation}\label{eq:L_0evs}
		 L_0\left(\prod_{j=1}^K (a_j^*)^{n_j} \right) \vert 0 \rangle = \sum_{j=1}^K j n_j  \left(\prod_{j=1}^K (a_j^*)^{n_j} \right) \vert 0 \rangle .
		\end{equation}
		The finite-dimensional subspaces 
		 \begin{equation}\label{def:momentumH}
				\mathcal{H}(M) \coloneqq \mathrm{span} \left\{ \left(\prod_{j=1}^K (a_j^*)^{n_j}\right) \vert 0 \rangle \ : \ K\in \mathbb{N}, n_1, \dots, n_K\in \mathbb{N}_0 \text{ such that } \sum_{j=1}^K j n_j =M \right\}
			\end{equation}
			are the eigenspaces of $ L_0 $ corresponding to 
			 a fixed total momentum $ M \in \mathbb{N}_0 $. They are hence mutually orthogonal, i.e.,  $\mathcal{H}(M)\perp \mathcal{H}(\tilde{M}) $ for $   M\neq \tilde{M}  $, and span the full space:
			\begin{equation}\label{eq:orthsum}
			  \mathcal{H} = \bigoplus_{M\in \mathbb{N}_0}  \mathcal{H}(M) . 
			\end{equation}

		\subsection{Motivation and definitions of operators}	
		To introduce the set of operators used in the iMPS representation on the virtual Hilbert space  $\mathcal{H}$, we first recall the construction used in~\cite{estienne2013fractional}, which 
		 considers formal power series with coefficients in $\mathrm{End}(\mathcal{H}_0)$, i.e., the coefficients are linear operators from $\mathcal{H}_0$ into itself -- most importantly		
		\begin{equation} \label{def S+ S-}
				S_+(z) = \sum_{n=1}^\infty \frac{1}{n} a_n z^{-n}, \qquad S_-(z) = \sum_{n=1}^\infty \frac{-1}{n} a_n^* z^{n}.
		\end{equation}
		Informally, one would like to define
		\begin{equation}\label{eq:informalW}
			W(z) = \sum_{n\in \mathbb{Z}} W_n z^n 
			=\exp\left(-\sqrt{q} S_-(z)\right) \exp\left(-\sqrt{q} S_+(z)\right).
		\end{equation}
		For this, one would consider Laurent series with coefficients in $\mathrm{End}(\mathcal{H}_0)$. Unfortunately, these do not form a ring, as multiplication is, in general, not defined. Indeed, formally, the coefficients of a product of two Laurent series will be a series of terms in $\mathrm{End}(\mathcal{H}_0)$, which has no meaning in an algebraic setting. These objects are also problematic from an analytic perspective, since it is unclear in what sense the series of unbounded operators should converge. This is a standard problem in algebraic quantum field theory, and a solution involves taking the formal product and realizing that the infinite series becomes a finite sum when tested against an element in a suitable domain. In our case, $\mathcal{H}_0$ will do the trick.
		
			Before we spell out a rigorous version, let us briefly recall the relevance of the above operators in relation to the Laughlin wavefunction if $ N = 2$. 
		Informally, we compute the quantity
		$$
			\langle 0\vert W(z_1) W(z_2) \vert 0 \rangle =  \langle 0 \vert \exp\left(-\sqrt{q} S_-(z_1)\right) \exp\left(-\sqrt{q} S_+(z_1)\right) \exp\left(-\sqrt{q} S_-(z_2)\right) \exp\left(-\sqrt{q} S_+(z_2)\right)\vert 0\rangle.
		$$
		As $a_j \vert 0\rangle =0$ for all $j\in \mathbb{N}$, one concludes  
		\begin{align*}
				\exp(-\sqrt{q} S_+(z)) \vert 0\rangle 
				= \vert 0 \rangle, \qquad \exp(-\sqrt{q} S_-(z))^*  \vert 0 \rangle = \vert 0\rangle. 
		\end{align*}
		Thus, it remains to evaluate
		\begin{align*}
			\langle 0 \vert  \exp\left(-\sqrt{q} S_+(z_1)\right) \exp\left(-\sqrt{q} S_-(z_2)\right) \vert 0\rangle. 
		\end{align*}
		To do so, we will commute $\exp\left(-\sqrt{q} S_+(z_1)\right)$ to the right by computing the commutator of $\exp(-\sqrt{q} S_+(z_1))$ and $\exp(-\sqrt{q} S_-(z_2))$. This is a setting where we can use the Baker-Campbell-Hausdorff theorem (by embedding the problem into $\mathrm{End}(\mathcal{H}_0)[[z_1^{-1}, z_2]]$, which is an associative $\mathbb{Q}$-algebra, see~\cite[Corollary 4.5]{bonfiglioli2011topics}). Since all higher commutators vanish as is seen from
		\begin{align*}
			[S_+(z_1),S_-(z_2)] &= \sum_{n=1}^\infty \sum_{m=1}^{\infty} \frac{-1}{nm} [a_n, a_m^*] z_1^{-n} z_2^{m}
			= \sum_{n=1}^\infty \frac{-1}{n} z_1^{-n} z_2^n , 
		\end{align*}
		the Baker-Campbell-Hausdorff  formula takes the form
		\begin{align*}
			\exp &(-\sqrt{q} S_+(z_1)) \exp(-\sqrt{q} S_-(z_2)) \\
			&=\exp(-\sqrt{q}S_-(z_2)) \exp(-\sqrt{q} S_+(z_1)) \exp(q[S_+(z_1), S_-(z_2)]).
		\end{align*}
		The last factor can be simplified by using the Taylor series
		$	\log(1-z) = - \sum_{n=1}^\infty \frac{z^n}{n}  $, which yields  
		\begin{align*}
			\exp\left(- q \sum_{n=1}^\infty \frac{z^n}{n}  \right) = \exp(q \log(1-z)) = (1-z)^q.
		\end{align*}
		Comparing coefficients (respectively, formally applying this with $z=z_2/z_1$), we obtain
		\begin{align*}
			\exp(q [S_+(z_1), S_-(z_2)]) = (1-z_1^{-1} z_2)^q = z_1^{-q} (z_1-z_2)^q,
		\end{align*}
		and therefore
		\begin{equation} \label{BCH}
			\exp(-\sqrt{q} S_+(z_1)) \exp(-\sqrt{q} S_-(z_2)) =  \left( 1-\frac{z_2}{z_1} \right)^q  \exp(-\sqrt{q}S_-(z_2)) \exp(-\sqrt{q} S_+(z_1)).
		\end{equation}
		Formally, one arrives at 
		\begin{align*}
			\sum_{k_1, k_2\in \mathbb{Z}} \langle 0\vert W_{k_1} W_{k_2} \vert 0 \rangle \ z_1^{k_1+q} z_2^{k_2} =  (z_1-z_2)^q ,
		\end{align*}
		which is the $ q $th power of the Vandermonde for $ N = 2 $. \\

		The only non-rigorous step in the above computation was the separation of $\exp(-\sqrt{q}S_-(z))$ and $\exp(-\sqrt{q} S_+(z))$ to simplify. We cannot do this, as they do not belong to any well-behaved ring where both exponentials are individually defined. The way to avoid this issue is to study in more detail how the creation operators of $\exp(-\sqrt{q} S_-(z))$ get coupled to the annihilation operators in $\exp(-\sqrt{q}S_+(z))$ when the two power series get multiplied. To do so, we expand both as power series, 
	\begin{align}\label{eq:relSD}
			 \exp(-\sqrt{q} S_-(z))  =  
			\sum_{k=0}^\infty  \frac{(-\sqrt{q})^k}{k!} \left( -\sum_{n=1}^\infty \frac{a_{n}^*}{n} z^n \right)^k 
			= \sum_{\ell=0}^\infty z^\ell D_\ell^- , & \\
			 \mbox{and similarly}\quad \exp(-\sqrt{q} S_+(z))  =  \sum_{\ell=0}^\infty z^{-\ell} D_\ell^+  & \notag
		\end{align}
		with coefficients in $\mathrm{End}(\mathcal{H}_0)$, which are defined the following.
		
		\begin{defn}\label{def:DW}
		For $ \ell \in \mathbb{Z} $, we define the operators $ D_\ell^\pm : \mathcal{H}_0 \to \mathcal{H}_0  $ by
			\begin{align}\label{def:Dell}
			\begin{split}
				D_\ell^+ & \coloneqq \begin{cases}
					\mathbbm{1} ,& \ell=0, \\ 
					\displaystyle\sum_{k=1}^\ell \frac{\sqrt{q}^k}{k!} \!\! \sum_{\substack{j_1, \dots, j_k >0 : \\ j_1+\dots+j_k=\ell}} (-1)^k \left( \prod_{s=1}^k \frac{1}{j_s} \right) \prod_{t=1}^k a_{j_t},& \ell>0, \\
					0,& \ell<0
				\end{cases}  \\
				D_\ell^- & \coloneqq \begin{cases}
					\mathbbm{1} ,& \ell=0, \\ 
					\displaystyle\sum_{k=1}^\ell \frac{\sqrt{q}^k}{k!} \sum_{\substack{j_1, \dots, j_k >0 : \\ j_1+\dots+j_k=\ell}} \left( \prod_{s=1}^k \frac{1}{j_s} \right) \ \prod_{t=1}^k a_{j_t}^*,\qquad & \ell>0, \\
					0,& \ell<0.
				\end{cases}
					\end{split}
			\end{align}	
			Moreover, for any  $m\in \mathbb{Z}$  we set $  W_m :  \mathcal{H}_0 \to \mathcal{H}_0  $ as
			\begin{equation} \label{definition Wm}
				W_m \coloneqq \sum_{\ell\in \mathbb{Z}} D_{m+\ell}^- D_\ell^+ . 
			\end{equation}	
			\end{defn}
			The operator $ D_\ell^+  $ annihilates a total momentum $ \ell $ and $ D_\ell^-  $ creates it as long as $ \ell \geq 0 $. They hence act as ladders among the family of subspaces~\eqref{def:momentumH}. In particular, the above operators are indeed well defined on the dense subspace $ \mathcal{H}_0 $, which they leave invariant.  As $D_\ell^+$ vanishes on $\mathcal{H}(k)$ for $k<\ell$, we conclude that
	\begin{align*}
		W_m\vert_{\mathcal{H}(k)} : \mathcal{H}(k) \rightarrow \mathcal{H}(k+m), v \mapsto \sum_{\ell=1}^k D_{m+\ell}^- D_\ell^+ v.
	\end{align*}
	Thus, when tested against an element in $\mathcal{H}_0$, each operator $W_m$ reduces to a finite sum, which also makes sense in an algebraic setting. As we are effectively dealing with finite sums,  we subsequently never have to worry about convergence issues when taking commutators. 
			\begin{lemma} \label{prop:DW} With  the convention $\mathcal{H}(M) = \{ 0 \} $ in case $  M < 0 $:
			\begin{enumerate}
			\item 
			$ D_\ell^\pm \mathcal{H}(M)\subseteq \mathcal{H}(M\mp \ell)$ for any $ \ell  \in \mathbb{Z} $,  $ M \in \mathbb{N}_0 $. Moreover:
			\begin{enumerate}
			\item \label{CR D} 	for all $ m,k \in \mathbb{Z} $:	\quad $ \displaystyle
				D_m^+ D_k^-  = \sum_{\ell=0}^q (-1)^\ell \binom{q}{\ell} D_{k-\ell}^- D_{m-\ell}^+ $,
			\item \label{CR1b} for all $ n \in \mathbb{N} $, $ \ell \in \mathbb{Z} $: \quad $ \displaystyle a_n D_\ell^-  = D_\ell^- a_n + \sqrt{q}\  D_{\ell-n}^- $. 
			\end{enumerate}
			 \item $W_k \mathcal{H}(M)\subseteq \mathcal{H}(M+k)$ for any $ k \in \mathbb{Z} $, $ M \in \mathbb{N}_0 $. Moreover: 
			\begin{enumerate}
			\item \label{CR W} 	for all $ m,k \in \mathbb{Z} $: \quad $ \displaystyle
			 W_m W_k   = (-1)^q \ W_{k-q} W_{m+q} $,
			 \item for all $ n \in \mathbb{N} $, $ m \in \mathbb{Z} $:
			\quad $ \displaystyle  a_n W_m  = W_m a_n + \sqrt{q} \ W_{m-n} $. 
			 	 \end{enumerate}
			 \end{enumerate}
			\end{lemma}
			\begin{proof}
			The asserted inclusions of the images of $ \mathcal{H}(M) $ under the various operators are immediate from the definition. 
			It remains to prove the items. 
			
			The proof of~\ref{CR D} is based on~\eqref{eq:relSD} and  the comparison of the coefficient of  $z_1^{-m} z_2^k$ in the 
			expansion of the  Baker\--Campbell\--Hausdorff formula~\eqref{BCH}. It is applicable \cite[Corollary 4.5]{bonfiglioli2011topics} as an identity in the associative, unital $\mathbb{Q}$-algebra $\mathrm{End}(H_0)[[z_1^{-1}, z_2]]$. 
			In the same spirit,  $a_n D_{\ell}^{-}$ is the coefficient of $z^{\ell}$ of 
			\begin{align*}
				a_n \sum_{s=0}^{\ell+1} \frac{1}{s!} \left( -\sqrt{q} S_-(z) \right)^s
				= \sum_{s=0}^{\ell+1} \frac{1}{s!} \left( -\sqrt{q} S_-(z) \right)^s a_n + \sum_{s=0}^{\ell+1} \frac{1}{s!} [a_n,\left( -\sqrt{q} S_-(z) \right)^s].
			\end{align*}
			Since
			$
				[a_n,-\sqrt{q} S_-(z)] = -\sqrt{q}\sum_{t=1}^\infty \frac{-1}{t} [a_n, a_t^*] z^{t}
				= \sqrt{q} z^n $, 
			is in the center of the algebra $\mathrm{End}(\mathcal{H}_0)[[z]]$, we conclude for all $s\in \mathbb{N}$
			\begin{align*}
				[a_n, (-\sqrt{q} S_-(z))^s] = s [a_n, -\sqrt{q} S_-(z)] (-\sqrt{q} S_-(z))^{s-1} = s \sqrt{q} z^n (-\sqrt{q} S_-(z))^{s-1}.
			\end{align*}
			Hence, we arrive at
			\begin{align*}
				a_n \sum_{s=0}^{\ell+1} \frac{1}{s!} \left( -\sqrt{q} S_-(z) \right)^s
				&= \sum_{s=0}^{\ell+1} \frac{1}{s!} \left( -\sqrt{q} S_-(z) \right)^s a_n + \sum_{s=1}^{\ell+1} \frac{1}{s!} s \sqrt{q} z^n (-\sqrt{q} S_-(z))^{s-1} \\
				&= \sum_{s=0}^{\ell+1} \frac{1}{s!} \left( -\sqrt{q} S_-(z) \right)^s a_n + \sqrt{q} z^n \sum_{s=0}^{\ell} \frac{1}{s!}  (-\sqrt{q} S_-(z))^{s}.
			\end{align*}
			Comparing the coefficient of $z^\ell$ in the last equation, we obtain the second item,~\ref{CR1b}.

			For a proof of~\eqref{CR W}, we use the definition~\eqref{definition Wm}, the commutation rule~\ref{CR D} and the fact that the $D^+ $s commute and so do the $ D^-  $s:			\begin{align*}
			W_m W_k  & = \sum_{\ell , \ell' \in \mathbb{Z} }  \sum_{n=0}^q (-1)^n \binom{q}{n}   D_{\ell'+k-n}^- D_{\ell+m}^-  D_{\ell'}^+ D_{\ell-n}^+  \\
			& =\sum_{\ell , \ell' \in \mathbb{Z} }   \sum_{n=0}^q (-1)^{n+q} \binom{q}{n}   D_{\ell'+k}^- D_{\ell+m -n+q }^-  D_{\ell'-n+q}^+ D_{\ell}^+ \\
			& = \sum_{\ell , \ell' \in \mathbb{Z} }  D_{\ell'+k}^- (-1)^q D_{\ell'+q}^+D_{\ell+m +q }^- D_{\ell}^+ = (-1)^q \ W_{k-q} W_{m+q} . 
			\end{align*}
			The second line resulted from the substitutions $\tilde{n}=q-n, \tilde{\ell}=\ell-q+n, \widetilde{\ell'}=\ell'-n$. The third line is based on another index shift and~\ref{CR D} again.
			To prove the last item, we note that 
			$D_\ell^+$ consists of annihilation operators,  such that $a_n$ commutes with $D_\ell^+$. We may thus compute
			\begin{equation*} 
					a_n W_m = a_n \sum_{k\in \mathbb{Z}} D_{m+k}^- D_k^+ 
					= \sum_{k\in \mathbb{Z}} (D_{m+k}^- a_n + \sqrt{q} D_{m-n+k}^-) D_k^+ \\
					= W_m a_n + \sqrt{q} W_{m-n}.
			\end{equation*}
			This completes the proof.
			\end{proof}

\subsection{Polynomial expansion of powers of the Vandermonde}

		The following is the core for the MPS representation of the coefficients of any integer power of the Vandermonde determinant as a complex multinomial. 
		\begin{lemma} \label{thm:MPS Laughlin}
			For any $ q, N  \in \mathbb{N} $ one has
			\begin{equation} \label{Laughlin with W}
				\prod_{1\leq i < j \leq N} (z_i-z_j)^q = \sum_{k_1, \dots, k_N \in \mathbb{Z}} \langle 0 \vert W_{k_1} \dots W_{k_N} \vert 0 \rangle \ z_1^{k_1+q(N-1)}\cdot \dots z_{N-1}^{k_{N-1} +q} z_N^{k_N}.
			\end{equation}

					\end{lemma}
\begin{proof}
We will derive a formula for the coefficients of the multinomial and check that $\langle 0\vert W_{k_1} \dots W_{k_N} \vert 0\rangle$ satisfy those too.		
The left side of~\eqref{Laughlin with W} is expanded as:
			\begin{align*}
				&\prod_{1\leq i < j \leq N} (z_i-z_j)^q 
				= z_1^{q(N-1)} \cdot \dots z_{N-1}^q \prod_{1\leq i < j \leq N} \left( 1-\frac{z_j}{z_i} \right)^q \\
				&= z_1^{q(N-1)} \cdot \dots z_{N-1}^q \ \widetilde{\sum_{\ell_{i,j}}} (-1)^{\sum_{1\leq i<j \leq N} \ell_{i,j}} \left( \prod_{1\leq i < j \leq N} \binom{q}{\ell_{i,j}} \right) \prod_{s=1}^N z_s^{\sum_{t=1}^{s-1}\ell_{t,s} -\sum_{u=s+1}^{N} \ell_{s,u}},
			\end{align*}
			where $\widetilde{\sum}_{\ell_{i,j}}$ denotes the sum over a collection of integers $0\leq \ell_{i,j}\leq q$ with $1\leq i<j \leq N$. 
			To compare the coefficients with the right side of~\eqref{Laughlin with W}, 
			we use the definition~\eqref{definition Wm} together with Lemma \ref{prop:DW}~\eqref{CR D} to bring all the $D^-$ to the left:
			\begin{align*}
				 W_{m_1} \dots W_{m_N}   = & \sum_{k_1, \dots, k_N \in \mathbb{Z}} D_{m_1+k_1}^- D_{k_1}^+ \dots D_{m_N+k_N}^- D_{k_N}^+ \\
				= & \sum_{k_1, \dots, k_N \in \mathbb{Z}} \ \widetilde{\sum_{\ell_{i,j}}} \ (-1)^{\sum_{1\leq i<j \leq N} \ell_{i,j}}  \left( \prod_{1\leq i < j \leq N} \binom{q}{\ell_{i,j}} \right)  \\
				& \mkern50mu  \times  \left( \prod_{s=1}^N D^-_{m_s+k_s-\sum_{t=1}^{s-1} \ell_{t,s}} \right) \left( \prod_{s=1}^N D^+_{k_s-\sum_{t=s+1}^N \ell_{s,t}} \right) .
			\end{align*}
			Since $	D_\ell^+ \vert 0 \rangle = \delta_{\ell, 0} \vert 0 \rangle =(D_\ell^-)^* \vert 0\rangle $, 
			 we hence obtain
			\begin{align*}
			 \langle 0 \vert W_{m_1} \dots W_{m_N} \vert 0\rangle  =&\  \widetilde{\sum_{\ell_{i,j}}} \ (-1)^{\sum_{1\leq i<j \leq N} \ell_{i,j}} \left( \prod_{1\leq i < j \leq N} \binom{q}{\ell_{i,j}} \right)   \\
								&\mkern10mu  \times  \sum_{k_1, \dots, k_N \in \mathbb{Z}} \left(\prod_{s=1}^N \delta_{k_s+m_s,\sum_{t=1}^{s-1} \ell_{t,s}} \right) \left( \prod_{s=1}^N \delta_{k_s, \sum_{t=s+1}^N \ell_{s,t}} \right) ,
			\end{align*}
			which implies the result.
\end{proof}

As is also well known \cite[Section 2.2]{estienne2013fractional}, the above may be generalized to include factors of power sum symmetric polynomials 
		\begin{align*}
			p_n(z_1, \dots, z_N) = \sum_{j=1}^N z_j^n.
		\end{align*}
In the iMPS, they are accommodated in terms of boundary charges. 
				\begin{cor} \label{prop:MPS power symmetric}
			For any $n_1, \dots, n_K \in \mathbb{N}$:				
			\begin{align}\label{eq:MPS power symmetric}
					& \left( \prod_{j=1}^K \sqrt{q} \ p_{n_j}(z_1, \dots, z_N) \right) \prod_{1\leq j<k\leq N} (z_j-z_k)^q  \notag \\
					 & \qquad =  \sum_{k_1, \dots, k_N\in \mathbb{Z}} \langle 0 \vert \left(\prod_{j=1}^K a_{n_j} \right) W_{k_1-q(N-1)} \dots W_{k_{N-1}-q} W_{k_N}\vert 0 \rangle \  z_1^{k_1}\cdots z_N^{k_N} .	\qquad 
					 \end{align}
		\end{cor}
		\begin{proof}
		We shift the momenta, $ k_j \to k_j + q(N-j) $ in the right side and pull through all annihilation operators one by one to the right employing the commutation rule in the second item of Proposition~\eqref{prop:DW}.  In case $ K = 1 $ we hence get using $a_n \vert 0\rangle =0$ and~\eqref{Laughlin with W}:
			$$
				 \sum_{k_1, \dots, k_N\in \mathbb{Z}}\!\! \langle 0 \vert a_n W_{k_1} \dots W_{k_N} \vert 0 \rangle \ z_1^{k_1+q(N-1)}\cdot \dots z_{N-1}^{k_{N-1} +q} z_N^{k_N} 
				= \sqrt{q} \left(\sum_{j=1}^N z_j^n\right)\!\! \prod_{1\leq i < j \leq N} (z_i-z_j)^q 
			$$
			The proof for general $K$ is exactly the same.
			\end{proof}

\subsection{Fundamental properties of operator products}

A repeated application of  the commutation rule from Lemma~\ref{prop:DW} in the coefficients on the right side of~\eqref{eq:MPS power symmetric} results in the appearance of operator products of the form: 
		\begin{equation}\label{eq:Wproductchange}
		W_{k_1-q(N-1) } \dots W_{k_{N-1}-q} W_{k_N} = (-1)^{qN(N-1)/2} \ W_{k_N-q(N-1)} W_{k_{N-1}-q(N-2)} \dots W_{k_1} 	
		\end{equation}
		which are well-defined on $ \mathcal{H}_0 $ for any collection of momenta $ \mathbf{k} = (k_1,k_2, \dots, k_N) \in \mathbb{Z}^N $.
The following theorem captures all the fundamental properties of these operator products. 
\begin{theorem}  \label{vanishing condition}
For any $ N \in \mathbb{N} $ the operator product 
\begin{equation} \label{def mathcal W}
			 \mathbb{W}(\mathbf{k}) \coloneqq W_{k_N-q(N-1)} W_{k_{N-1}-q(N-2)} \dots W_{k_2 - q} W_{k_1} 
			 \end{equation}
with $ \mathbf{k} = (k_1,k_2, \dots, k_N) \in \mathbb{Z}^N $ defined on $ \mathcal{H}_0 $ has the following properties: 
\begin{enumerate}
\item permutation symmetry, i.e., for any permutation $ \sigma \in \mathcal{S}_N $
\begin{equation}\label{eq:permutations} 
 \mathbb{W}(\mathbf{k}) = \mathrm{sgn}(\sigma)^q  \  \mathbb{W}(k_{\sigma(1)},k_{\sigma(2)}, \dots, k_{\sigma(N)})  .
\end{equation}
\item factorization property, $  \displaystyle \mathbb{W}(\mathbf{k}) \vert 0 \rangle \in \mathcal{H}\big(\sum_{j=1}^N (k_j - q(j-1))\big) $. 
For any $  \mathbf{k}  \in  \mathbb{N}_0^N  $ with $\displaystyle \sum_{j=1}^N k_j \leq  q \frac{N(N-1)}{2} $:
\begin{equation}\label{eq:partition}
 \mathbb{W}(\mathbf{k}) \vert 0 \rangle = 	\langle 0 \vert \mathbb{W}(\mathbf{k}) \vert 0 \rangle \ \vert 0 \rangle   .
\end{equation}
\item normalization, $  \langle 0 \vert \mathbb{W}\left( \mathbf{k}^{(q)}_N\right) \vert 0 \rangle = 1 $,  for 
\begin{equation}\label{eq:rootp}
  \mathbf{k}^{(q)}_N \coloneqq \left( 0, q , \dots , q (N-2) , q (N-1) \right) . 
\end{equation}
\item\label{item:partitions} $  \displaystyle  \langle 0 \vert \mathbb{W}(\mathbf{k}) \vert 0 \rangle =0 $ unless  $  \mathbf{k}  \in  \mathbb{N}_0^N  $, and for any $ s \in \{1 , \dots , N\} $:
\begin{equation} \label{eq:domination}
	\sum_{j=1}^s k_j \geq  \sum_{j=1}^s q(j-1) , \quad\mbox{and}\quad  \sum_{j=s}^N k_j \leq  \sum_{j=s}^N q(j-1) ,
\end{equation}
and, consequently, $  \displaystyle  \sum_{j=1}^N k_j =  \sum_{j=1}^N q(j-1) = q \frac{N(N-1)}{2} $.

\end{enumerate}
\end{theorem}
\begin{proof}
	1.~The permutation symmetry~\eqref{eq:permutations}  is a straightforward generalization of~\eqref{eq:Wproductchange} using the commutation rule from Proposition~\eqref{prop:DW}. \\
\noindent		
2.~The  inclusion 
	$$  \mathbb{W}(\mathbf{k}) \vert 0 \rangle \in \mathcal{H}\left( \sum_{j=1}^N ( k_j  - q(N-j) ) \right) $$ 
	expresses the momentum conservation derived in Lemma~\ref{prop:DW}. Note that $ \mathcal{H}\left(0\right) = \text{span}\{ \vert 0 \rangle \} $ and $ \mathcal{H}\left(M\right) = \{ 0 \} $ for any $ M < 0 $. This  proves~\eqref{eq:partition}. \\
	\noindent	
	3.~The claimed normalization in case $ \mathbb{W}\left( \mathbf{k}^{(q)}_N\right) = W_0^N $ follows from $ \langle 0 | W_0 | 0 \rangle = 1 $ and the previous item. \\
		\noindent	
	4.~The fact that the momenta  $  \mathbf{k} $ need to be non-negative integers follows from the 
	representation~\eqref{eq:MPS power symmetric} as a polynomial, which by construction only involves non-negative powers, together with \eqref{eq:Wproductchange}. 
	The inequalities~\eqref{eq:domination} result from the momentum conservation established in~2. More precisely, the first inequality follows from 2. with $ N = s $. The second inequality follows by considering the adjoint operators, for which momentum conservation implies 
	 $$ \mathbb{W}(\mathbf{k})^* \vert 0 \rangle \in \mathcal{H}\big(\sum_{j=1}^N (q(j-1)-k_j)\big)  . $$
	Setting $ N = s $ again completes the proof. 
\end{proof}

\noindent		
The above theorem is the starting point for most implications discussed in this work:
\begin{enumerate}
\item Lemma~\ref{thm:MPS Laughlin} together with~\eqref{eq:Wproductchange} when combined with the permutation symmetry~\eqref{eq:permutations}   allows to express the $q$th power of the 
Vandermonde determinant in terms of Slater determinants ($q$ odd), respectively, permanents ($q$ even).
This is the key identity behind the the expansion of the Laughlin wavefunction in the occupation basis. 
\item The permutation symmetry~\eqref{eq:permutations} together with item~\ref{item:partitions} can be used to restrict attention to momenta $  \mathbf{k} =(k_1, \dots, k_N)  \in  \mathbb{N}_0^N  $ which 
\begin{itemize}
\item are ordered, $ 0 \leq k_1 \leq \dots \leq k_N $, and
\item sum up to the total momentum $  \sum_{j=1}^N k_j= q \frac{N(N-1)}{2} $. 
\end{itemize}
In other words, the momenta are partitions of the total momentum of the root partition $ \lambda_N^{(q)} $, cf.~Definition~\ref{def:partitions}. 
The second relation in~\eqref{eq:domination} expresses the domination of $ \lambda $ by $ \lambda_N^{(q)} $, cf.~Definition~\ref{def:dom}. 
\item The property~\eqref{eq:partition} is the key to the factorization of the expansion coefficients
of the Laughlin wavefunction as detailed in Theorem~\ref{thm:fact}.
Pictorially, factorization occurs if and only if  $\sum_{j=1}^s k_j$ intersects $q\sum_{j=1}^s (j-1)=\frac{q}{2}s(s-1)$, cf.~Figure~\ref{fig:factor}. This is related to the notion of renewal points, cf.~Lemma~\ref{lm:splitting}.

	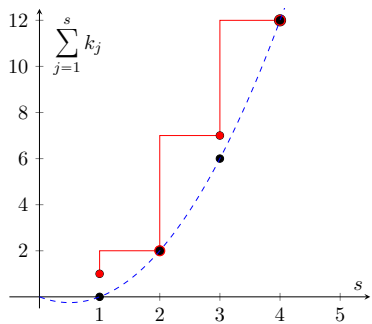
\begin{figure}[h]
	\begin{center}
	\begin{tikzpicture}[scale=0.7]
		
		\begin{axis}[
			axis lines = middle,
			xlabel = {$s$},
			ylabel = {$\displaystyle \ \sum_{j=1}^{s} k_j$},
			xmin=-0.5, xmax=5.5,
			ymin=-0.5, ymax=12.5		]	
			
			\draw[red] (1,1) -- (1,2) -- (2,2) -- (2,7) -- (3,7) -- (3,12) -- (4,12);
			
			\draw (1,1) circle[radius=2pt];
			\fill[red] (1,1) circle[radius=2pt];
			\draw (2,2) circle[radius=2.5pt];
			\fill[red] (2,2) circle[radius=3pt];
			\draw (3,7) circle[radius=2pt];
			\fill[red] (3,7) circle[radius=2pt];
			\draw (4,12) circle[radius=3pt];
			\fill[red] (4,12) circle[radius=3pt];
			
			\draw (1,0) circle[radius=2pt];
			\fill (1,0) circle[radius=2pt];
			\draw (2,2) circle[radius=2pt];
			\fill (2,2) circle[radius=2pt];
			\draw (3,6) circle[radius=2pt];
			\fill (3,6) circle[radius=2pt];
			\draw (4,12) circle[radius=2pt];
			\fill (4,12) circle[radius=2pt];
			
			\addplot [
			name path = A1,
			-,
			domain = 0:5, dashed, blue,
			samples = 100] {x*x-x} ;

		\end{axis}	
		
	\end{tikzpicture}
\end{center}
\caption{Graphical depiction of factorization for the example $ \mathbf{k} =(1, 1, 5 , 5 ) $. Since $ k_j \geq 0 $ for all $ j $, the sums $\sum_{j=1}^{s} k_j$ are increasing. On top of this, the constraint~\eqref{eq:domination} requires the sum to lie above the dotted parabola $ \frac{q}{2}s(s-1)$. Instances of factorization occur when the sum touches the parabola.}\label{fig:factor}
\end{figure}
\end{enumerate}

Before following up on the above points, we present bounds on the expansion coefficients $ 	\langle 0 \vert \mathbb{W}(\mathbf{k}) \vert 0 \rangle  $, which will be essential for our results on the thin cylinder. 
	
\subsection{Key estimate}
While the operators $ W_m $ entering the product~\eqref{def mathcal W} in our iMPS representation are unbounded, they are bounded once one weights them by an exponential of the momentum operator $ L_0 $ defined in \eqref{def:L0}. This idea is inspired by findings in~\cite{konig2017matrix}. 
		\begin{theorem} \label{thm:key estimate}
			For every $m\in \mathbb{Z}$ we have the operator-norm bound
			\begin{equation} \label{operatornorm bound Wm}
				\Vert e^{-(4q+1) L_0/2 } W_m e^{-(4q+1) L_0/2 } \Vert \leq 1
			\end{equation}
					\end{theorem}
			\begin{proof}
					Let $ P_M $ with $ M \in \mathbb{N}_0 $ stand for the orthogonal projection onto the subspace  $ \mathcal{H}(M) $ of total momentum $ M $, again with the convention that $ P_M = 0 $ in case $ M < 0 $. 
					Since $ W_m \mathcal{H}(M)  \subseteq  \mathcal{H}(M+m) $ and $ L_0 $ is diagonal with respect to the orthogonal decomposition~\eqref{eq:orthsum}, we have
					\begin{equation}\label{eq:normestkey}
						\Vert e^{-(4q+1)L_0/2} W_m e^{-(4q+1)L_0/2} \Vert  \leq \sup_{M\in \mathbb{N}_0}  \left\| P_{M+m} e^{-(2q+1)L_0/2} W_m e^{-(4q+1)L_0/2} P_M \right\| .
					\end{equation}
					To estimate the norm, we may assume without loss of generality that $ m \geq -M $, since otherwise the norm is zero. 
					We now distinguish two cases. In case $ M \in \mathbb{N} $, we estimate for $ x \in  \mathcal{H}(M+m) $ and $ y \in  \mathcal{H}(M)$ :
					\begin{align*}
			 \vert \langle x\vert  W_m y\rangle \vert   & \leq   \sum_{\ell = \min\{0, -m\} }^\infty \vert \langle \left( D_{\ell+m}^-\right)^*   x\vert  D_\ell^+  y\rangle \vert \\
			&\leq  \left(\sum_{k\in \mathbb{N}_0} \frac{\sqrt{q}^k}{k!} \sum_{j_1, \dots, j_k\in \mathbb{N}} \frac{1}{j_1 \dots j_k} \Vert a_{j_1} \dots a_{j_k} x\Vert\right) \\
			&\mkern30mu \times \left(\sum_{n\in \mathbb{N}_0} \frac{\sqrt{q}^n}{n!} \sum_{j_1, \dots, j_n\in \mathbb{N}} \frac{1}{j_1 \dots j_n} \Vert a_{j_1} \dots a_{j_n} y\Vert\right) \\
			&\leq   \left(\sum_{k\in \mathbb{N}_0} \frac{\sqrt{2q}^k}{k!} \Vert L_0^{k/2} x \Vert \right) \left(\sum_{n\in \mathbb{N}_0} \frac{\sqrt{2q}^n}{n!}  \Vert L_0^{n/2} x \Vert \right) \\
			&\leq e^{2q} \ \Vert e^{L_0/2} x \Vert \ \Vert e^{L_0/2} y \Vert  = e^{( 4q + 2M+ m)/2}  \ \Vert x \Vert \Vert y \Vert 
		\end{align*}
		The second inequality used the triangle and the Cauchy-Schwarz inequality. The third inequality follows from a Cauchy-Schwarz estimate applied to the $ j $-sums together with the crude bound $\sum_{n=1} n^{-2}=\pi^2/6< 2$ and Lemma~\ref{lem:simplebound} below. The fourth inequality is by Cauchy-Schwarz again and the self-adjointness of $ L_0 $.
		Since $ 2M+ m \geq M \geq 1 $ by assumption, the exponential on the right side is bounded by $ \exp( (4q+1) (2M+m)/2 ) $. Passing this exponential to the left side, and using \eqref{eq:L_0evs} yields
		$$
		 \vert \langle x\vert  e^{-(4q+1)L_0/2} W_m e^{-(4q+1)L_0/2} y\rangle \vert \leq \Vert x \Vert \Vert y \Vert ,
		$$
		and hence the claim in this case. 
		
		In case $ M = 0 $,  the norm on the right side of~\eqref{eq:normestkey} is zero if $ m < 0 $ and one if $ m = 0 $. We then estimate for $ x \in  \mathcal{H}(m) $ and $ m \geq 1 $ similarly as above:
				\begin{align*}
					 \vert \langle x |  W_m  | 0 \rangle \vert \leq  \| (D_m^-)^*  x \|  & \leq \sum_{k\in \mathbb{N}_0} \frac{\sqrt{q}^k}{k!} \sum_{j_1, \dots, j_k\in \mathbb{N}} \frac{1}{j_1 \dots j_k} \Vert a_{j_1} \dots a_{j_k} x\Vert \\
					 & \leq e^{q} \  \Vert e^{L_0/2} x \Vert = e^{(2q +m)/2} \  \Vert x \Vert \leq e^{(4q +1) m /2}  \ \Vert x \Vert ,
				\end{align*}
				which again yields the desired bound. 
			\end{proof}
		The proof was based on the following
	\begin{lemma}\label{lem:simplebound}
			For all $m\in \mathbb{N}$ and $ x \in \mathcal{H}_0 $:
			\begin{equation} \label{norm of L0}
				\sum_{j_1, \dots, j_m \in \mathbb{N}} \Vert a_{j_1} \dots a_{j_m} x \Vert^2 
				\leq \langle x\vert L_0^m \vert x\rangle = \Vert L_0^{m/2} x \Vert^2.
			\end{equation}
		\end{lemma}
		\begin{proof}
			 We prove	~\eqref{norm of L0} by induction. The case $m=1$ is trivial. For $m=2$ we observe that the commutation relations~\eqref{commutation relation} yield
			$a_j^* a_k^* a_k a_j = a_j^* a_k^* a_j a_k
				= a_j^* a_j a_k^* a_k - j \delta_{j,k} a_j^* a_j $, such that 
			 for any $x\in \mathcal{H}_0$
			\begin{align*}
				0 &
				\leq \sum_{j, k \in \mathbb{N}} \langle x\vert  a_j^* a_k^* a_k a_j \vert x \rangle
				= \sum_{j, k \in \mathbb{N}} \langle x \vert a_j^* a_j a_k^* a_k \vert x \rangle - \sum_{j \in \mathbb{N}} j \langle x\vert a_j^* a_j \vert x \rangle \\
				&\leq  \sum_{j, k \in \mathbb{N}} \langle x\vert a_j^* a_j a_k^* a_k \vert x \rangle 
				= \langle x\vert L_0^2 \vert x \rangle. 
			\end{align*}
			For $m\geq 3$, a similar computation using~\eqref{commutation relation L0 aj} yields
			\begin{align*}
				0& \leq \sum_{j_1, \dots, j_m \in \mathbb{N}} \Vert a_{j_1} \dots a_{j_m} x \Vert^2 
				\leq \sum_{j_1, \dots, j_m \in \mathbb{N}} \langle x\vert  a_{j_1}^* a_{j_2}^* \dots, a_{j_{m-2}}^* L_0^2 a_{j_{m-2}} \dots a_{j_1} \vert x \rangle \\
				&\leq \sum_{j_1, \dots, j_{m-2} \in \mathbb{N}} \langle x\vert  a_{j_1}^* a_{j_2}^* \dots, a_{j_{m-2}}^* (L_0+j_1+\dots + j_{m-2})^2 a_{j_{m-2}} \dots a_{j_1} \vert x \rangle \\
				&= \sum_{j_1, \dots, j_{m-2} \in \mathbb{N}} \langle  x\vert L_0 a_{j_1}^* a_{j_2}^* \dots a_{j_{m-2}}^* a_{j_{m-2}} \dots a_{j_1}  L_0  \vert x \rangle,
			\end{align*}
			which, by iteration, implies the claim.
		\end{proof}

		\begin{cor}\label{cor:keyest}
			For all $ \mathbf{k} \in \mathbb{N}_0^N $:
			\begin{equation} \label{key estimate}
					\Vert   \mathbb{W}(\mathbf{k})  \vert 0 \rangle \Vert \leq \exp\left((4q+1)  \sum_{j=1}^N (N+\frac{1}{2}-j)(k_j -q(j-1))  \right)  . 
			\end{equation}
		\end{cor} 
		\begin{proof}

		Using the commutation relations \eqref{commutation relation L0 aj} and the definition \eqref{definition Wm}, we readily obtain
		\begin{align*}
			W_m L_0 = (L_0-m) W_m .
		\end{align*}
		This is used to move an exponential of the momentum operators through:
		\begin{align*}
			&W_{k_j-q(j-1)} e^{-(4q+1) (N+1-j) L_0} \\
			&\quad= W_{k_j-q(j-1)} e^{-(4q+1) (N+1/2-j) L_0}  e^{-(4q+1) L_0/2} \\
			&\quad= e^{-(4q+1) (N+1/2-j) (L_0-(k_j-q(j-1)))} W_{k_j-q(j-1)} e^{-(4q+1) L_0/2} \\
			&\quad= e^{(4q+1) (N+1/2-j)(k_j-q(j-1))} e^{-(4q+1) (N+1-(j+1)) L_0} \left(e^{-(4q+1) L_0/2} W_{k_j-q(j-1)} e^{-(4q+1) L_0/2} \right).
		\end{align*}
		Applying this relation $N$ times, and using $e^{-(4q+1)N L_0} \vert 0 \rangle = \vert 0 \rangle$ and the bound~\eqref{operatornorm bound Wm}, we arrive at the claim. 
		\end{proof}
		
\section{Partitions, dominance order and renewal structure}\label{sec:partitions}
	Theorem~\ref{vanishing condition} laid out the fundamental structure of the expansion coefficients of power symmetric polynomials times the $q$th power of the Vandermonde determinant. It was proven that the coefficients of the Laughlin function are zero unless they form a partition of the total momentum of the root partition~\eqref{eq:rootp}. In this section, we explore this structure of the non-zero coefficients using the language of partitions and the concept of dominance order and squeezing. 
We also collect some auxiliary combinatorial results that will help us control the entropy of this expansion. Much of the mathematical background for this chapter can be found in~\cite{macdonald1995symmetric, stanley1989some}, which we reformulate for our purposes. In the context of the Laughlin function, many of them have already been used in~\cite{jansen2009symmetry} (see also~\cite{DiGioacchinoMasterthesis} for a pedagogical introduction).
	
\subsection{Partitions, their addition and dominance order} 
Our definition of partition is slightly different from the one typically used in the mathematics literature (e.g.\ the standard definition of a partition considers positive integers and monotone decreasing families). 
		\begin{defn}\label{def:partitions}
		1. A \emph{partition} $\lambda=(\lambda_1, \dots, \lambda_k) \in \mathbb{N}_0^k$ of an integer $|\lambda| \in \mathbb{N}_0 $ is a monotone increasing family of nonnegative integers summing to $|\lambda| = \sum_{j=1}^k \lambda_k $. 		The integer $k$ is called the \emph{length of the partition} $\lambda=(\lambda_1, \dots, \lambda_k)$, and we write $\ell(\lambda):=k$.\\
		We call
			\begin{equation}\label{def:occupation}
				m(\lambda, j) \coloneqq \vert \{ k\in \{1, \dots, N \} \ : \ \lambda_k = j\} \vert 
			\end{equation}
			the \emph{occupation number} of $j\in \mathbb{N}_0$, and set $ M(\lambda)! \coloneqq \prod_{j=0}^\infty m(\lambda, j)! $. \\
\noindent	 2. For partitions $\lambda, \alpha$ of length $k$,  we define their \emph{sum}
		\begin{equation} \label{def addition partitions}
			\lambda+\alpha :=(\lambda_1+\alpha_1, \dots, \lambda_k + \alpha_k).
		\end{equation}
		If $\sigma\in S_k$ is a permutation, then we define
		\begin{equation} \label{def addition partition with permutation}
			(\lambda+\alpha)_\sigma := S(\lambda_1+\alpha_{\sigma(1)}, \dots, \lambda_k+\alpha_{\sigma(k)}),
		\end{equation}
		where $S$ reorders the sequence into increasing order.
\end{defn}
\begin{rem}
The root partition from Defintion~\ref{def:rootp} is an example of a sum of partitions,  $   \lambda^{(q)}_N(b) =  \lambda^{(q)}_N + b $ for any partition $ b \in \mathbb{N}_0^N $.\\
In Theorem~\ref{thm:root states}, we have also introduced the abbreviation  $\alpha_\sigma \coloneqq  (\alpha_{\sigma(1)}, \dots, \alpha_{\sigma(k)}) $ for a permuted partition, which generally does not yield a partition. We note that, in contrast to~\eqref{def addition partition with permutation}, there is no reordering of the sequence. 
\end{rem}
		The number of partitions of a given integer $ n \in \mathbb{N} $ of length at most $ n $ is
		\begin{equation} \label{def p(n)}
			p(n) \coloneq \left\vert \left\{ (\lambda_1, \dots, \lambda_n) \in \mathbb{N}_0^n \ : \ \lambda_1\leq \lambda_2 \leq \dots \leq \lambda_n, \sum_{j=1}^n \lambda_j =n\right\} \right\vert.
		\end{equation}
		We recall a well-known upper bound.
		\begin{prop}[\cite{erdos1942elementary}]
			For every $n\in \mathbb{N}$ we have
			\begin{equation} \label{partition function pointwise bound}
				p(n)\leq \exp\left(\pi \sqrt{\frac{2n}{3}}\right) . 
			\end{equation}
		\end{prop}
\begin{rem} The precise asympotics
$ p(n) \sim \frac{1}{4\sqrt{3}n} \exp\left(\pi \sqrt{\frac{2n}{3}}\right) $ as $ n \to \infty $ was proven originally in \cite{hardy1918asymptotic}. For an elementary proof, see \cite{erdos1942elementary}.
\end{rem}

The notion of dominance order (also called natural order \cite{macdonald1995symmetric, stanley1989some}) is an order relation on the set of partitions of a given integer of fixed length, and will help to sort the non-zero coefficients appearing in the expansions of the powers of the Vandermonde determinant. 	
		\begin{defn}\label{def:dom}
		Let $\lambda=(\lambda_1, \dots, \lambda_k), \mu=(\mu_1, \dots, \mu_k)$ be partitions of some integer $n$ with the same length $k$. We say that $\lambda$ dominates $\mu$, in symbols $\mu \preceq \lambda$, if for all $i\in \{1, \dots, k \}$ we have
		\begin{equation} \label{dominance order def}
			\mu_i + \dots + \mu_k \leq \lambda_i + \dots  + \lambda_k.
		\end{equation}
		We write $\mu \precneqq \lambda$ if $\mu \preceq \lambda$ and $\mu \neq \lambda$. 
				\end{defn}
\begin{rem}				
1.~One should distinguish dominance order from the lexicographic order. Two partitions $\lambda=(\lambda_1, \dots, \lambda_k), \mu=(\mu_1, \dots, \mu_k)$ of length $ k $ are lexicographically ordered, in symbols $\mu \preceq_\mathrm{lex} \lambda$,  if $\mu=\lambda$ or $\mu_j <\lambda_j$ for the largest $j\in \{1, \dots, k\}$ for which $\lambda_j \neq \mu_j$. If $\mu \preceq \lambda$, then $\mu \preceq_\mathrm{lex} \lambda$. 
This is established by the following short argument. Using $i=k$ in \eqref{dominance order def}, one has $\mu_k\leq \lambda_k$. If $\mu_k< \lambda_k$, then $\lambda$ is bigger than $\mu$ in lexicographic order. Otherwise, $\mu_k = \lambda_k$. Then one uses \eqref{dominance order def} with $i=k-1$ to obtain $\mu_{k-1}+\mu_k \leq \lambda_{k-1}+\lambda_k$, which implies $\mu_{k-1}\leq \lambda_{k-1}$. Proceeding like this, we either arrive at $\mu = \lambda$ or $\mu$ is smaller than $\lambda$ in lexicographic order.\\
		The reverse implication does not hold in general. For example, $(1,2,7)$ is bigger in lexicographic order than $(0,5,5)$, but not in dominance order. Indeed, we have $7>5$, but $2+7<5+5$. \\
2.~In the definition of dominance order we could have replaced \eqref{dominance order def} by the condition 
	\begin{equation} \label{dominance order def 2}
		\mu_1 + \dots + \mu_i \geq \lambda_1 + \dots  + \lambda_i.
	\end{equation}
as we require $\lambda$ and $\mu$ to partition the same integer.
\end{rem}

The next lemma recalls, mostly from \cite[Ch. 1]{macdonald1995symmetric}, on how the dominance order behaves under the addition of partitions.
		\begin{lemma}
			Let $\lambda, \alpha$ be partitions of length $k$ and $\sigma\in S_k$ a permutation. 
			\begin{enumerate}
				\item Both $\lambda+\alpha$ and $(\lambda+\alpha)_\sigma$ are partitions of length $k$.
				\item If $\mu$ is yet another partition of length $k$, then we have
				\begin{equation} \label{dominance order addition}
					\mu \preceq \lambda \Rightarrow \mu + \alpha \preceq \lambda + \alpha
				\end{equation}
				and
				\begin{equation} \label{dominance order permutation}
					(\lambda+\alpha)_\sigma \preceq \lambda+\alpha.
				\end{equation}
			\end{enumerate}
		\end{lemma}
		\begin{proof}
1.~As $\alpha$ and $\lambda$ are both increasing, so is $\lambda+\alpha$. By construction $(\lambda+\alpha)_\sigma$ is increasing.\\
\noindent
2.~The implication~\eqref{dominance order addition} is trivial. To prove \eqref{dominance order permutation}, we pick a permutation $\tau\in S_k$ such that $(\lambda+\alpha)_\sigma$ is given by $
					(\lambda_{\tau(1)} + \alpha_{\sigma(\tau(1))}, \dots, \lambda_{\tau(k)} + \alpha_{\sigma(\tau(k))}) $. 
				As $\alpha$ is an increasing sequence, we have for all $s\in \{1, \dots, k \}$
				\begin{align*}
					\alpha_s + \dots + \alpha_k \geq \alpha_{\sigma(s)} + \dots + \alpha_{\sigma(k)},
				\end{align*}
				and similarly for $ \lambda $, which readily implies \eqref{dominance order permutation}.
		\end{proof}

\subsection{Proof of Theorem~\ref{thm:root states} and  Corollaries~\ref{cor:planar} and~\ref{cor:cyl}}\label{subsec:ProofofRoot}

 \begin{proof}[Proof of Theorem~\ref{thm:root states}]
 We first show that the left side in~\eqref{MPS with monomial symmetric} can be expanded using the second line in~\eqref{MPS with monomial symmetric_c}. 
 To do so, we start from the representation~\eqref{def Fb} of the monomial symmetric polynomial as a polynomial $ \pol_{b} $ of power symmetric polynomials. Without loss of generality, we may hence focus on the representation 
 of 
 $$
 F_J(z_1,\dots , z_N ) = \left( \prod_{j\in J} p_j(z_1,\dots,z_N)  \right) \times \prod_{1\leq j<k\leq N} (z_j-z_k)^q ,
 $$
 where $ J \subseteq \{ 1, 2, \dots , |b| \} $ with $ |b| = \sum_{k=1}^N b_k $ is a multi-subset (accommodating multiplicities) with the property $ \sum_{j\in J} j = | b| $ reflecting the fact that $ m_{b} $ is homogeneous of degree degree~$  | b | $.  
Corollary~\ref{prop:MPS power symmetric} and~\eqref{eq:Wproductchange} implies that
\begin{align}\label{eq:Laughint}
 F_J(z_1,\dots , z_N ) = & \sum_{k_1, \dots, k_N\in \mathbb{Z}} (-1)^{qN(N-1)/2} \ \widehat w_J(k_1, \dots k_N) \  \prod_{j=1}^N z_j^{k_j} , \\
 & \quad\mbox{with}\quad  \widehat w_J(\mathbf{k}) = \langle 0 \vert \left(\prod_{j\in J } \frac{a_{j}}{\sqrt{q}} \right) \mathbb{W}(\mathbf{k}) \vert 0 \rangle . \notag
\end{align}
From Theorem~\ref{vanishing condition}  $(4)$, we learn that we may restrict the summation to $ \mathbf{k} \in \mathbb{N}_0^N $ such that 
$$ \sum_{j=1}^N k_j = q \frac{N(N-1)}{2} +   \sum_{j\in J} j  = q \frac{N(N-1)}{2} + |b| . $$
Moreover, the coefficient  $  \widehat  w_J(\mathbf{k}) $ is permutation symmetric, i.e., $$ \widehat  w_J(k_{1} , \dots ,k_{N} ) = (\textrm{sgn} \sigma)^q \  \widehat  w_J(k_{\sigma(1)} , \dots ,k_{\sigma(N)} ) . $$
For any such function, we may rewrite the summation
\begin{equation}\label{eq:correctcounting}
 \sum_{k_1, \dots, k_N\in \mathbb{N}_0} \widehat w_J(k_1, \dots k_N) \  \prod_{j=1}^N z_j^{k_j} = \sum_{\lambda}  \frac{ \widehat w_J(\lambda)}{M(\lambda)!}  \sum_{\sigma\in S_N} \mathrm{sgn}(\sigma)^q \prod_{j=1}^N z_{j}^{\lambda_{\sigma(j)}} 
\end{equation}
in terms of partitions $ \lambda=(\lambda_1, \dots, \lambda_N ) $ of the integer $ q N(N-1)/2 + |b| $. 
This proves~\eqref{MPS with monomial symmetric} with the second line in~\eqref{MPS with monomial symmetric_c}, albeit without restricting the summation to those partitions $ \lambda $, which are dominated by $ \lambda^{(q)}_N(b ) $. 

For proof of this, we first consider the case $ b = 0 $. In this case, the claim follows from Theorem~\ref{vanishing condition}: its third item ensures that $ w_\emptyset(\lambda) = 0 $ unless $ \sum_{j=s}^N \lambda_j \leq \sum_{j=s}^N q(j-1) $  for any $ s \in \{1, \dots , N\} $. i.e.,  $\lambda \preceq  \lambda^{(q)}_{N} $. 

In case of a non-zero partition $ b $, we show that the left side in~\eqref{MPS with monomial symmetric} can be expanded in terms of partitions $  \lambda  \preceq  \lambda^{(q)}_{N}(b) $ with expansion coefficients from the first line in~\eqref{MPS with monomial symmetric_c}. By the uniqueness of expansion, this also proves that one may restrict to those partitions using the coefficients from the second line in~\eqref{MPS with monomial symmetric_c}. 

We start from the representation~\eqref{eq:Laughint} with $ J = \emptyset $ and the definition~\eqref{def monomial symmetric polynomial} of monomial symmetric polynomials to rewrite
\begin{align*}
& (-1)^{qN(N-1)/2}  M(b)! \ m_{b}(z_1, \dots, z_N)  \prod_{1\leq j<k\leq N} (z_j-z_k)^q \notag \\
& = \sum_{\tau \in \mathcal{S}_N}   \sum_{k_1, \dots, k_N\in \mathbb{Z}}   \langle 0 \vert \mathbb{W}(\mathbf{k}) \vert 0 \rangle  \prod_{j=1}^N z_j^{k_j+b_{\tau(j)}}  \notag \\
& =   \sum_{k_1, \dots, k_N\in \mathbb{Z}}\ \sum_{\tau \in \mathcal{S}_N}  \langle 0 \vert \mathbb{W}(\mathbf{k},b_{\tau}) \vert 0 \rangle \   \prod_{j=1}^N z_j^{k_j} .
\end{align*}
The last line results from the change of variables $ k_j \to k_j - b_{\tau(j)} $ and definition~\eqref{def:Wwithb}. Since we are expanding a polynomial, we may restrict the summation to $ k_1, \dots, k_N \in \mathbb{N}_0 $. Since the expansion coefficients are permutation symmetric,
we may again use~\eqref{eq:correctcounting} to express the summation in terms of partitions $ \lambda $. The expansion coefficients agree with the first line in~\eqref{MPS with monomial symmetric_c}, i.e., a sum over permutations $ \tau \in \mathcal{S}_N $ over $   \langle 0 \vert \mathbb{W}(\lambda_1 -b_{\tau(1)}, \dots , \lambda_N -b_{\tau(N)} ) \vert 0 \rangle $.  By Theorem~\ref{vanishing condition}  $(4)$ these terms are zero unless
for all $ s \in \{1,\dots, N\} $:
$$
\sum_{j=s}^N \lambda_j \leq \sum_{j=s}^N b_{\tau(j)} +  \sum_{j=s}^N q (j-1) ,
$$
i.e., $ \lambda  \preceq  (\lambda^{(q)}_{N} + b)_\tau $ using~\eqref{def addition partition with permutation}.  By \eqref{dominance order permutation}  for any $ \tau \in \mathcal{S}_N $ we have $(\lambda+b)_{\tau} \preceq  \lambda^{(q)}_N +b^{(N)}=  \lambda^{(q)}_N(b)$, which finishes the proof of the dominance order.

In the special case $ \lambda =  \lambda^{(q)}_N(b) $, the only relevant permutation is the identity. For any other permutation~$\tau$, there exists a smallest $j_0\in \{ 1, \dots, N\}$ such that $b_{\tau(j)}=b_j$ for all $j\in \{1, \dots, j_0-1\}$ and $b_{\tau(j_0)}> b_{j_0}$. By Theorem \ref{vanishing condition} $(4)$ this implies
	$$ \langle 0 \vert \mathbb{W}(0, q + (b_1-b_{\tau(1)}), \dots, q(N-1)+b_N-b_{\tau(N)}) \vert 0\rangle =0. $$  One thus arrives at $ w_b( \lambda^{(q)}_N(b) ) = w_0( \lambda^{(q)}_N ) = 1 $. 
 \end{proof}

 \begin{proof}[Proof of Corollary~\ref{cor:planar}]
 This is a consequence of  Theorem~\ref{thm:root states}. 
 The normalization~\eqref{eq:kappanorm} absorbes the factor $(-1)^{qN(N-1)/2}$ from the right side of~\eqref{MPS with monomial symmetric}, the missing $ \sqrt{N!} $ as well as $ \sqrt{\pi}^N  $ in the Slater determinant or permanent~\eqref{eq:Slater} caused by the factors in the basis~\eqref{eq:orbitalsp} compared to the plain monomials featuring in the right side of~\eqref{MPS with monomial symmetric}. The remaining factorials from~\eqref{eq:orbitalsp} are put into the geometric factor $ g_b(\lambda) $. The latter also absorbs the factorials which were included in~\eqref{eq:kappanorm}. 
 
 The claimed normalization, $ h_b(\lambda^{(q)}_N(b) ) = 1  $, follows from the respective normalization of $ w_b $ established in Theorem~\ref{thm:root states}, the one for $ g_b $ and the fact that $ m((\lambda^{(q)}_N(b) , k ) != 1 $ for any $ k\in \mathbb{N}_0 $. 
  \end{proof}
  
 \begin{proof}[Proof of Corollary~\ref{cor:cyl}] To apply Theorem~\ref{thm:root states}, we replace the monomials $ z_j $ by $ e^{\gamma z_j} $. The remainder of the proof proceeds analogously to  Corollary~\ref{cor:planar}. The normalization~\eqref{eq:norm2cyl} absorbes $(-1)^{qN(N-1)/2}$ from the right side of~\eqref{MPS with monomial symmetric} as well as $ \sqrt{N!} $ and the missing factors in the Slater determinant or permanent~\eqref{eq:Slater} caused by the prefactors aside from the exponential ($  e^{- \gamma^2 k^2/2}  $) in the basis~\eqref{eq:cylinderONB}. The geometric factor $ g_b $ takes care of these exponentials as well as the exponential in the normalization~\eqref{eq:norm2cyl}.
 
The normalization  $ h_b(\lambda^{(q)}_N(b) ) = 1  $ follows, similarly as above, from Theorem~\ref{thm:root states}. 
 \end{proof}
\subsection{Squeezing and behavior of the cylindrical weight function}
The squeezing operation is a method for sorting through the tree of partitions in dominance order.
\begin{defn}
Let $\lambda$ be a partition, $0\leq i <j \leq \ell(\lambda)$ and $s\in \mathbb{N}_0$ with $ s \leq (\lambda_j-\lambda_i)/2$, $\lambda_i<\lambda_{i+1}, \lambda_j >\lambda_{j-1}$. Then a squeezing operator $R_{ij}^s$ is called \emph{$\lambda$-admissible}, and it acts in the following way 
		\begin{align*}
			R_{ij}^s (\lambda_1, \dots, \lambda_{\ell(\lambda)})
			= S(\lambda_1, \dots, \lambda_{i-1}, \lambda_i+s, \lambda_{i+1}, \dots, \lambda_{j-1}, \lambda_j -s, \lambda_{j+1}, \dots, \lambda_{\ell(\lambda)}),
		\end{align*}
		where $S$ reorders the sequence into increasing order. 
\end{defn}
Squeezing operators applied to a partition yield a partition of the same integer and with the same length. As is well known~\cite[Ch. 1]{macdonald1995symmetric}, squeezings induce the dominance order. We summarize several facts from~\cite[Ch. 1]{macdonald1995symmetric} adapted to our setting in the following proposition. It is mainly included to be able to connect with the physicist's notion of squeezing. 
\begin{prop}\label{lm:characterization dominance order}
			Let $\lambda, \mu$ be partitions of the same integer and the same length. 		
			\begin{enumerate}
				\item For all  $ \lambda $-admissible squeezing operators $R_{i,j}^s$, we have 
				$	R_{i,j}^s \lambda \preceq \lambda $. 
				\item The following statements are equivalent:
				\begin{enumerate}
					\item $\lambda$ dominates $\mu$.
					\item There exists a finite sequence of  $ \lambda $-admissible squeezing operators
					$R_{i_1, j_1}^{s_1}, \dots, R_{i_m, j_m}^{s_m}$ such that
					$
						\mu = R_{i_m, j_m}^{s_m} \dots R_{i_1, j_1}^{s_1} \lambda $.
					\item There exists a finite sequence of  $ \lambda $-admissible squeezing operators
					$R_{i_1, j_1}^{1}, \dots, R_{i_m, j_m}^{1}$ such that
					$	\mu = R_{i_m, j_m}^{1} \dots R_{i_1, j_1}^{1} \lambda $.
				\end{enumerate}	
			\end{enumerate}	
		\end{prop}
		\begin{proof}
		1.~This is a straightforward computation using the definition of the dominance order.\\
		\noindent 
		2.~ Clearly, $c)$ implies $b)$ Using $1.$ item $b)$ implies $a)$  Finally, we show that $a)$ implies $c)$ by induction over the length $\ell=\ell(\lambda)$. In case $\ell=1$ there is nothing to show and the statement is trivially true. We provide an algorithm for the induction step. If $\lambda_\ell=\mu_\ell$, then we are back in the case $\ell-1$. Thus, we can assume that $\lambda_\ell>\mu_\ell$. We then define
$ s \coloneqq \max\{k\in \{1, \dots, \ell\} \ : \ \mu_k > \lambda_k\} $  and 	$\tilde{\ell} \coloneqq \min \{ k\in \{1, \dots, \ell\} \ : \ \lambda_k = \lambda_\ell\} $. 
As $\lambda$ and $\mu$ partition the same integer, such an $s$ must exist. The sqeezing algorithm sets $\tilde{\lambda} \coloneqq R_{s,\ell}^1 \lambda$. We need to show that we still have $\mu \preceq \tilde{\lambda}$. 	
Indeed, for $1\leq t \leq s$ we have
				\begin{align*}
					\tilde{\lambda}_t + \dots + \tilde{\lambda}_\ell &= \lambda_t + \dots + \lambda_{s-1}+ (\lambda_s+1) + \lambda_{s+1} + \dots + \lambda_{\ell-1}+ (\lambda_\ell-1) \\
					&= \lambda_t + \dots + \lambda_\ell \geq \mu_t + \dots + \mu_\ell.  
				\end{align*}
				On the other hand, for $j\in \{s+1, \dots, \ell\}$ we have, by definition of $s$, $\lambda_j\geq \mu_j$ and $\lambda_\ell\geq \mu_\ell+1$. Thus, if $s+1\leq t \leq \tilde{\ell}$, this implies
				\begin{align*}
	\tilde{\lambda}_t + \dots + \tilde{\lambda}_\ell
	&= \lambda_t + \dots + \lambda_{\tilde{\ell}-1}+ (\lambda_{\tilde{\ell}}-1)+\lambda_{\tilde{\ell}+1}+\dots +\lambda_\ell \\
	&\geq\mu_t + \dots + \mu_\ell.
\end{align*}
The case $t>\tilde{\ell}$ is trivial.
We can now repeat this algorithm until the leading coefficient is equal to $\mu_\ell$. Then we are back in the case $\ell-1$.
		\end{proof}
		
A one-step squeezing operation may be thought of as a descent from a root partition on the directed graph of the partial order on partitions dominated by that root. 
The quantity
	\begin{equation} \label{definition Delta}
			\Delta_b(\mu) = \sum_{j=1}^{N} \left[ \lambda_N^{(q)}(b)_j^2-\mu_j^2\right]
		\end{equation}
which, in the cylinder geometry, will help us control the geometric part of the fractional quantum Hall wavefunctions characterized by a root $(N,b) $, is easily seen to be strictly monotone under squeezing, 
$				\Delta_b(R_{i,j}^s \mu) = \Delta_b(\mu) +2s (\mu_j-\mu_i -s)> \Delta_b(\mu) $ for any $\mu $-admissible squeezing $R_{i,j}^{s}$. Moreover, $ \Delta_b( \lambda_N^{(q)}(b)) = 0 $. 
Hence, it is a measure of the distance to the root $  \lambda_N^{(q)}(b) $ in the graph of the partial order, i.e., the Hasse diagram. 

The function $ \Delta_b $ can be used to estimate the right side of the key estimate in Corollary~\ref{cor:keyest}. Using the language of partitions $\mu=(\mu_1, \dots, \mu_N)$, we rewrite the quantity in the right side of~\eqref{key estimate} in terms of	
	\begin{equation}\label{def:Gamma}
			\Gamma_b(\mu) \coloneqq \sum_{j=1}^N (N+\frac{1}{2}-j)(\mu_j -\lambda_N^{(q)}(b)_j) .
	\end{equation}
	This quantity is easily seen to be also strictly monotone under squeezing. It is a non-trivial observation that
		for partitions dominated by the root partition characterized by $ b $, we can bound $\Gamma_b$ in terms of $\Delta_b$.

		 \begin{lemma}\label{lem:GammaDelta}
		For any root $ (N,b)$ and $ \mu\preceq \lambda^{(q)}_N(b) $,
		we have
			\begin{equation} \label{Gamma bounded by Delta}
				\Gamma_b(\mu) \leq \Delta_b(\mu).
			\end{equation}  
		\end{lemma}
		\begin{proof}
			Let $\mu \preceq \lambda_N^{(q)}(b)$. In the following, we abbreviate $\lambda_j \coloneqq \lambda_N^{(q)}(b)_j$, which allows us to write
			$$
				\Delta_b(\mu) - \Gamma_b(\mu) = \sum_{j=1}^N \left[\lambda_j^2-\mu_j^2 - (N+\frac{1}{2}-j)(\mu_j-\lambda_j)\right]
				= \sum_{j=1}^N (\lambda_j-\mu_j)(N+\frac{1}{2}-1+\lambda_j+\mu_j).
			$$
			The partial sums $s_k=\sum_{j=k}^N (\lambda_j-\mu_j)$ for $k\in \{1, \dots, N\}$ are non-negative, $s_k\geq 0$, and $s_1=0$  by definition of the dominance order as  $\mu \preceq \lambda_N^{(q)}(b)$. 
			Using~\eqref{integration by parts} below with $c_j =\lambda_j-\mu_j$, $d_0=0$ and  $d_j=N+1/2-j+\lambda_j+\mu_j$ for $j\in \{1, \dots, N\}$, we arrive at 
			$$
				\sum_{j=1}^N (\lambda_j-\mu_j)(N+\frac{1}{2}-1+\lambda_j+\mu_j)
				=\sum_{j=1}^N c_j (d_j-d_0)
				= \sum_{k=1}^N s_k (d_k-d_{k-1})
				= \sum_{k=2}^N s_k (d_k-d_{k-1}).
			$$
			The last equality is due to $s_1=0$. For $k\in \{2, \dots, N\}$, we calculate 
			\begin{align*}
				d_k-d_{k-1} 
				= \lambda_k-\lambda_{k-1}+\mu_k -\mu_{k-1} -1
				\geq q -1 \geq 0.
			\end{align*}
			The last estimate is based on the explicit form of the root partition $ \lambda =\lambda_N^{(q)}(b) $ and the fact that $\mu_k-\mu_{k-1}\geq 0$ and $ b_k - b_{k-1} \geq 0 $ as $\mu, b$ are partitions. We thus conclude 	$ \Delta_b(\mu) - \Gamma_b(\mu) =\sum_{k=2}^N s_k (d_k-d_{k-1}) \geq 0 $. 
		\end{proof}
		The last proof used the following elementary and widely known summation by parts formula, which we include for the reader's convenience.
		\begin{prop}
			Let $(c_j)_{j=1}^N, (d_j)_{j=0}^N$ be families of real numbers, then we have
			\begin{equation} \label{integration by parts}
				\sum_{k=1}^N c_k (d_k-d_0) = \sum_{k=1}^N s_k (d_k-d_{k-1})
			\end{equation}
			where $s_k=\sum_{j=k}^N c_j$.
		\end{prop}
		\begin{proof}
			The proof follows from a straightforward computation by exchanging two summations.

		\end{proof}

\subsection{Concatenating partitions and renewal points}

To fully harvest the factorization structure~\eqref{eq:partition}, we need to decompose a partition into its irreducible parts. 
To define these notations, we first define the opposite surgery, namely, the concatenation of partitions that are dominated by root partitions of different lengths. 
\begin{defn}\label{def:composition}
Let $ N_1, N_2 \in \mathbb{N} $ and $ b^{(1)} \in \mathbb{N}_0^{N_1} $, $ b^{(2)} \in \mathbb{N}_0^{N_2} $ be two partitions.
		If $\mu^{(1)} \preceq \lambda^{(q)}_{N_1}\big(b^{(1)}\big)$ and $\mu^{(2)}\preceq \lambda^{(q)}_{N_2}\big(b^{(2)}\big)$, then we define the \emph{concatenation} by
		\begin{equation} \label{definition union partition}
			\mu^{(1)} \cup \mu^{(2)} :=(\mu^{(1)}_1, \dots, \mu^{(1)}_{N_1}, qN_1 + b_{N_1}^{(1)} +\mu^{(2)}_1, \dots, qN_1+ b_{N_1}^{(1)}+  \mu^{(2)}_{N_2}).
		\end{equation}
\end{defn}
\begin{figure}[h]
	\begin{center}
		\includegraphics[scale=0.4]{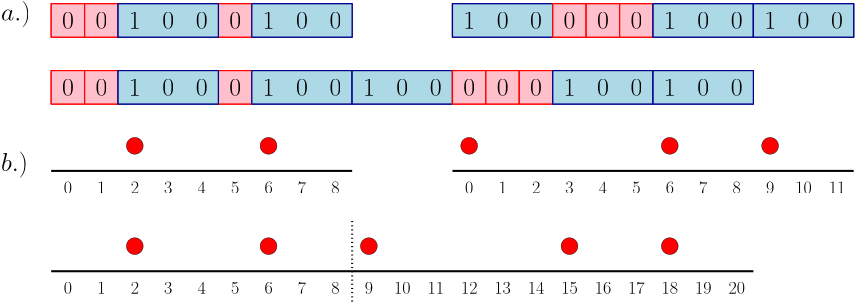}
	\end{center}
	\caption{The concatenation of the root partitions $\lambda^{(3)}_2( b^{(1)}) = ( 2,6)$ corresponding to $ b^{(1)} = (2,3) $ and $\lambda^{(3)}_3( b^{(2)})  (0,6,9)$  and $ b^{(2)} = (0,3,3) $ in $a.)$ the tiling picture, $b.)$ the occupation picture. The concatenated root is $ \lambda^{(3)}_5( b) = (2,6,9,15,18) $ with $ b= ( b^{(1)}, b^{(2)}) = (2,3,3,6,6)$. }\label{fig:concatenation}
\end{figure}
		For example, we have $ \lambda^{(q)}_{N_1}\big(b^{(1)}\big) \cup \lambda^{(q)}_{N_2}\big(b^{(2)}\big)= \lambda^{(q)}_{N_1+N_2}\big(b^{(1)},b^{(2)}\big)$ where the concatenated partition $ \big(b^{(1)},b^{(2)}\big)  $ was defined in~\eqref{eq:bconcat}, cf.~Figure~\ref{fig:concatenation}. 
		Note that concatenation is generally non-commutative.  In the language of partitions, the renewal structure is captured in the following
\begin{defn}\label{def:irred}
			Let $ N \geq 2 $ and $b\in \mathbb{N}_0^N $ be a partition. We say that a partition $\mu\preceq \lambda^{(q)}_N(b)$ has a \emph{renewal point} at $ N_1 \in [1, N-1] $ if
			there exists 
			$\mu^{(1)} \preceq \lambda^{(q)}_{N_1}\big(b^{(1)}\big) $, $ \mu^{(2)}\preceq \lambda^{(q)}_{N-N_1}\big(b^{(2)}\big) $ with partitions $ b^{(1)} \in  \mathbb{N}_0^{N_1} $, $ b^{(2)} \in  \mathbb{N}_0^{N-N_1} $  such that 
			$$ 
			\mu= \mu^{(1)} \cup \mu^{(2)} \quad \mbox{ and }\quad  b = \big(b^{(1)},b^{(2)}\big) . 
			$$
			We say that $\mu \preceq \lambda^{(q)}_N(b)$ is irreducible if it has no renewal point in $[1, N-1]$ and reducible otherwise. 
\end{defn}

There is a simple characterization of when a renewal point occurs. 
\begin{lemma} \label{lm:splitting}
	Let $(N,b)$ be a root and $\mu\preceq \lambda_N^{(q)}(b)$. Then $\mu$ has a renewal point at $s\in \{1, \dots, N-1\}$ if and only if $\sum_{j=1}^s \mu_j = \sum_{j=1}^s (\lambda_N^{(q)}(b))_j$.
	\end{lemma}
\begin{proof}
	If $s$ is a renewal point, then we can write $\mu=\mu^{(1)}\cup \mu^{(2)}$ with  $\mu^{(1)} \preceq \lambda^{(q)}_{s}\big(b^{(1)}\big)$, $\mu^{(2)} \preceq \lambda^{(q)}_{N-s}\big(b^{(2)}\big)$ and $b=(b^{(1)}, b^{(2)})$. Thus, $\sum_{j=1}^s \mu_j = \sum_{j=1}^s \mu^{(1)}_1 = \sum_{j=1}^s (\lambda_s^{(q)}(b^{(1)}))_j = \sum_{j=1}^s (\lambda_N^{(q)}(b))_j$. 
	To establish the reverse implication, we define $\mu^{(1)} = (\mu_1, \dots, \mu_s)$ and $\mu^{(2)} = (\mu_{s+1}-qs-b_s, \mu_{s+2}-qs-b_s, \dots, \mu_N -qs-b_s)$. We then have for any $k\in \{1, \dots, N-s\}$ thanks to the domination order  $\mu\preceq \lambda_N^{(q)}(b)$:
	\begin{align*}
		\sum_{j=k}^{N-s} \mu^{(2)}_j = \sum_{j=k}^{N-s} (\mu_{s+j}-qs-b_s)
		= \sum_{j=s+k}^N (\mu_j -qs-b_s) \leq \sum_{j=s+k}^N ((\lambda_N^{(q)}(b))_j -qs-b_s) = \sum_{j=k}^{N-s} (\lambda_{N-s}^{(q)}(b^{(2)}))_j.
	\end{align*}
	If $\sum_{j=1}^s \mu_j = \sum_{j=1}^s \lambda_N^{(q)}(b)_j$, then $\mu\preceq \lambda_N^{(q)}(b)$ implies that $\sum_{j=s+1}^N \mu_j = \sum_{j=s+1}^N \lambda_N^{(q)}(b)_j$. Hence, for any $k\in \{1, \dots, s\}$ we arrive at:
	\begin{align*}
		\sum_{j=k}^s \mu^{(1)}_j = \sum_{j=k}^N \mu_j - \sum_{j=s+1}^N \lambda_N^{(q)}(b)_j \leq \sum_{j=k}^s \lambda_N^{(q)}(b)_j =  \sum_{j=k}^s \lambda_s^{(q)}(b^{(1)})_j.
	\end{align*}
	It is hence straightforward to check that $\mu^{(1)}\preceq \lambda_s^{(q)}(b^{(1)})$.
	\end{proof}

The concatenation of partitions is tailored such that the operator product  $\mathbb{W}(\mu,b) $ involving a partition $\mu =  \mu^{(1)} \cup \mu^{(2)} \preceq \lambda^{(q)}_N(b) $ with a renewal point at $ N_1 < N $ is multiplicative. This is proven in the following theorem, which also establishes the crucial factorization property of the geometry-independent part $ w_b(\mu ) $ of the coefficient in the expansion of fractional quantum Hall wavefunctions, cf.~\eqref{MPS with monomial symmetric_c}. 
\begin{theorem}\label{thm:factorizationw}
Let $ \mu= \mu^{(1)} \cup \mu^{(2)}  $ with $\mu^{(1)} \preceq \lambda^{(q)}_{N_1}\big(b^{(1)}\big)$ and $\mu^{(2)}\preceq \lambda^{(q)}_{N_2}\big(b^{(2)}\big)$ and partitions $ b^{(1)} \in \mathbb{N}_0^{N_1} $, $ b^{(2)} \in \mathbb{N}_0^{N_2} $. 
\begin{enumerate}
\item  For any permutation $ \tau \in \mathcal{S}_N $:
\begin{equation}\label{eq:nullity}
\langle 0 \vert \mathbb{W}(\mu,b_\tau) \vert 0 \rangle = 0  ,\quad 
\mbox{unless} \quad \sum_{j=1}^{N_1} b_{\tau(j)} =  \sum_{j=1}^{N_1} b_{j} = \vert b^{(1)} \vert .
\end{equation}
In the latter case, if $ \tau = \big(\tau^{(1)}, \tau^{(2)} \big) $ with $ \tau^{(1)} \in \mathcal{S}_{N_1} $ and $  \tau^{(2)} \in \mathcal{S}_{N_2} $, then 
\begin{equation}\label{eq:frstprodrule}
\langle 0 \vert \mathbb{W}(\mu,b_\tau) \vert 0 \rangle = \langle 0 \vert \mathbb{W}\big(\mu^{(2)} ,b^{(2)}_{\tau^{(2)}}\big) \vert 0 \rangle  \langle 0 \vert \mathbb{W}\big(\mu^{(1)} ,b^{(1)}_{\tau^{(1)}}\big) \vert 0 \rangle
\end{equation}
\item $ \displaystyle  w_b(\mu) = w_{b^{(2)}}\big(\mu^{(2)} \big) w_{b^{(1)}}\big(\mu^{(1)} \big) $, where $b = \big(b^{(1)},b^{(2)}\big) $.
\end{enumerate}
\end{theorem}
\begin{proof}
1.~If $ \mu^{(1)}\preceq \lambda^{(q)}_{N_1}\big(b^{(1)}\big)$ and $b = \big(b^{(1)},b^{(2)}\big) $, then 
$$
0 = \sum_{j=1}^{N_1} \left(  \mu_j^{(1)} - q(j-1) -b_j\right) \geq \sum_{j=1}^{N_1} \left(  \mu_j^{(1)} - q(j-1) -b_{\tau(j)} \right) .
$$ 
The inequality holds for all $   \tau \in \mathcal{S}_N  $, because $ j \mapsto b_j $ is non-decreasing.  In case that inequality is strict, then $ W_{\mu_{N_1}-q(N_1-1)-b_{\tau(N_1)}} \dots W_{\lambda_2 - q-b_{\tau(2)} }W_{\mu_1-b_{\tau(1)}} \vert 0\rangle= 0 $ by Theorem~\ref{vanishing condition} $(4)$. This finishes the proof of~\eqref{eq:nullity}. 

For a proof of the remaining assertion, we again use Theorem~\ref{vanishing condition} $(2)$, which also guarantees that $ \sum_{j=1}^{N_1} b_{\tau(j)} =  \vert b^{(1)} \vert $ implies
$$
 W_{\mu_{N_1}-q(N_1-1)-b_{\tau(N_1)}} \dots W_{\lambda_2 - q-b_{\tau(2)}} W_{\mu_1-b_{\tau(1)}} \vert 0 \rangle     =  \vert 0 \rangle \  \langle 0 \vert W_{\mu_{N_1}-q(N_1-1)-b_{\tau(N_1)}} \dots W_{\lambda_2 - q-b_{\tau(2)}} W_{\mu_1-b_{\tau(1)}} \vert 0 \rangle .
$$
This finishes the proof of~\eqref{eq:frstprodrule}. \\
\noindent
2.~According to the first assertion, the non-zero terms in the sum
\begin{equation}\label{eq:prooffact}
 w_b(\mu) = \frac{1}{M(b)!} \sum_{\tau \in  \mathcal{S}_{N}} \langle 0 \vert \mathbb{W}(\mu,b_\tau) \vert 0 \rangle
\end{equation}
satisfy $ \sum_{j=1}^{N_1} b_{\tau(j)} =  \vert b^{(1)} \vert $. In case $ b_{N_1} < b_{N_1+1} $ this implies that $ \tau\{1,\dots, N_1\} = \{1,\dots, N_1\}  $, i.e.\ $ \tau = \big(\tau^{(1)}, \tau^{(2)} \big) $ with $ \tau^{(1)} \in \mathcal{S}_{N_1} $ and $  \tau^{(2)} \in \mathcal{S}_{N_2} $. In this case, the factorization $ M(b)! = M(b^{(1)})! \ M(b^{(2)})! $ of the occupation numbers hence implies the claim. 
In case $ b_{N_1} = b_{N_1+1} $, the number $ m(b,b_{N_1}) $ of repetitions of $ b_{N_1} $ is larger than one. It is additive among the two parts $  \big(b^{(1)},b^{(2)}\big) $, which compose $ b $, i,e,
$$
m(b,b_{N_1}) =:  m  = m^{(1)} + m^{(2)} \quad \mbox{with} \quad m^{(1)} \coloneqq m(b^{(1)},b_{N_1}) , \;  m^{(2)} \coloneqq m(b^{(2)}, b_{N_1+1}  - b_{N_1}) . 
$$
Since there are exactly $ \binom{m}{m^{(1)}} $ ways of choosing $ m^{(1)} $ elements out of a set of $ m $ elements, the binomial factor in 
$$ m! =   \binom{m}{m^{(1)}}\ m^{(1)} ! \  m^{(2)} ! $$ 
included in the prefactor in~\eqref{eq:prooffact} exactly matches the overcounting of permutations in that sum over non-zero terms compared to the case that $ \tau =  \big(\tau^{(1)}, \tau^{(2)} \big) $ with $ \tau^{(1)} \in \mathcal{S}_{N_1} $ and $  \tau^{(2)} \in \mathcal{S}_{N_2} $. In the latter case, one may use the factorization~\eqref{eq:frstprodrule}, which completes the proof. 
\end{proof}

For the decomposition of the fractional quantum Hall wavefunctions into sums of products of irreducible parts, in addition to the above factorization result of the geometry-independent part $ w_b $ of the expansion coefficients,  we also need to ensure that their geometric part factorizes. 
In the cylinder geometry, this is equivalent to the additivity of the quantity $ \Delta_b $ from~\eqref{definition Delta} under concatenation.  Moreover, as has already been
observed in~\cite{jansen2009symmetry}, the combinatorial weight of irreducible partitions is extensive. 
		\begin{lemma}\label{lem:Deltafact}
		\begin{enumerate}
		\item 
			If $\mu^{(1)} \preceq \lambda^{(q)}_{N_1}(b^{(1)}) $ and $\mu^{(2)}\preceq \lambda^{(q)}_{N_2}(b^{(2)}) $ with partitions $ b^{(1)} \in \mathbb{N}_0^{N_1} $, $ b^{(2)} \in \mathbb{N}_0^{N_2} $, then $\mu^{(1)} \cup \mu^{(2)}\preceq \lambda^{(q)}_{N_1+N_2}(b)$ with $ b = \big(b^{(1)},b^{(2)}\big) $ from~\eqref{eq:bconcat}, and we have
			\begin{equation} \label{Delta and factorization}
				\Delta_b(\mu^{(1)}\cup \mu^{(2)}) = \Delta_{b^{(1)}}(\mu^{(1)}) + \Delta_{b^{(2)}}(\mu^{(2)}).
			\end{equation}
		\item
			Let $\mu \preceq \lambda^{(q)}_N(b)$ be irreducible, then 
			\begin{equation} \label{lower bound Delta irreducible}
				\Delta_b(\mu) \geq q(N-1) + b_N - b_1 =: D_{b,N} .
			\end{equation}
		\end{enumerate}
		\end{lemma}
		\begin{proof}
			1.~The first claim follows from a straightforward computation using the definition of dominance order.		
			For \eqref{Delta and factorization} we compute
			\begin{align*}
				\Delta_b(\mu^{(1)}\cup \mu^{(2)})
				=  & \sum_{j=1}^{N_1+N_2} \left((q (j-1) + b_j )^2 - (\mu^{(1)}\cup \mu^{(2)})_j^2\right) \\
				=  &\sum_{j=1}^{N_1} \left((q (j-1) + b_j )^2 - (\mu^{(1)}_j)^2\right) \\
				& + \sum_{j=N_1+1}^{N_1+N_2} \left(q (j-1) + b_{N_1}^{(1)} + b_{j-N_1}^{(2)} )^2-(qN_1+b_{N_1}^{(1)} +\mu^{(2)}_{j-N_1} )^2\right) \\
				= \ &  \Delta_{b^{(1)}}(\mu^{(1)}) + \sum_{j=1}^{N_2} \left((q(N_1+j-1 )+b_{N_1}^{(1)}+ b_{j}^{(2)} )^2 -(qN_1+b_{N_1}^{(1)}+\mu^{(2)}_{j} )^2\right) \\
			= \ &  \Delta_{b^{(1)}}(\mu^{(1)}) +  \Delta_{b^{(2)}}(\mu^{(2)}) +  2( q N_1 + b_{N_1}^{(1)} ) \sum_{j=1}^{N_2} \left( q(j-1) +b_{j}^{(2)}  -\mu_j^{(2)}\right) \\
				= &  \Delta_{b^{(1)}}(\mu^{(1)}) +  \Delta_{b^{(2)}}(\mu^{(2)}) ,
			\end{align*}
		which completes the proof of the first item.\\
		\noindent
		2. In case $ b = 0 $, this has already been established within the proof of~\cite[Lemma 5]{jansen2009symmetry} (based on an argument from \cite{di1994laughlin}), which we repeat here for the general case.
		If $\mu \preceq \lambda^{(q)}_N(b)$, then $ \delta_s \coloneqq \sum_{j=1}^s \left( \mu_j - q (j-1) - b_j\right) \geq 0 $ for all $ s \in \{1, \dots , N \} $. Moreover, since $ \mu $ is irreducible, only $ \delta_N = 0 $ and $ \delta_s \geq 1 $ for all $ 1 \leq s \leq N-1 $. Setting $ \delta_0 \coloneqq 0 $, we may write $ \mu_j = q (j-1) + b_j + \delta_j - \delta_{j-1} $. Inserting this expression into the definition of $ \Delta_b(\mu)  $, we conclude
		\begin{align*}
		 \Delta_b(\mu) & = \sum_{j=1}^N \left( q(j-1) + b_j +\mu_j \right) \left( \delta_{j-1} - \delta_j \right) \\
		 & = \sum_{j=1}^{N-1}  \left( q + b_{j+1} - b_j  +\mu_{j+1}-\mu_j \right) \delta_j + \left( q (N-1) + b_N + \mu_N \right) \delta_N \\
		 & \geq q (N-1) + b_N - b_1 ,
		\end{align*} 
		where the last step used $\delta_N = 0 $ and $ \delta_j \geq 1 $ for all $ 1 \leq j \leq N-1 $. 
		\end{proof}

\subsection{Estimating irreducible contributions}

The following is the key estimate of the geometry-independent coefficient of fractional quantum Hall wavefunctions. 
\begin{lemma}
	For any root $ (N,b) $ and all partitions $ \lambda\preceq \lambda^{(q)}_N(b)$:
	\begin{equation}\label{eq:wbest}
		\left| w_b(\lambda) \right| \leq \frac{ \exp\left( (4q+1) \Gamma_b( \lambda) \right) }{\left(1- \exp\left[-(4q+1)\right] \right)^{N-1}} . 
	\end{equation}
\end{lemma}
\begin{proof}
	We combine the key estimate from Corollary~\ref{cor:keyest} with the representation~\eqref{eq:prooffact} to obtain
	\begin{equation}
		\left| w_b(\lambda) \right| \leq  \frac{1}{M(b)!} \sum_{\tau \in  \mathcal{S}_{N}} 
		\exp\left( (4q+1) \Gamma_{b_\tau}(\lambda) \right) ,
	\end{equation}
	where $b_\tau  = (b_{\tau(1)} , \cdots , b_{\tau(N)}) $. To estimate the sum, we isolate the contribution involving the permutation $ \tau $:
	$$
	\Gamma_{b_\tau}(\lambda ) =  \Gamma_b(\lambda) + \sum_{j=1}^N j \left( b_{\tau(j)} - b_j \right) .
	$$
	In the following, we abbreviate
	\begin{align*}
		\mathcal{S}_M(\tau) \coloneqq (4q+1)\sum_{j=1}^{M} j (b_{\tau(j)}-b_j)
	\end{align*}
	and define $M_j = M(b_1, \dots, b_j)$. Our goal is to prove for $ N \in \mathbb{N} $ the recursive bound 
	\begin{equation} \label{recursion}
		\frac{1}{M_{N+1}!} \sum_{\tau\in S_{N+1}} \exp\left(\mathcal{S}_{N+1}(\tau)\right) 
		\leq \frac{1}{1-e^{-(4q+1)}} \frac{1}{M_{N}!} \sum_{\sigma\in S_N} \exp\left(\mathcal{S}_N(\sigma)\right).
	\end{equation}
	For a proof, we first note that we can rewrite the left side:
	\begin{align*}
		\frac{1}{M_{N+1}!} \sum_{\tau\in S_{N+1}} \exp\left(\mathcal{S}_{N+1}(\tau)\right) 
		= 	\frac{1}{M_{N+1}!} \sum_{k=1}^{N+1} \sum_{\substack{\tau\in S_{N+1} \\ \tau(k)=N+1}} \exp\left(\mathcal{S}_{N+1}(\tau)\right).
	\end{align*}
	In order to reduce to $S_N$ we want to use the change of variables $Q_k: S_{N+1} \rightarrow S_{N+1}, \tau \mapsto \tau \circ (N+1,k)$, where $(N+1,k)$ denotes the transposition which swaps $N+1$ and $k$ and fixes everything else. If  $\tau\in S_{N+1}$ has the property $\tau(k)=N+1$, then $Q_k(\tau)(N+1)=N+1$, which is effectively a permutation in $S_N$. To make use of the change of variable $Q_k$, we need to express everything in terms of $Q_k(\tau)$. 

	We write
	\begin{align*}
		\sum_{j=1}^{N+1} j (b_{\tau(j)}-b_j)
		= \sum_{j=1}^N j (b_{Q_k(\tau)(j)}-b_j) -(N+1-k) (b_{N+1}-b_{Q_k(\tau)(k)})
	\end{align*}
	and change variables in
	\begin{align*}
		&\frac{1}{M_{N+1}!} \sum_{\tau\in S_{N+1}} \exp\left(\mathcal{S}_{N+1}(\tau)\right) \\
		&=\frac{1}{M_{N+1}!} \sum_{k=1}^{N+1} \sum_{\substack{\tau\in S_{N+1} \\ \tau(N+1)=N+1}} \exp\left(\mathcal{S}_{N+1}(\tau)\right) \exp\left(-(4q+1)(N+1-k) (b_{N+1}-b_{\tau(k)})\right) \\
		&=\frac{1}{M_{N+1}!} \sum_{\sigma\in S_N} \exp\left(\mathcal{S}_N(\sigma)\right)\left( \sum_{k=1}^{N+1} \exp\left(-(4q+1)(N+1-k)(b_{N+1}-b_k)\right)\right).
	\end{align*}
	Abbreviating the set $[N+1]\coloneqq \{k\in \{1, \dots, N+1\} \ : \ b_k=b_{N+1}\}$, 
	 there are exactly $\vert [N+1]\vert$ elements $k\in \{1, \dots, N+1\}$ for which $b_k= b_{N+1}$ and for the remaining $k$ we have $b_{N+1}-b_k\geq 1$. Thus, we may bound 
	\begin{align*}
		\sum_{k=1}^{N+1} \exp\left(-(4q+1)(N+1-k)(b_{N+1}-b_k)\right)
		\leq \vert [N+1]\vert + \sum_{k=1}^{N+1} e^{-(4q+1)k}
		\leq  \frac{\vert [N+1]\vert}{1-e^{-(4q+1)}}.
	\end{align*}
	Using $\vert [N+1]\vert/(M_{N+1}!)=1/(M_N!)$, we arrive at \eqref{recursion}, which finishes the proof.
\end{proof} 

In the cylinder geometry, the above estimate allows us to show that the contribution of the irreducible partitions to the norm of the fractional quantum Hall wavefunctions decays exponentially for sufficiently small cylinders. The decay rate 
$
C_q(\gamma)  = (\gamma^2 - 2(4q+1)) -2   q^{-1} \ln  (1- e^{-(4q+1)} )  -  \pi \sqrt{\frac{2}{3}}
$ from \eqref{def:Cgamma}
 is strictly positive for sufficiently large $ \gamma $. 
\begin{cor}
	For any root $ (N,b) $:
	\begin{equation}\label{eq:normirrep}
		\big\| \widehat \Psi_{b,N} \big\|^2   = \sum_{\substack{ \lambda \preceq \lambda_N^{(q)}(b)  \\ \lambda \; \irr }} \frac{\big| g_b^{(c)}(\lambda) \ w_b(\lambda) \big|^2}{M(\lambda)!} \leq \exp\left( - C_q(\gamma) D_{b,N}  \right) .
	\end{equation}
	Moreover,  $ \alpha_1(b) = 1 $ for any $ b \in \mathbb{N}_0 $. 
\end{cor} 
\begin{proof}
	We use \eqref{eq:wbest} with the abbreviation $ c_q \coloneqq - 2\ln  (1- e^{-(4q+1)} ) > 0 $ together with~\eqref{Gamma bounded by Delta} and the trival bound $ M(\lambda)! \geq 1 $ to estimate
	\begin{align*}
		\big\| \widehat \Psi_{b,N} \big\|^2  & \leq  \sum_{\substack{ \lambda \preceq \lambda_N^{(q)}(b)  \\ \lambda \; \irr } }  \exp\left( - \left( \gamma^2 - 2(4q+1) \right) \Delta_b(\lambda)  + (N-1) c_q  \right) \\
		& \leq   \sum_{\lambda \preceq \lambda_N^{(q)}(b) }  \exp\left( - \left( \gamma^2 - 2(4q+1) \right)\left( q(N-1) + b_N - b_1\right)) + (N-1) c_q  \right)  . 
	\end{align*}
	The last line is by~\eqref{lower bound Delta irreducible}. 
	To estimate the number of partitions  $\lambda\preceq \lambda_N^{(q)}(N)$, we first note that the number $ b_1 $ can be ignored. This is most easily see by noting that $\tilde{\lambda}=\lambda-(b_1, \dots, b_1)$ is a partition as well. Indeed, using the alternative definition of dominance order, \eqref{dominance order def 2}, we get that $\lambda_1\geq (\lambda_N^{(q)}(b))_1=b_1$. Thus, we may hence use the number of partitions of $ |\tilde{\lambda} | = \sum_{j=1}^N ( \lambda_N^{(q)}(b)_j -b_1) = qN(N-1)/2+ \sum_{j=1}^N (b_N-b_1) \leq qN(N-1)/2+(N-1)(b_N-b_1) $ as an upper bound on
	\begin{align*}
		\sum_{\lambda\preceq \lambda_N^{(q)}(N)} 1
		\leq p\left(qN(N-1)/2+(N-1)(b_N-b_1)\right).
		\end{align*}
	Using the bound~\eqref{partition function pointwise bound} on the number of partitions $ p $, as well as the elementary estimates $N(N-2)/2\leq (N-1)^2$ for $N\in \mathbb{N}$ and $\sqrt{a+b}\leq \sqrt{a}+\sqrt{b}$ and $2ab \leq a^2+b^2$, we arrive at
	\begin{align*}
		\sum_{\lambda\preceq \lambda_N^{(q)}(N)} 1 
		&\leq \exp\left(\pi \sqrt{\frac{2}{3}}  \left[\left(\sqrt{q}+\frac{1}{2}\right)(N-1)+\frac{b_N-b_1}{2}\right]\right) \leq \exp\left( \pi \sqrt{\frac{2}{3}} \left(q(N-1)+b_N-b_1\right)\right), 
		\end{align*}
	 which completes the proof of~\eqref{eq:normirrep}.
\end{proof}

\section{Implications for cylinder geometry}

\subsection{Proof of the factorization}\label{sec:factor}
\begin{proof}[Proof of Theorem~\ref{thm:fact}]
We start from Corollary~\ref{cor:planar} in which  $  \Psi_{b,N} $ is represented as a sum over partitions $\lambda \preceq \lambda^{(q)}_{N}(b) $. Conditioning on the first renewal point $ N_1 \in \{1, \dots , N\} $ in the partition $ \lambda $, we use 
\begin{itemize}
\item the factorization of the expansion coefficients for $ \lambda = \lambda^{(1)} \cup \lambda^{(2)}  $ with $ \lambda^{(j)} \preceq   \lambda_{N_j}^{(q)}\big(b^{(j)}\big) $, $ j = 1,2 $, and $ N_2 = N - N_1 $:
$$
h_b\big( \lambda^{(1)} \cup \lambda^{(2)} \big) = h_b\big( \lambda^{(1)}\big) \ h_b\big( \lambda^{(2)} \big) .
$$
This follows from the factorization of $ g_b $ established in Lemma~\ref{lem:Deltafact}, the factorization of $ w_b $ from Theorem~\ref{thm:factorizationw}, as well as the factorization of $ M(\lambda)! $. 
\item the factorization~\eqref{eq:Slaterfact} of the Slater determinants/permanents. 
\end{itemize}
Separating the contribution without a renewal point and summing over the location $ N_1 \in \{1, \dots , N\} $ of the first renewal point in the remainder, we hence obtain
\begin{equation}\label{eq:renewal}
	 \Psi_{b,N}  =  \widehat \Psi_{b,N} + \sum_{N_1=1}^{N-1}  \widehat \Psi_{b^{(1)},N_1} \odot  \Psi_{b^{(2)} ,N-N_1} 
\end{equation}
where $  (b^{(1)}, b^{(2)} ) $ is the decomposition of $ b $ corresponding to $ N_1$, $ N-N_1 $ with $ b^{(1)} \in \mathbb{N}_0^{N_1} $. Repeating this procedure until the last renewal point establishes~\eqref{eq:fact}. 

To prove the claimed orthogonality of the expansion, we take two functions with different parameters corresponding to 
$$ N = \sum_{j=1}^r N_j = \sum_{k=1}^s \widetilde N_k \quad \mbox{and} \quad b= \big(b^{(1)} , \dots , b^{(r)} \big)  =  \big(\widetilde b^{(1)} , \dots , \widetilde b^{(s)} \big) , $$ 
and compute their scalar product
\begin{align*}
& \left\langle  \widehat \Psi_{b^{(1)},N_1} \odot  \widehat \Psi_{b^{(2)},N_2} \odot \cdots \odot  \widehat \Psi_{b^{(r)},N_r}  ,   \widehat \Psi_{\widetilde b^{(1)},\widetilde N_1} \odot  \widehat \Psi_{\widetilde b^{(2)},\widetilde N_2} \odot \cdots \odot  \widehat \Psi_{\widetilde b^{(s)},\widetilde N_s} \right\rangle \\
& =\mkern-10mu  \sum_{\substack{ \lambda^{(j)} \preceq  \lambda_{N_j}^{(q)}\big(b^{(j)}\big) \\  j \in \{ 1, \dots , r \}  \\ \lambda^{(j)} \irr}}   \sum_{\substack{ \mu^{(j)} \preceq  \lambda_{\widetilde N_k}^{(q)}\big(\widetilde b^{(k)}\big) \\  k \in \{ 1, \dots , s \}  \\ \mu^{(j)} \irr}}  \mkern-20mu \overline{h_b\big( \lambda^{(1)} \cup \dots \cup  \lambda^{(r)} \big)} h_b\big( \mu^{(1)} \cup \dots \cup  \mu^{(s)} \big) \times  \delta_{ \lambda^{(1)} \cup \dots \cup  \lambda^{(r)} , \mu^{(1)} \cup \dots \cup  \mu^{(s)}}  .
\end{align*}
The Kronecker delta resulted from the orthonormality of the $ \odot $ product~\eqref{eq:productSlater} of the occupation basis. The claim will follow from the fact that none of the Kronecker deltas is one. 

As the $\lambda^{(j)}$ are irreducible, Lemma~\ref{lm:splitting} ensures that the renewal points of $\lambda^{(1)}\cup \dots \cup \lambda^{(r)}$ are $N_1, N_1+N_2, \dots, N-N_r$. Similarly, the renewal points of $\mu^{(1)}\cup \dots \cup \mu^{(s)}$ are given by $\tilde{N}_1, \tilde{N}_1+\tilde{N}_2, \dots, N-\tilde{N}_s$. Thus, if the two partitions agree, we have that the renewal points agree and therefore $(N_1, \dots, N_r)=(\tilde{N}_1, \dots, \tilde{N}_s)$. However, we have assumed that the parameters are different, which finishes the proof of the claim.

Finally, the estimate~\eqref{eq:normseg} is just~\eqref{eq:normirrep}.
\end{proof}

\subsection{Proof of the entanglement gap}\label{sec:Egap}

\begin{proof}[Proof of Theorem~\ref{thm:EGap}] 
The triangle inequality shows that the claim follows from the bound 
\begin{equation}\label{eq:eq:EGap2}
\left\| \Psi_{b,N} -  \Psi_{b^{(1)},N_1}  \odot \Psi_{b^{(2)},N_2} \right\|^2 \leq  \frac{e^{-C_q(\gamma)q}}{\left(1-e^{-C_q(\gamma)q}\right)^2}\Vert \Psi_{b,N}\Vert^2. 
\end{equation}
For its proof, we use Theorem~\ref{thm:fact} to represent $  \Psi_{b,N} $ as a sum over products of irreducible segments separated by renewal points. The term $ \Psi_{b^{(1)},N_1}  \odot \Psi_{b^{(2)},N_2}  $ gathers all contributions which have a renewal point exactly at $ N_1 $. By the orthogonality of the expansion~\eqref{eq:fact}, the left side in~\eqref{eq:eq:EGap2} hence equals 
\begin{equation}\label{eq:entproof}
\sum_{r=1}^N \sum_{ \substack{M_1, \dots , M_r \in \mathbb{N}\\  \sum_{j=1}^r M_j = N } } 1[ \mbox{For all $ j \in \{ 1, \dots , r\}$:} \; \sum_{k=1}^j M_k \neq N_1 ] \ \prod_{j=1}^r  \left\|   \widehat \Psi_{b^{(j)},N_j} \right\|^2.
\end{equation}
The constraint of the indicator function can be satisfied in several ways. One option would be $M_1=N$, the only term present in \eqref{eq:entproof} in this case would be $\Vert \widehat{\Psi}_{b,N}\Vert^2$. If we fix $M_1$ with $N_1<M_1<N$, then we can resum all of those contributions and obtain the term $\Vert \widehat{\Psi}_{b^{(1)}, M_1} \odot \Psi_{b^{(2)}, N-M_1}\Vert^2$.   
Resuming the configuration, which have the last renewal point before $N_1$ at $L_1<N_1$ and the next renewal point after $L_1$ at $L_2>N_1$, yields $\Vert \Psi_{b^{(1)}, L_1} \odot \widehat{\Psi}_{b^{(2)}, L_2-L_1} \odot \Psi_{b^{(3)}, N-L_2} \Vert^2$. Thus, we can rewrite \eqref{eq:entproof} as
\begin{align*}
	\sum_{M_2=2}^N \mkern-5mu \sum_{\substack{M_1, M_3\in \mathbb{N}_0 \\ M_1+M_2+M_3=N}} \mkern-15mu 1\big[N_1\in [M_1+1, M_1+M_2-1]\big] \;  \Vert \Psi_{b^{(1)},M_1} \odot \widehat{\Psi}_{b^{(2)}, M_2} \odot \Psi_{b^{(3)},M_3} \Vert^2,
\end{align*}
with the convention that $ \Psi_{b^{(1)},0} \odot \widehat{\Psi}_{b^{(2)}, M_2} \odot \Psi_{b^{(3)},M_3} =  \widehat{\Psi}_{b^{(2)}, M_2} \odot \Psi_{b^{(3)},M_3}$ and similarly, when $M_3=0$. Using the estimate  \eqref{eq:normirrep} on the norm of the irreducible parts and the fact \eqref{eq:supermult} that the norms of the wavefunctions are supermultiplicative, we arrive at
\begin{align*}
	\Vert \Psi_{b^{(1)},M_1} \odot \widehat{\Psi}_{b^{(2)}, M_2} \odot \Psi_{b^{(3)},M_3} \Vert^2 & = \Vert \Psi_{b^{(1)},M_1} \Vert^2 \ \Vert  \widehat{\Psi}_{b^{(2)}, M_2} \Vert^2 \ \Vert \Psi_{b^{(3)},M_3} \Vert^2 \\
	& \leq \Vert \Psi_{b,N}\Vert^2 \ e^{-C_q(\gamma)D_{b^{(2)}, M_2}} .
\end{align*}
For fixed $M_2$ there are at most $M_2-1$ choice for $(M_1, M_2, M_3)$ such that $N_1\in [M_1+1, M_1+M_2-1]$. Furthermore, $D_{b^{(2)}, M_2}\geq q(M_2-1)$. Thus, we arrive at 
\begin{align*}
	&\Vert \Psi_{b,N}-\Psi_{b^{(1)},N_1}  \odot \Psi_{b^{(2)},N_2} \Vert^2
	\leq \left( \sum_{M_2=1}^{\infty} M_2 e^{-C_q(\gamma)q M_2}\right) \Vert \Psi_{b,N}\Vert^2
	= \frac{e^{-C_q(\gamma)q}}{\left(1-e^{-C_q(\gamma)q}\right)^2} \Vert \Psi_{b,N}\Vert^2,
\end{align*}
which concludes the proof of~\eqref{eq:eq:EGap2}.
\end{proof}

\subsection{Proof of the exponential clustering}\label{sec:cluster}
The proof of exponential clustering, Theorem~\ref{thm:clustering}, proceeds in three steps. We first approximate the given wavefunction $ \Psi_{b,N} $ by a truncation, whose factor representation of Theorem~\eqref{thm:fact} contains only irreducible segments intersecting the interval covering the support of the two observables. We then show exponential clustering for the truncated expectation value. The proof is concluded by a series of auxiliary results, in particular, the bounds on the expectation values in the bosonic case.

\subsubsection{Truncation}	We introduce the truncated versions of the fractional quantum Hall wavefunctions. The main idea is to allow only short segments in the representation~\eqref{eq:fact}. More precisely, 
for a given interval $I = \{a, \dots, b\}=[a,b]\cap \{1, \dots, N\}$ on the line of $ N $ monomers,  we discard from~\eqref{eq:fact} all segments, which extend over more than exactly one monomer and which contain particles indexed by an element in $I$, in case they have length as measured by~\eqref{lower bound Delta irreducible} greater or equal to $l $. %
\begin{defn}
For a root $ (N,b) $ and  $ l \in \mathbb{N} $ and  an interval $I \subseteq \{1, \dots, N\}$, we abbreviate by $\mathcal{M}_{ b, N}(I,l )$ the set of segmentations, $(N_1, \dots, N_r) \in \mathbb{N}^r $,  $\sum_{j=1}^r N_j=N$, which have the property that for all $ j \geq 1 $:
\begin{align*}
	\left[ \sum_{k=1}^{j-1} N_k, \sum_{k=1}^j N_k\right] \cap I \neq \emptyset \quad \Rightarrow \quad D_{b^{(j)}, N_j} = q(N_j-1)+ b^{(j)}_{\ell(b^{(j)})}-b^{(j)}_1 < l .
\end{align*}
(Recall that $ \ell(b^{(j)}) $ is the length of the partition $ b^{(j)} $ which is the $j $th segment of $ b $ associated to the segmentation $(N_1, \dots, N_r)$.)
The corresponding truncated wavefunction  is
\begin{equation}\label{eq:truncfact}
	\Psi_{b,N}^{(I, l)} \coloneqq \sum_{r=1}^{N} \sum_{ \substack{N_1, \dots , N_r \in \mathbb{N}\\  \sum_{j=1}^r N_j = N } } \mkern-10mu 1[(N_1, \dots, N_r)\in \mathcal{M}_{ b,N}(I, l )] \; \widehat \Psi_{b^{(1)},N_1} \odot   \cdots \odot  \widehat \Psi_{b^{(r)},N_r}.
\end{equation}
\end{defn}
Any fractional quantum Hall wavefunction is close to its truncated counterpart. 
\begin{lemma} \label{lm:truncated Laughlin}
	Suppose that $ \gamma $ is large enough such that $ C_q(\gamma) > 0 $. 
	There exist a constant $C>0$, which  depends only on $\gamma$ and $q$, such that for any root $ (N,b) $ and interval $I \subseteq \{1, \dots, N\}$ and all $l>0$ 
	we have
	\begin{equation}
		\Vert \Psi_{b,N} - \Psi_{b,N}^{(I, l)} \Vert \leq C\sqrt{|I|} \ \sqrt{l^2+1} e^{-C_q(\gamma) l/2} \Vert \Psi_{b,N}\Vert.
	\end{equation}
\end{lemma}
\begin{proof}
	We start with the case, where $I =\{x\}\subseteq \{1, \dots, N\}$, i.e., in the tiling decomposition of the root $(N,b) $,  there is a single block, which contains $x$. 
	We proceed as in the proof of Theorem~\ref{thm:EGap}. Comparing terms in~\eqref{eq:fact} and~\eqref{eq:truncfact} and using the orthogonality of the expansion~\eqref{eq:fact}, we conclude
	\begin{align*}
		&\Vert \Psi_{b,N}-\Psi_{b,N}^{(\{x\},l)} \Vert^2 \\
		&\leq \sum_{N_2=1}^N \mkern-5mu \sum_{\substack{N_1, N_3\in \mathbb{N}_0 \\ N_1+N_2+N_3=N}} \mkern-15mu 1\big[x\in [N_1+1, N_1+N_2], D_{b^{(2)}, N_2}\geq l \big] \;  \Vert \Psi_{b^{(1)},N_1} \odot \widehat{\Psi}_{b^{(2)}, N_2} \odot \Psi_{b^{(3)},N_3} \Vert^2 
	\end{align*}
	were he have again adopted the convention that $ \Psi_{b^{(1)},0} \odot \widehat{\Psi}_{b^{(2)}, N_2} \odot \Psi_{b^{(3)},N_3} =  \widehat{\Psi}_{b^{(2)}, N_2} \odot \Psi_{b^{(3)},N_3}$ and similarly, when $N_3=0$. 
	Using the estimate \eqref{eq:normirrep}  on the norm of the irreducible parts and the supermultiplicativity of the norms~\eqref{eq:supermult}, we arrive at
	\begin{align*}
		\Vert \Psi_{b^{(1)},N_1} \odot \widehat{\Psi}_{b^{(2)}, N_2} \odot \Psi_{b^{(3)},N_3} \Vert^2 & = \Vert \Psi_{b^{(1)},N_1} \Vert^2 \ \Vert  \widehat{\Psi}_{b^{(2)}, N_2} \Vert^2 \ \Vert \Psi_{b^{(3)},N_3} \Vert^2 \\
		& \leq \Vert \Psi_{b,N}\Vert^2 \ e^{-C_q(\gamma)D_{b^{(2)}, N_2}} .
	\end{align*}
	For fixed $N_2$ there are at most $N_2$ choice for $(N_1, N_2, N_3)$ where $x$ is contained in the second block. Thus, we arrive at 
	\begin{align*}
		&\Vert \Psi_{b,N}-\Psi_{b,N}^{(\{x\},l)} \Vert^2
		\leq \left( \sum_{N_2=1}^{l} N_2 e^{-C_q(\gamma) \, l} + \sum_{N_2=l+1}^{\infty} N_2 e^{-C_q(\gamma) q (N_2-1)}\right) \Vert \Psi_{b,N}\Vert^2,
	\end{align*}
	which yields the claim in the case $|I| =1$. The case of a general interval $I$, follows from the case $I =\{x\} $ by a union-type bound.
\end{proof}

\subsubsection{Bounds on the bosonic expectation values}
In the fermionic case, local observables $ \mathcal{O}_\loc  $ are bounded in norm. In the bosonic case, this ceases to be the case. The following lemma ensures that expectation values of observables in $\mathcal{O}_\mathrm{loc}$ on (truncated) fractional quantum Hall states are bounded. For the Laughlin state ($b=0$), sharper estimates have been obtained via Coulomb gas methods in \cite[Prop.~5.5]{jansen2012fermionic}.
\begin{lemma} \label{lm:norms observables}
	For a root $ (N,b) $ and an interval  $I\subseteq \{1, \dots, N\}$ and $l>0$, 
	and all finite multisets $X,Y\subseteq \mathbb{N}_0$:
	\begin{equation} \label{norms local observables}
		\begin{split}
			\Vert c_X^* c_Y \Psi_{b,N} \Vert &\leq C \Vert \Psi_{b,N}\Vert \\
			\Vert c_X^* c_Y \Psi_{b,N}^{(I,l)} \Vert &\leq C \Vert \Psi_{b,N}\Vert,
		\end{split}
	\end{equation}
	with $C$ depending only on $C_q(\gamma), \vert X\vert$ and $\vert Y\vert$.
\end{lemma}
\begin{proof}
	Since $ \Vert c_X^* c_Y \Psi_{b,N} \Vert^2 = \langle    \Psi_{b,N}  , c_Y^* c_X c_X^* c_Y  \Psi_{b,N}  \rangle $, the canonical commutation rules allow to deduce the claim from a bound on terms of the form $  \langle    \Psi_{b,N}  , n_{X'} n_{Y'}  \Psi_{b,N}  \rangle $ with suitable subsets $ X' \subseteq X $ and $ Y' \subseteq Y $ and $ n_X = \prod_{x \in X } n_x $, $ n_x = c_x^* c_x $.
	By the Cauchy-Schwarz inequality,  this case is in turn traced back to the case  $X=\{x\}=Y$.  Indeed, using Cauchy-Schwarz repeatedly, we have 
	\begin{align*}
		 \langle    \Psi_{b,N}  , n_X n_Y  \Psi_{b,N}  \rangle  \leq \Vert n_X \Psi_{b,N}\Vert \Vert n_Y \Psi_{b,N} \Vert
		\leq \left(\prod_{x\in X} \Vert n_x^{\vert X\vert} \Psi_{b,N} \Vert\right)^{\frac{1}{\vert X\vert}} \left(\prod_{y\in Y} \Vert n_y^{\vert Y\vert} \Psi_{b,N} \Vert\right)^{\frac{1}{\vert Y\vert}}. 
	\end{align*}
	It remains to consider a single factor $n_x^s$ with $ 1\leq s \leq \max\{ | X | , |Y| \} $ arbitrary. We pick $L\in \{1, \dots, N\}$ such that $x\in [b_{L-1}+q(L-1), b_L+qL-1]$ (with the convention $b_0=0$). The orbital $x$ is then contained in the $L$th elementary segment of the monomer segmentation~\eqref{eq:segroot0} of the orbitals. Using the decomposition~\eqref{eq:fact}, we sort the irreducible parts  according to the unique segment on which the orbital $x$  lies and the rest. Resumming the rest (which is potentially to the right and left of that unique segment),  we obtain
	\begin{equation} \label{norm nx}
		\Vert n_x^s \Psi_{b,N} \Vert^2
		= \mkern-5mu \sum_{\substack{N_1, N_2, N_3\in \mathbb{N}_0\\ N_1+N_2+N_3=N, \\ N_2\geq 1}} \mkern-5mu 1[N_1<L, N_1+N_2\geq L] \Vert n_x^s (\Psi_{b^{(1)},N_1} \odot \widehat{\Psi}_{b^{(2)}, N_2} \odot \Psi_{b^{(3)},N_3}) \Vert^2,
	\end{equation}
	with again the convention that $ \Psi_{b^{(1)},0} \odot \widehat{\Psi}_{b^{(2)}, N_2} \odot \Psi_{b^{(3)},N_3} =  \widehat{\Psi}_{b^{(2)}, N_2} \odot \Psi_{b^{(3)},N_3}$ and similarly, when $N_3=0$. 
	Using the fact that the second segment has $ N_2 $ particles, the norm estimate~\eqref{eq:normirrep} and the supermultiplicativity of the norms~\eqref{eq:supermult}, we estimate
	\begin{align*}
		\Vert n_x^s (\Psi_{b^{(1)},N_1} \odot \widehat{\Psi}_{b^{(2)}, N_2} \odot \Psi_{b^{(3)},N_3}) \Vert^2
		&\leq N_2^{2s} \Vert \Psi_{b^{(1)},N_1} \Vert^2  \Vert  \widehat{\Psi}_{b^{(2)}, N_2} \Vert^2  \Vert  \Psi_{b^{(3)},N_3} \Vert^2 \\
		&\leq N_2^{2s} e^{-C_q(\gamma)q (N_2-1)} \Vert \Psi_{b,N}\Vert^2 . 
	\end{align*}
	For fixed $N_2\in \mathbb{N}$ there are at most $N_2$ choices of $(N_1, N_2, N_3)$ satisfying the constraints in \eqref{norm nx}. Thus, we obtain
	\begin{align*}
		\Vert n_x^s \Psi_{b,N} \Vert^2 \leq \left(\sum_{k=1}^\infty k^{2s+1} e^{-C_q(\gamma) q (k-1)}\right) \Vert \Psi_{b,N} \Vert^2.
	\end{align*}
	
	The second inequality in \eqref{norms local observables} follows from the previous argument combined with the estimate
	\begin{align*}
		\Vert n_x^s \Psi_{b,N}^{(I, l)} \Vert \leq \Vert n_x^s \Psi_{b,N}\Vert ,
	\end{align*}
	which follows from the orthogonality of the expansion~\eqref{eq:fact}. 
\end{proof}

\subsubsection{Preparations and auxiliary results}
As already done in Lemma~\ref{lm:norms observables}, it will become convenient to identify partitions with multisets on $\mathbb{N}_0$ via its occupation numbers.  
Multisets are given by functions $A:\mathbb{N}_0\rightarrow \mathbb{N}_0$. We define the union of two multisets $A,B$ as $A\cup B: \mathbb{N}_0\rightarrow \mathbb{N}_0, (A\cup B)(n)=A(n)+B(n)$, the intersection as $A\cap B: \mathbb{N}_0\rightarrow \mathbb{N}_0, (A\cap B)(n)=\min\{A(n), B(n)\}$, and the difference as $(A\setminus B) : \mathbb{N}_0\rightarrow \mathbb{N}_0, (A\setminus B)(n)=\min\{0, A(n)-B(n)\}$. Finally, the cardinality of a multiset $A$ is $\vert A \vert = \sum_{n\in \mathbb{N}_0} A(n)$. Furthermore, we define $\max (A)=\max_{n\in \mathbb{N}_0} A(n), \min (A)=\min_{n\in \mathbb{N}_0} A(n)$ and $\dist(A,B)=\dist(A^{-1}(\mathbb{N}), B^{-1}(\mathbb{N}))$.\\

From now on, we fix the observables 
$$A=c_{K'}^* c_K \quad \mbox{and}\quad  B=c_{L'}^* c_L$$ 
corresponding to finite multisets $K,K', L,L'\subseteq \mathbb{N}_0$. We define $\mathrm{supp}(A)=K'\cup K$ and $\mathrm{supp}(B)= L'\cup L$. Our standing assumption will be $\max \mathrm{supp}(A)  < \min \mathrm{supp}(B)$. To prove exponential clustering, we will define a truncated Laughlin wavefunction tailored to $A,B$. Through
\begin{equation}
	\begin{split}
		\textstyle \min_q(A) &= \min \{n\in \{1, \dots, N\} \ : \ [b_{n-1}+q(n-1), b_{n}+qn ) \cap \mathrm{supp}(A) \neq \emptyset \}, \\
		\textstyle \max_q(A) &= \max \{n\in \{1, \dots, N\} \ : \ [b_{n-1}+q(n-1), b_n+qn )\cap \mathrm{supp}(A) \neq \emptyset \},
	\end{split}
\end{equation}
with the convention $ b_0 \coloneqq 0 $, we identify the first and last elementary block ('monomer plus potentially preceding voids') in the root tiling, which has a non-empty intersection with $ A $. Similarly, $\min_q(B)$ and $\max_q(B)$ identify the first and last elementary segment that intersects the support of $ B $. We further abbreviate
\begin{equation}
	\begin{split}
		d_{AB} &= \mathrm{dist}(\mathrm{supp}(A), \mathrm{supp}(B)),	\quad	I_{AB} = [ \textstyle \max_q(A)-d_{AB}, \textstyle \min_q(B)+d_{AB}] \cap \{1, \dots, N\}, \\ 
		l_{AB} &= d_{AB}/3, \hspace{3cm} \Psi_{b,N}^{(AB)} = \Psi_{b,N}^{(I_{AB}, l_{AB})}.
	\end{split}
\end{equation}
This truncation yields a good approximation for our purpose. 
\begin{cor} \label{cor:difference norm truncation}
	Suppose that $\gamma$ is large enough such that $C_q(\gamma)>0$. There exist a constant $C>0$, which depends only on $\gamma$ and $q$, such that for any root  $(N,b)$, we have 	
	\begin{equation}
		\Vert \Psi_{b,N} - \Psi_{b,N}^{(AB)}\Vert \leq C d_{AB}^{3/2} \, e^{-C_q(\gamma) d_{AB}/6} \Vert \Psi_{b,N} \Vert.
	\end{equation}
\end{cor}
\begin{proof}
	The claim follows from Lemma \ref{lm:truncated Laughlin} and the crude bound $\min_q(B)-\max_q(A)\leq d_{AB}$.
\end{proof}
Our next aim is to show that $\Psi_{b,N}^{(AB)}$ exhibits exponential clustering. To do so, we need some auxiliary combinatorial lemmata. We abbreviate by $\mathcal{M}_{b,N}^{(A,B)}$ the set of partitions which have a nonzero expansion coefficient in $\Psi_{b,N}^{(AB)}$, such that 
\begin{equation}
	\Psi_{b,N}^{(AB)} = \sum_{\lambda\in \mathcal{M}_{b,N}^{(A,B)}} \mkern-10mu h_b(\lambda) \ \Phi_\lambda , 
\end{equation}
and identify through
\begin{equation} \label{def L1 L2}
	\begin{split}
		L_1 &\coloneqq \max \{ n\in \{1, \dots, N\} \ : \ b_{n-1}+q(n-1)\leq \max \mathrm{supp}(A)+d_{AB}/3\}, \\
		L_2 &\coloneqq \  \min \{ n\in \{1, \dots, N\} \ : \ b_{n}+qn\geq \min \mathrm{supp}(B)-d_{AB}/3\}, 
	\end{split}
\end{equation}
the last elementary block to the right of the support of $ A $ at a distance at most $ d_{AB}/3 $ from it, respectively the first elementary block which starts at a distance at most $ d_{AB}/3 $ to the left of the support of $ B $.

Because of the truncation, the elements in $\mathcal{M}_{b,N}^{(A,B)}$ have a special renewal structure. Either the elementary blocks corresponding to $ \max_q(A) $ and $ \min_q(B) $ are shorter than $ d_{AB}/3 $, or there is a long elementary block, which is separated by renewal points. 
	\begin{figure}[h]
		\begin{center}
			\includegraphics[scale=0.35]{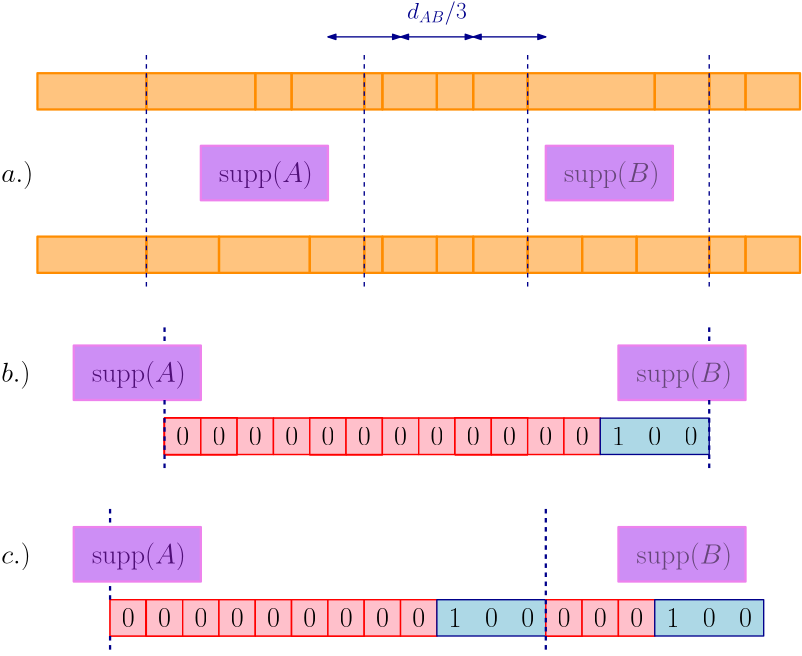}
		\end{center}
		\caption{The three cases in Lemma~\ref{lm:partitions PsiAB}: $a.)$ when $L_1<L_2$. The orange blocks correspond to segmentations of two partitions $\lambda, \mu\in \mathcal{M}_{b,N}^{(AB)}$ with $\langle \Phi_\lambda, AB \Phi_\mu \rangle\neq 0$. The renewal points of blocks that do not intersect the supports of $A$ or $B$ coincide. The dashed line indicates the transition when the segments intersect the supports. Illustrated in $b.)$ and $c.)$ are cases with an elementary block which is longer than or equal to $d_{AB}/3$. The endpoints of those blocks, indicated again by dashed lines, correspond to renewal points common to all partitions in $\mathcal{M}_{b,N}^{(AB)}$.}\label{fig:blocks}
	\end{figure}
\begin{lemma} \label{lm:partitions PsiAB}
	Let $(N,b) $ be a root and $\lambda, \mu\in \mathcal{M}_{b,N}^{(A,B)}$ such that 
	\begin{align*}
		\langle \Phi_\lambda, AB \Phi_\mu \rangle \neq 0.
	\end{align*}
	Then the renewal points of $\mu, \lambda$ strictly between the elementary blocks  $\max_q(A) $ and $ min_q(B) $ coincide.
	Furthermore, either:
	\begin{enumerate}
	\item $ \max_q(A)<  L_1 < L_2 <  \min_q(B) $ and every partition in $\mathcal{M}_{b,N}^{(A,B)}$ has a renewal point between the elementary blocks $\max_q(A)$ (excluding) and $ L_1$ (including) as well as  between $L_2$ (including) and $ \min_q(B)$ (excluding), or, if not, then
	\item there exists an elementary block longer than or equal to $ d_{AB} /3 $ with non-empty intersection with the orbitals in $ [ \max \supp A + d_{AB}/3 , \min \supp B - d_{AB}/3]$ which is separated from the rest by renewal points that are common to all partitions in $\mathcal{M}_{b,N}^{(A,B)}$.
	\end{enumerate}
\end{lemma}
\begin{proof}

	Let us first consider the case that the elementary blocks addressed by $ \max_q(A) $ and $ \min_q(B) $ are shorter than $ d_{AB}/3 $, i.e.
	$$
		b_{\max_q(A)}-b_{\max_q(A)-1}+q< d_{AB}/3, \qquad b_{\min_q(B)}-b_{\min_q(B)-1}+q< d_{AB}/3.
	$$
	This implies $\max_q(A)<L_1$ and $L_2<\min_q(B)$. In this situation, first suppose that
	 $\max_q(A)<L_1<L_2<\min_q(B)$. By construction, every element of $\mathcal{M}_{b,N}^{(A,B)}$ must then have one renewal point after the elementary block $\max_q(A) $ and before $L_1$ and one between $L_2 $ and $ \min_q(B)$. 
	 We define the number of particles, respectively the momentum generated by $ B = c_{L'}^* c_L $ by
	\begin{align*}
		\Delta_B \coloneqq \vert L'\vert - \vert L \vert, \qquad p_B \coloneqq \sum_{k\in L'} k - \sum_{k\in L} k, 
	\end{align*}
	and assume without loss of generality  that $\Delta_B\geq 0$, since otherwise we swap $\mu$ and $\lambda$ (effectively swapping $L$ and $L'$). 
	For any  pair of partitions $\lambda, \mu \in \mathcal{M}_{b,N}^{(A,B)}$, which we may identify with multisets through their occupation numbers, and which satisfy
	$
		\langle \Phi_\mu, c_{K'}^* c_K c_{L'}^* c_L \Phi_\lambda \rangle \neq 0 $, 
	we have 
	\begin{equation} \label{mu from lambda}
		\mu = (\lambda\setminus (L\cup K)) \cup (K'\cup L') . 
	\end{equation} 
	In particular, any entry $\mu_j $ on the elementary blocks from $ \max \mathrm{supp}(A)+1 $ to $ \min \mathrm{supp}(B)-1$ must correspond to some $\lambda_k$ (as those momenta are left unchanged by the action of $A$ and $B$). Let $s\leq L_1$ be a renewal point of $\mu$, i.e.
	\begin{equation} \label{renewal condition mu}
		\sum_{j=s+1}^N (\mu_j - (\lambda_b^{(q)})_j) = 0.
	\end{equation}   
	The relation \eqref{mu from lambda} implies
	\begin{equation} \label{momenta difference}
		\sum_{j=s+1}^N \mu_j - \sum_{j=s-\Delta_B+1}^N \lambda_j = -p_B.
	\end{equation}
	Indeed, \eqref{mu from lambda} tells us that we obtain $\mu$ from $\lambda$ by replacing the momenta in $L\cup K$ by the ones in $K'\cup L'$. Combining \eqref{renewal condition mu} and \eqref{momenta difference}, we obtain
	\begin{equation} \label{renewal equation mu}
		0 = \sum_{j=s+1}^N (\mu_j-(\lambda_b^{(q)})_j)
		= \sum_{j=s-\Delta_B+1}^N (\lambda_j-(\lambda_b^{(q)})_j) -p_B + \sum_{j=s-\Delta_B+1}^s (\lambda_b^{(q)})_j.
	\end{equation}
	As $\lambda \preceq \lambda_b^{(q)}(N)$, we hence conclude
	\begin{equation} \label{upper bound pB}
		p_B - \sum_{j=s-\Delta_B+1}^s (\lambda_b^{(q)})_j = \sum_{j=s-\Delta_B+1}^N (\lambda_j-(\lambda_b^{(q)})_j) \leq 0.
	\end{equation}
	Next, we pick a renewal point $\tilde{s}\geq L_2$ of $\lambda$. Using again \eqref{mu from lambda}, we obtain
	\begin{align*}
		\sum_{j=\tilde{s}+1}^N \lambda_j - \sum_{j=\tilde{s}+\Delta_B+1}^N \mu_j = p_B.
	\end{align*}
	As $\tilde{s}$ is a renewal point of $\lambda$, we get
	\begin{align*}
		0 = \sum_{j=\tilde{s}+1}^N (\lambda_j-(\lambda_b^{(q)})_j)
		= \sum_{j=\tilde{s}+\Delta_B+1} (\mu_j -(\lambda_b^{(q)})_j) + p_B -\sum_{j=\tilde{s}+1}^{\tilde{s}+\Delta_B} (\lambda_b^{(q)})_j
	\end{align*}
	and therefore, using $\mu \preceq \lambda_b^{(q)}(N)$, 
	\begin{equation} \label{lower bound pB}
		p-\sum_{j=\tilde{s}+1}^{\tilde{s}+\Delta_B} (\lambda_b^{(q)})_j = \sum_{j=\tilde{s}+\Delta_B+1}^N ((\lambda_b^{(q)})_j-\mu_j) \geq 0.
	\end{equation}
	Combining \eqref{upper bound pB} and \eqref{lower bound pB}, we arrive at
	\begin{equation} \label{bounds pB}
		\sum_{j=\tilde{s}+1}^{\tilde{s}+\Delta_B} (\lambda_b^{(q)})_j \leq p_B \leq \sum_{j=s-\Delta_B+1}^s (\lambda_b^{(q)})_j.
	\end{equation}
	However, $\tilde{s}\geq L_2>L_1\geq s$ and the entries of $\lambda_b^{(q)}(N)$ are strictly increasing. Thus, \eqref{bounds pB} is only possible if $p_B=0=\Delta_B$. However, in this case \eqref{momenta difference} implies that $s$ is also a renewal point of $\lambda$. Thus, $\lambda$ and $\mu$ share a common renewal point  on the elementry blocks $ \max_q(A)+1, \dots , \min_q(B)-1$. However, because the entries on those blocks must coincide, one obtains that all the renewal points on those blocks must coincide.
	
	If $L_1=L_2$, then $
		b_{L_1}-b_{L_1-1} +q  \geq d_{AB}/3 $, and hence the elementary block $ L_1 $ is separated from the rest by renewal points for all
	$\lambda\in \mathcal{M}_{b,N}^{(A,B)}$. Again, $\mu$ and $\lambda$ must share all their renewal points on the blocks between $ \max_q(A)$ and $  \min_q(B)$. 
	
	We now address the case that the elementary block $ \max_q(A) $ has length
	$$ b_{\max_q(A)}-b_{\max_q(A)-1}+q\geq d_{AB}/3. $$
	By definition of $\Psi_{b,N}^{(AB)}$, the elementary block $\max_q(A)$ is then separated by the rest through renewal points for all elements in $\mathcal{M}_{b,N}^{(A,B)}$. Since the entries of $\mu, \lambda$ in the on the orbitals $ \max \mathrm{supp}(A)+1, \dots ,  \min \mathrm{supp}(B)-1$ must coincide, this implies that all the renewal points on the blocks $\max_q(A), \dots , \min_q(B)-1$ must coincide.
	
	The case $ b_{\min_q(B)}-b_{\min_q(B)-1}+q\geq d_{AB}/3 $ works the same way.
\end{proof}

The dichotomy in Lemma~\ref{lm:partitions PsiAB} is reflected in two scenarios for the wavefunction $\Psi_{b, N}^{(AB)}$. In the second case of Lemma~\ref{lm:partitions PsiAB}, the function factorizes,
\begin{equation}\label{eq:longblockfactor}
\Psi_{b, N}^{(AB)} = \Psi_{b(1,M-1)}^{(AB)} \odot  \widehat \Psi_{ b(M,M)}  \odot \Psi_{b(M+1,N)}^{(AB)}
\end{equation}
with suitable $ M\in \mathbb{N} $, such that the $ M $th monomer belongs to the long elementary block and such that and $ b(M,M) \in \mathbb{N}_0 $ is the number of voids on the long elementary block, and accordingly $ b = \big(b(1,M-1), b(M,M) , b(M+1,N) )  $. Here and in the following, we 
use the following notation to address the void segmentation.
\begin{defn}
 For a root $ (N,b) $ and $1\leq N_1\leq N_2\leq N$ the \emph{void partition on the segment $ [N_1,N_2] $} is abbreviated by
\begin{equation}
	b(N_1, N_2) = (b_{N_1}-b_{N_1-1}, b_{N_1+1}-b_{N_1-1}, \dots, b_{N_2}-b_{N_1-1})
\end{equation}
with the convention $b_0=0$.\\
We define $\mathcal{M}_{b(N_1, N_2)}^{(A,B)}$ as the set of partitions $\lambda\preceq \lambda^{(q)}_{N_2-N_1+1}\big(b(N_1,N_2)\big) $ such that there exists $\mu^{(1)}\preceq \lambda^{(q)}_{N_1-1}(b(1,N_1-1)), \mu^{(2)}\preceq \lambda{(q)}_{N-N_2}(b(N_2+1,N)) $ with  $\mu^{(1)}\cup \lambda \cup \mu^{(2)} \in \mathcal{M}_{b,N}^{(A,B)}$. Then we set the block contributions
\begin{equation}
	\Psi_{b(N_1,N_2)}^{(AB)} \coloneqq \sum_{\lambda \in \mathcal{M}_{b(N_1, N_2)}^{(A,B)}} h_{b(N_1, N_2)}(\lambda) \ \Phi_{\lambda} . 
\end{equation}
\end{defn}

In the first case in Lemma~\ref{lm:partitions PsiAB}, it is natural to condition the expansion~\eqref{eq:fact} on the first renewal point between $ \max_q(A) $ and $ L_1 $, as well as the last renewal point between $ L_2 $ and $ \max_q(B) $. We may then write 
\begin{align}\label{eq:maincasecond}
	\Psi_{b, N}^{(AB)} = 
	& \mkern-5mu \sum_{\substack{ n \in (  \max_q(A) , L_1 ] \\  n \in [L_2 ,  \max_q(B) ) }}  \widetilde \Psi_{b(1,n-1)}^{(AB)} \odot   \Psi_{b(n,m)}^{(AB)}  \odot \widetilde \Psi_{b(m+1,N)}^{(AB)} \\
&\qquad \mbox{where} \qquad 
	 \widetilde \Psi_{b(1,n-1)}^{(AB)}  \coloneqq  \mkern-50mu \sum_{\substack{\lambda \in \mathcal{M}_{b(1,n-1)}^{(A,B)}\\ \text{no renewal point in } (\max_q(A), n-1)}}  \mkern-50mu  h_{b(1,n-1)}(\lambda)\  \Phi_\lambda , \notag
\end{align}
and likewise for $ \widetilde \Psi_{b(m+1,N)}^{(AB)} $. Furthermore, we define the irreducible parts of $\Psi_{b(N_1,N_2)}^{(AB)}$ as 
\begin{equation}
	\widehat\Psi_{b(N_1, N_2)}^{(AB)} \coloneqq \sum_{\substack{\lambda \in \mathcal{M}_{b(N_1, N_2)}^{(A,B)} \\ \lambda \, \irr}} h_{b(N_1, N_2)}(\lambda) \ \Phi_{\lambda}
\end{equation}
One readily sees that
\begin{equation} \label{norms truncated vs untruncated}
	\begin{split}
			\Vert \Psi_{b(N_1,N_2)}^{(AB)} \Vert &\leq \Vert \Psi_{b(N_1, N_2), N_2-N_1+1} \Vert \\
			\Vert \widetilde{\Psi}_{b(N_1,N_2)}^{(AB)} \Vert &\leq \Vert \Psi_{b(N_1, N_2), N_2-N_1+1} \Vert \\
			\Vert \widehat{\Psi}_{b(N_1,N_2)}^{(AB)} \Vert &\leq \Vert \widehat \Psi_{b(N_1, N_2), N_2-N_1+1} \Vert.
	\end{split}
	\end{equation}

As a last preparation for the proof of the exponential clustering of $\Psi_{b,N}^{(AB)}$, we need some a priori control on products of norms. 
We now show that for blocks containing the piece between $\supp A$ and $\supp B$, the product of norms are comparable. This has also been key in the proof of clustering in~\cite{jansen2009symmetry}. However, due to the presence of voids, one can no longer use standard tools from renewal theory. Instead, we employ a simple combinatorial switching argument. %
\begin{lemma}
	In the first case of Lemma~\ref{lm:partitions PsiAB}, for all $n,m$ with $\max_q (A)<n\leq L_1<L_2\leq m<\min_q(B)$, we have
	\begin{equation} \label{difference of products norms}
		\big\vert \Vert \Psi_{b(n,m)}^{(AB)} \Vert^2 \Vert \Psi_{b,N}^{(AB)}\Vert^2 - \Vert \Psi_{b(n,N)}^{(AB)} \Vert^2 \Vert \Psi_{b(1,m)}^{(AB)} \Vert^2 \big \vert
		\leq 2d_{AB}^2 e^{-C_q(\gamma)d_{AB}/3+\ln(4)d_{AB}/3} \Vert \Psi_{b, N}\Vert^2,
	\end{equation}
	with $C_q(\gamma)$ defined in \eqref{def:Cgamma}.
\end{lemma}
\begin{proof}
Using the orthogonality of the expansion~\eqref{eq:fact} in irreducible segments, we write 
	\begin{align} \label{expansion products}
			&\Vert \Psi_{b(n,m)}^{(AB)} \Vert^2 \Vert \Psi_{b,N}^{(AB)}\Vert^2 - \Vert \Psi_{b(n,N)}^{(AB)} \Vert^2 \Vert \Psi_{b(1,m)}^{(AB)} \Vert^2  = \\
			& \quad \sum_{r=1}^{m-n+1} \sum_{\substack{N_1, \dots, N_r\in \mathbb{N}\\ \sum_{j=1}^r N_j = m-n+1}} \sum_{s=1}^{N} \sum_{\substack{M_1, \dots, M_s\in \mathbb{N}\\ \sum_{j=1}^s M_j = N}} \ \prod_{j=1}^r \Vert \widehat \Psi_{b(n+\sum_{k=1}^{j-1} N_k, n-1+\sum_{k=1}^j N_k)}^{(AB)} \Vert^2   \prod_{j=1}^s \Vert \widehat \Psi_{b(1+\sum_{k=1}^{j-1} M_k, \sum_{k=1}^j M_k)}^{(AB)} \Vert^2    \notag \\
			&\qquad-\sum_{\tilde{r}=1}^{N-m} \sum_{\substack{\tilde{N}_1, \dots, \tilde{N}_{\tilde{r}}\in \mathbb{N}\\ \sum_{j=1}^{\tilde{r}} \tilde{N}_j = N-n+1}} \sum_{\tilde{s}=1}^{n-1} \sum_{\substack{\tilde{M}_1, \dots, \tilde{M}_{\tilde{s}}\in \mathbb{N}\\ \sum_{j=1}^{\tilde{s}} \tilde{M}_j = m}} \ \prod_{j=1}^{\tilde{r}} \Vert \widehat \Psi_{b(n+\sum_{k=1}^{j-1} \tilde{N}_k, n-1+\sum_{k=1}^j \tilde{N}_k)}^{(AB)} \Vert^2\notag \prod_{j=1}^{\tilde{s}} \Vert \widehat \Psi_{b(1+\sum_{k=1}^{j-1} \tilde{M}_k, \sum_{k=1}^j \tilde{M}_k)}^{(AB)} \Vert^2 \notag
	\end{align}
	Let $\mathbf{N}=(N_1, \dots, N_r)$, $\mathbf{N}=(M_1, \dots, M_s)$ and similarly $\tilde{\mathbf{N}}=(\tilde{N}_1, \dots, \tilde{N}_{\tilde{r}})$, $\tilde{\mathbf{M}}=(\tilde{M}_1, \dots, \tilde{M}_{\tilde{s}})$  abbreviate the segmentation in the above sums. In the following, we distinguish two cases in these sums, e.g.\ for the first sums either 
	\begin{itemize}
	\item   $ \mathbf{N} $ and $ \mathbf{M} $ have a common renewal point in $\{n, n+1, \dots, m\}$, which we call the gap, 
	or
	\item   $ \mathbf{N} $ and $ \mathbf{M} $ have no common renewal point in the gap.
	\end{itemize}
	Here and in the following, we call  $a\in \{n, n+1, \dots, m\}$ 
	 is a common renewal point of  $ \mathbf{N} $ and $ \mathbf{M} $  in the gap if there exists $u\in \{0,1, \dots, r\}$ and $v\in \{1, \dots, s\}$ such that
	\begin{equation} \label{def L}
		n+\sum_{k=1}^u N_k = a = \sum_{k=1}^v M_k.
	\end{equation}
	We write $C(  \mathbf{N}, \mathbf{M} )=a$, where $a$ is the smallest common renewal point in the gap, and $C(  \mathbf{N} , \mathbf{M} )=\emptyset$ if there is no common renewal point in the gap. 
	Furthermore, we set $L(\mathbf{N},a)=u, L(\mathbf{M},a) = v$ as defined in \eqref{def L}. 
	\begin{figure}[h]
		\begin{center}
					\includegraphics[scale=0.35]{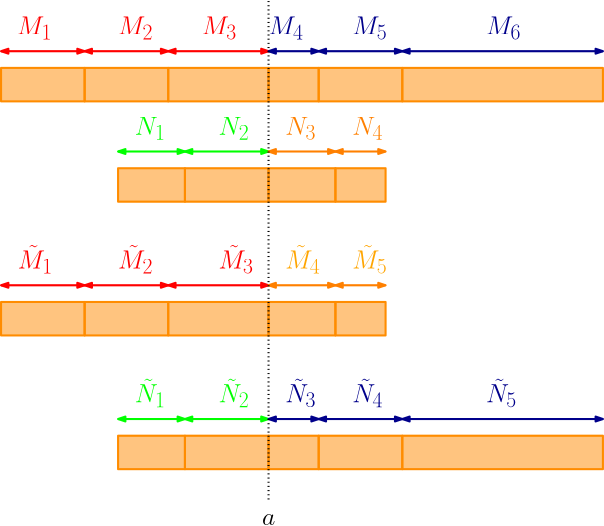}
		\end{center}
		\caption{Construction of $\tilde{\mathbf{M}}, \tilde{\mathbf{N}}$ from $\mathbf{M},\mathbf{N}$ and the first common renewal point $a$ in the gap.}\label{fig:swap}
	\end{figure}
In the first case, $C(  \mathbf{N} , \mathbf{M} )=a$, we identify a swapped segmentation
	\begin{align*}
		\tilde{\mathbf{N}}&=(N_1, \dots, N_{L(\mathbf{N},a)}, M_{L(\mathbf{M},a)+1}, \dots, M_s), \\
		\tilde{\mathbf{M}}&= (M_1, \dots, M_{L(\mathbf{M},a)}, N_{L(\mathbf{N},a)+1}, \dots, N_r).
	\end{align*}
	The $\tilde{\mathbf{N}}, \tilde{\mathbf{M}}$ are admissible segmentations for the second sum in~\eqref{expansion products} with common renewal point $C(\tilde{\mathbf{N}}, \tilde{\mathbf{M}})=a$ in the gap, cf.~Figure~\ref{fig:swap}. 
	Thus, conditioning both sums in~\eqref{expansion products} on the existence of common renewal points,  these contributions exactly cancel, and we are left with the terms without a common renewal point in the gap. \\

For an estimate of each of the sums in case $C(  \mathbf{N} , \mathbf{M} )=\emptyset$ and likewise for $C(\tilde{\mathbf{N}}, \tilde{\mathbf{M}})=\emptyset$, we recall from~\eqref{eq:normirrep} that $
		\Vert \widehat \Psi_{b,N}^{(AB)} \Vert^2 \leq \exp\left(-C_q(\gamma) D_{b,N} \right) $ with $ D_{b,N}  = q(N-1)+b_N-b_1 $.
	The true length, counting the number of orbitals of the corresponding segment, is $qN+b_N$. The difference is $q$ plus the number of voids $ b_1 $ at the beginning of the segment. Hence, we may bound
	\begin{align*}
		\prod_{j=1}^r \Vert \widehat \Psi_{b(n+\sum_{k=1}^{j-1} N_k, n-1+\sum_{k=1}^j N_k)}^{(AB)} \Vert^2 
		\leq \exp\left(-C_q(\gamma)(d(n,m)-qr-R(\mathbf{N}))\right)
	\end{align*} 
	where $R(\mathbf{N})$ is the sum of the number of voids at the renewal points of $\mathbf{N}$ and $d(n,m)=q(m-n+1)+b_m-b_{n-1}$ is the total length of the elementary blocks $n$ to $m$. By construction, $d(n,m)\geq d_{AB}/3$ as $n\leq L_1$ and $m\geq L_2$. 
	Let 
	\begin{align*}
		B(\mathbf{M}) &\coloneqq \max\{t \in \{1, \dots, s\} \ : \ 1+\sum_{k=1}^{t-1} M_k \leq n\}, \\
		E(\mathbf{M}) & \coloneqq  \min\{t \in \{1, \dots, s\} \ : \ \sum_{k=1}^t M_k \geq m+1\}.
	\end{align*}
		 As $\mathbf{N}$ and $\mathbf{M} $ have no common renewal point in the gap, for every renewal point given by $\mathbf{N}$ there is a longer block of $\mathbf{M}$ at that very location, which interlaces the segmentation $\mathbf{N}$, see Figure~\ref{fig:interlace}.
	 Hence, given the number $ r $ of segments in $\mathbf{N}$ and the number $ R(\mathbf{N}) $ of voids at its renewal points, we have
	 \begin{align*}
	 	\prod_{j=B(\mathbf{M})}^{E(\mathbf{M})} \Vert \widehat \Psi_{b(1+\sum_{k=1}^{j-1} M_k, \sum_{k=1}^j M_k)}^{(AB)} \Vert^2
	 	\leq \exp\left(-C_q(\gamma)(qr+R(\mathbf{N}))\right) . 
	 \end{align*}
	  \begin{figure}[h]
	 \begin{center}
	 	\includegraphics[scale=0.4]{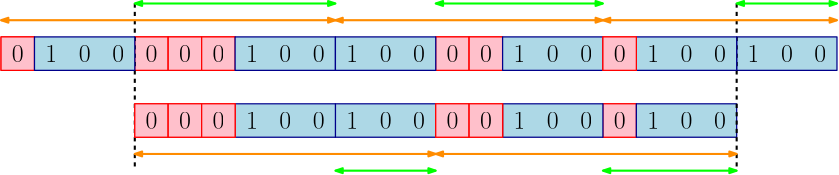}
	 \end{center}
	 \caption{In red and blue we see the elementary blocks, the orange lines on the bottom indicate some segmentations corresponding to $\mathbf{N}$, and the orange lines on top indicate the segmentation corresponding to the parts of $\mathbf{M}$ which overlap with the elementary blocks of $b(n,m)$. The exponential decay estimate \eqref{eq:normirrep} is telling us that the irreducible part of the wavefunction corresponding to such a segment is exponentially suppressed with the exponent given by $C_q(\gamma)$ times the length of the segment without the leading elementary block. The length of the segments without their leading elementary block is indicated by green lines. As the segmentations have no common renewal points, they are interlacing in such a way that the length of all the green pieces sums to at least the sum of the lengths of the elementary blocks of $b(n,m)$, which in this case is greater than or equal to $d_{AB}/3$. }\label{fig:interlace}
	 \end{figure}
	If  $C(  \mathbf{N} , \mathbf{M} )=\emptyset$, we thus have
	\begin{align*}
		&\prod_{j=1}^r \Vert \widehat \Psi_{b(n+\sum_{k=1}^{j-1} N_k, n-1+\sum_{k=1}^j N_k)}^{(AB)} \Vert^2 \prod_{j=1}^s \Vert \widehat \Psi_{b(1+\sum_{k=1}^{j-1} M_k, \sum_{k=1}^j M_k)}^{(AB)} \Vert^2 \\
		&\leq e^{-C_q(\gamma)d(n,m)} \left(\prod_{j=1}^{B(\mathbf{M})-1}  \Vert \widehat \Psi_{b(1+\sum_{k=1}^{j-1} M_k, \sum_{k=1}^j M_k)}^{(AB)} \Vert^2 \right)
		\left(\prod_{j=E(\mathbf{M})+1}^s \Vert \widehat \Psi_{b(1+\sum_{k=1}^{j-1} M_k, \sum_{k=1}^j M_k)}^{(AB)} \Vert^2\right)
	\end{align*}
	There are at most $d_{AB}$ possible choices for the last renewal point before $n$ and at most $d_{AB}$ possible choices for the first renewal point after $m$, such that
	\begin{align*}
		&\sum_{r=1}^{m-n+1} \sum_{\substack{N_1, \dots, N_r\in \mathbb{N}\\ \sum_{j=1}^r N_j = m-n+1}} \sum_{s=1}^{N} \sum_{\substack{M_1, \dots, M_s\in \mathbb{N}\\ \sum_{j=1}^s M_j = N}} 1[C(  \mathbf{N} , \mathbf{M} )=\emptyset] \\ 
		&\qquad \times \prod_{j=1}^r \Vert \widehat \Psi_{b(n+\sum_{k=1}^{j-1} N_k, n-1+\sum_{k=1}^j N_k)}^{(AB)} \Vert^2 \prod_{j=1}^s \Vert \widehat \Psi_{b(1+\sum_{k=1}^{j-1} M_k, \sum_{k=1}^j M_k)}^{(AB)} \Vert^2 \\
		&\leq \sum_{\ell=1}^{d_{AB}/3} \sum_{\tilde{\ell}=1}^{d_{AB}/3} \Vert \Psi_{b(1,n-\ell)}^{(AB)} \Vert^2 \left(\sum_{r=1}^{m-n+1} \sum_{s=1}^{m-n+1} \sum_{\substack{T_1, \dots, T_r\in \mathbb{N} \\ \sum_{j=1}^r T_j=m-n+1}} \sum_{\substack{U_1, \dots, U_s\in \mathbb{N}\\ \sum_{j=1}^s U_j=m-n+1}} e^{-C_q(\gamma) d(n,m)} \right) \Vert \Psi_{b(m+1+\tilde{\ell}, N)}^{(AB)}\Vert^2 \\
		&\leq d_{AB}^2 4^{m-n+1}  e^{-C_q(\gamma)d(n,m)} \Vert \Psi_{b, N}\Vert^2,
	\end{align*}
	where we have used \eqref{norms truncated vs untruncated}, that the norms are supermultiplicative \eqref{eq:supermult} and the bounds $m-n+1\leq d(n,m)$ and $d(n,m)\geq d_{AB}/3$ combined with the fact that there are exactly $2^{m-n}$ ways of partitioning the set $\{1, \dots, m-n+1\}$. One gets the same estimate for the remaining term, which completes the proof.
\end{proof}

\subsubsection{Proof of Theorem \ref{thm:clustering}}
We first prove the exponential clustering for the truncated wavefunction $\Psi_{b,N}^{(AB)}$.
\begin{prop} \label{prop:exponential clustering truncated}
	In the situation of Theorem \ref{thm:clustering}, there exists a constant $C>0$ depending only on the cardinality of $\supp(A)$ and $\supp(B)$ such that 
	\begin{equation}\label{eq:clustertrunc}
		\begin{split}
			&\big\vert \langle \Psi_{b,N}^{(AB)}, AB \Psi_{b,N}^{(AB)} \rangle\  \Vert \Psi_{b, N}\Vert^2 - \langle \Psi_{b,N}^{(AB)}, A \Psi_{b,N}^{(AB)} \rangle \langle \Psi_{b,N}^{(AB)}, B \Psi_{b,N}^{(AB)} \rangle\big\vert \\
			&\qquad\leq C d_{AB}^4 e^{-(C_q(\gamma)-\ln(4)) d_{AB}/3} \Vert \Psi_{b, N}\Vert^4.
		\end{split}
	\end{equation}
\end{prop}
\begin{proof}
	According to Lemma \ref{lm:partitions PsiAB}, there are two cases. 
	
	If there is an elementary block of length greater or equal to $ d_{A_B}/ 3 $ with non-empty intersection with the orbitals in $ [ \max \supp A + d_{AB}/3 , \min \supp B - d_{AB}/3]$, the wavefunction factorizes over the long elementary block~\eqref{eq:longblockfactor} and hence the left side in~\eqref{eq:clustertrunc} is zero. 
	
	In the other case,  $\max_q(A) < L_1 < L_2 < \min_q(B) $, we expand and use the conditioning on the first and last renewal point in~\eqref{eq:maincasecond}. Since every $\lambda, \mu\in \mathcal{M}_{b,N}^{(A,B)}$ with $\langle \Phi_{\lambda}, AB \Phi_\mu\rangle\neq 0$ must share all renewal points in $(\max_q(A), \min_q(B))$, we write
		\begin{align*}
		&\langle \Psi_{b,N}^{(AB)}, AB \Psi_{b, N}^{(AB)} \rangle \Vert \Psi_{b,N}^{(AB)} \Vert^2 - \langle \Psi_{b,N}^{(AB)}, A \Psi_{b, N}^{(AB)} \rangle \langle \Psi_{b,N}^{(AB)}, B \Psi_{b, N}^{(AB)} \rangle \\
		&= \sum_{\substack{n\in (\max_q(A), L_1] \\ m\in [L_2, \min_q(B))}}
			f_A(n) f_B(m) \left( \Vert \Psi_{b(n+1,m)}^{(AB)} \Vert^2 \Vert \Psi_{b,N}^{(AB)}\Vert^2 - \Vert \Psi_{b(n+1,N)}^{(AB)}\Vert^2 \Vert \Psi_{b(1,m)}^{(AB)} \Vert^2\right).
	\end{align*}
	Here, we abbreviated  for $n\in (\max_q(A), \min_q(B))$:
\begin{align*}
		f_A(n) &\coloneqq \mkern-30mu \sum_{\substack{\lambda, \mu\in \mathcal{M}_{b(1,n)}(A,B)\\ \text{no renewal point in } (\max_q(A), n)}}  \mkern-30mu h_{b(1,n)}(\lambda) \   h_{b(1,n)}(\mu) \ \langle \Phi_\mu, A \Phi_\lambda \rangle, \\
		f_B(m) & \coloneqq  \mkern-30mu \sum_{\substack{\lambda, \mu\in \mathcal{M}_{b(m+1,N)}(A,B) \\ \text{no renewal point in } [1,\min_q(B)-m)}}  \mkern-30mu h_{b(m+1,N)}(\lambda) \ h_{b(m+1,N)}(\mu) \ \langle \Phi_{\lambda^{(q)}_{b(1,m)} \cup\mu}, B \Phi_{\lambda_{b(1,m)}^{(q)}\cup \lambda} \rangle. 
\end{align*}

	A Cauchy-Schwarz estimate and Lemma \ref{lm:norms observables} implies that there are  $C_A, C_B \in (0,\infty)$, which  only depend on the cardinality of $\supp (A)$ and $\supp(B)$, such that 
	\begin{align*}
		\vert f_A(n) \vert \leq C_A \Vert \Psi_{b(1,n), n}\Vert^2, \qquad 
		\vert f_B(m) \vert \leq C_B \Vert \Psi_{b(m+1,N), N-m}\Vert^2.
	\end{align*}
	Combining this with \eqref{difference of products norms} and the crude estimates, $L_1-\max_q(A)-1\leq d_{AB}, \min_q(B)-1-L_2\leq d_{AB}$, as well as the supermultiplicativity of the norm~\eqref{eq:supermult}, we arrive at the claim.
\end{proof}

We finally have all the ingredients to prove exponential clustering of the untruncated wavefunction.
\begin{proof}[Proof of Theorem \ref{thm:clustering}]
	The triangle  inequality combined with a Cauchy-Schwarz estimate yields 
	\begin{align*}
		\vert \langle \Psi_{b,N}, O \Psi_{b,N}\rangle -\langle \Psi_{b,N}^{(AB)}, O \Psi_{b,N}^{(AB)}\rangle \vert
		\leq \left(\Vert O \Psi_{b,N}\Vert + \Vert O\Psi_{b,N}^{(AB)}\Vert\right) \Vert \Psi_{b,N}- \Psi_{b,N}^{(AB)} \Vert
	\end{align*}
	for any $O\in \mathcal{O}_\mathrm{loc}$. Thus, the claim follows from Proposition \ref{prop:exponential clustering truncated} and Corollary \ref{cor:difference norm truncation}.
\end{proof}

		\appendix 
		\section{Renewal analysis and the norm of the Laughlin wavefunction}\label{app:new}

		In this appendix, we gather known results on the analysis of discrete renewal equations and prove some variations of such results. We refer to \cite{mitov2014renewal} and references therein for the general background on renewal theory. We then show how to combine these results with Theorem~\ref{thm:fact} for a quantitative estimate of the norm of the Laughlin wavefunction. This improves results in~\cite{jansen2009symmetry}.

		\subsection{Basics}
		In the following we consider solutions $ (C_n) $ of the discrete renewal equations
		\begin{equation} \label{renewal equation appendix}
			C_n = \sum_{j=1}^{n-1} \alpha_j C_{n-j} + \beta_n, \quad n \in \mathbb{N} ,  
		\end{equation}
		defined  in terms of two non-negative sequences $ (\alpha_n), (\beta_n)  \in [0,\infty)^\mathbb{N} $. 
		Associated to these sequences are the power series
		\begin{equation}\label{def:powerseries}
		 \alpha(z) = \sum_{n=1}^{\infty} \alpha_n z^n, \qquad \beta(z)=\sum_{j=1}^{\infty} \beta_n z^n , 
		 \end{equation}
		 whose radii of convergence we denote by $ R_\alpha $ and $ R_\beta $, respectively. By resumming terms, it is easy to see that the power series 
		 $$ C(z) = \sum_{n=1}^\infty C_n z^n 
		 $$ associated to any solution of~\eqref{renewal equation appendix} satisfies:
		\begin{equation} \label{renewal power series equation}
			C(z) \left(1-\alpha(z)\right) = \beta(z). 
		\end{equation}

		\begin{lemma}
			If $ R_\alpha > 0 $ and $ \alpha_1 > 0 $, then 
			the radius of convergence $ r $ of the power series $ C(z) $ satisfies:
			$$
			r = \min\left\{ R_\beta , \min\{ t \in [0,R_\alpha]    \ : \ \alpha(t) = 1 \} \right\}  \leq  \min\left\{ R_\beta , \alpha_1^{-1} \right\} .
			$$
		\end{lemma} 
		\begin{proof} Since $ (C_n) \in [0,\infty)^\mathbb{N}  $, the radius of convergence of $ C(z) $ equals
		$$
		r = \sup\left\{  t \in [0,\infty) \ : \ C(t) = \sum_{n=1}^\infty C_n t^n < \infty \right\} .
		$$
		By~\eqref{renewal power series equation}, this radius is restricted by: 1.~the radius of convergence of $ \beta(z) $, and 2.~the first solution on $[0,\infty) $ of the equation $ \alpha(t) = 1 $. Since  
		$ \alpha(0) = 0 $ and, for any $ 0\leq t < R_\alpha $ also  $ \alpha'(t) \geq \alpha'(0)  = \alpha_1 > 0 $, the first  solution on $[0,\infty) $ of $ \alpha(t) = 1 $ is in the interval $ (0, \alpha_1^{-1} ] $. This finishes the proof. 
		\end{proof}
		
		\subsection{Criteria for exponential convergence}
		By the above lemma, in case $ 0 < r < \min\{ R_\alpha, R_\beta \}  $, then $ r = \min\{ t \in [0,R_\alpha]    \ \big| \ \alpha(t) = 1 \} $, which means that $p_n \coloneqq \alpha_n r^n $ is a probability distribution on $ \mathbb{N} $ whose mean
		\begin{equation}
		\mu \coloneqq \sum_{n=1}^{\infty} n \alpha_n r^n
		\end{equation}
		is finite. In this situation, it is known \cite[XIII, Theorem 1]{feller1968introduction} that 
		\begin{equation}\label{eq:renewalFeller}
		\lim_{n\to \infty}  C_n r^n = \frac{\beta(r)}{\mu} . 
		\end{equation} 
		The following criterion ensures that the rate of this convergence is exponential with an explicit rate.  The idea of the proof is very similar to the one in \cite{baxendale2005renewal}.
	\begin{theorem} \label{thm:renewal appendix}
			Assume that $0<r< 1 \leq R_\beta$ and that there exist constants $c,C>0$ such that
			\begin{equation}\label{assumption renewal theorem}
				\alpha_n \leq C e^{-c(n-1)} \qquad \text{ for all } n\in \mathbb{N}.
			\end{equation}
					Then, we have
			\begin{equation} \label{renewal decay appendix}
				\vert C_n r^n - \beta(r)/\mu \vert = \mathcal{O}(R^{-n}), \qquad \text{ for } n\rightarrow \infty.
			\end{equation}
			for any 
			$$
				0<R< \min\left\{\frac{R_\beta}{r}, \frac{e^c}{r} \left(\frac{1}{1+\frac{C}{1-e^{-c}}}\right) \right\} .
			$$
		\end{theorem}
		\begin{proof}
			Note that $R_\alpha \geq  e^{c } > 1   $ by the exponential bound~\eqref{assumption renewal theorem}, and hence  $\mu<\infty$.
			We want to estimate 
			\begin{align*}
				x_n= C_n r^n -\frac{\beta(r)}{\mu}.
			\end{align*}
			For $\vert z \vert<1 $ all the power series below converge due to our assumptions and we have, using \eqref{renewal power series equation},
			\begin{equation} \label{def X(z)}
				X(z) :=\sum_{n=1}^\infty x_n z^n 
				= C(rz) - \frac{\beta(r)}{\mu} \frac{z}{1-z}
				= \frac{\beta(rz)}{1-\alpha(rz)} - \frac{\beta(r)}{\mu} \frac{z}{1-z}.
			\end{equation}
			As $\alpha(r)=1$, we get $\sum_{n=1}^{\infty} p_n =1$ and therefore,
			\begin{align*}
				\alpha(rz) - 1 & = \sum_{n=1}^{\infty} p_n z^n -1 
				= \sum_{n=1}^\infty p_n(z^n -1)
				= \sum_{n=1}^{\infty} p_n (z-1) \sum_{j=0}^{n-1} z^j \\
				& = (z-1) \left(1+\sum_{k=1}^\infty \left(\sum_{n=k+1}^\infty p_n\right) z^k\right)
			\end{align*}
			where the last step is by interchanging the sums. Inserting this into \eqref{def X(z)} yields
			\begin{align*}
				X(z) = \frac{1}{1-z} \left(\frac{\beta(rz)}{1+\sum_{k=1}^{\infty} \left(\sum_{n=k+1}^{\infty} p_n\right) z^k} -\frac{\beta(r)}{\mu} z \right).
			\end{align*}
			The second factor vanishes for $z=1$, since $ \mu =  \sum_{k=0}^{\infty} \left(\sum_{n=k+1}^{\infty} p_n\right)  $. Moreover, it has radius of convergence $R_X$ equal to the minimum of $R_\beta/r >R_\beta \geq 1 $ and
			\begin{align*}
				R_0\coloneqq\min \left\{ |z| \ :  \ \sum_{k=1}^\infty \left(\sum_{n=k+1}^\infty p_n\right) z^k =-1 \right\} .
			\end{align*}
			Thus, $X(z)$ also has radius of convergence $R_X= \min\{R_\beta/r, R_0\}$.
			
			We now need to establish a lower bound for $R_0$. For this we define
			\begin{align*}
				G(z) \coloneqq \sum_{k=1}^{\infty} \left(\sum_{n=k+1}^{\infty} p_n\right) z^k.
			\end{align*}
			As $p_n\geq 0$, we get
			$	\vert G(z) \vert \leq G(\vert z\vert) $, and hence
			$
				R_0\geq \min \left\{ t\in [0,\infty) \ : \ G(t) =1 \right\} $. 
			Using $r\leq 1$ and \eqref{assumption renewal theorem}, we obtain
			\begin{align*}
				G(\vert z \vert) 
				&= \sum_{k=1}^{\infty} \left(\sum_{n=k+1}^{\infty} \alpha_n r^n \right) \vert z \vert^k
				\leq  C \sum_{k=1}^{\infty} \left(\sum_{n=k+1}^{\infty} e^{-c(n-1)}\right) \vert rz \vert^k \\
				&=C \sum_{k=1}^\infty \frac{e^{-ck}}{1-e^{-c}} \vert rz \vert^k
				= \frac{C}{1-e^{-c}} \frac{\vert rz\vert e^{-c}}{1-\vert rz\vert e^{-c}}.
			\end{align*}
			The solution of
		$
			\frac{C}{1-e^{-c}} \frac{\vert rz\vert e^{-c}}{1-\vert rz\vert e^{-c}} = 1 $ 
			is equal to
			$
				\vert z \vert = \frac{e^c}{r} \left(\frac{1}{1+\frac{C}{1-e^{-c}}}\right) $. 
			Thus, the radius of convergence of $X(z)$ is bounded from below by
			\begin{align*}
				R_X \geq \min\left\{\frac{R_\beta}{r}, \frac{e^c}{r} \left(\frac{1}{1+\frac{C}{1-e^{-c}}}\right) \right\}.
			\end{align*}
		\end{proof}
		
		\subsection{Implication for the norm of the Laughlin wavefunction on cylinder}
		In case of the Laughlin wavefunction, $ b = 0 $, on the cylinder, the norm-squared $ \| \Psi_{0,N}\|^2 $ satisfies the discrete renewal equation~\eqref{renewal equation appendix} with $ \alpha_n = \beta_n $
		given by the irreducible contribution~\eqref{eq:normirrep}  to the norm (with $ b = 0 $). 
		More generally, if $ M \in \mathbb{N}_0 $ and $ b^{(2)} \in \mathbb{N}_0^M $ is a partition, then $ b = \big(0 , \dots , 0 , b^{(2)} \big) \in \mathbb{N}_0^N $ with $ N > M $ models a fractional quantum Hall state 
		which is a Laughlin state in the bulk, but a thinned-out version at the right boundary. Renewal theory then implies
		\begin{cor}
		Assume that $C_q(\gamma)>0$ and let $ M \in \mathbb{N}_0 $ be fixed and $ b^{(2)} \in \mathbb{N}_0^M $  be a partition. Then the norm-squared 
		$$ C_N \coloneqq \big\| \Psi_{b,N} \big\|^2 $$ 
		of the fractional quantum Hall wavefunction associated with $ b = \big(0 , \dots , 0 , b^{(2)} \big) \in \mathbb{N}_0^N $ and $ N > M $ grows exponentially  with a rate
		\begin{equation}\label{eq:pressure}
		r^{-1} \coloneqq \lim_{N\to \infty} \frac{\ln C_N}{N} = \sup_{N \geq M} \frac{\ln C_N}{N} > 1,
		\end{equation}
		which does not depend on $b^{(2)}$. 
		The power series~\eqref{def:powerseries} with $ \alpha_n = \| \widehat \Psi_{0,n} \|^2 $ and $ \beta_n = \| \widehat \Psi_{b,n} \|^2 $
		converge for $ | z | < e^{C_q(\gamma) q} \leq R_\beta $ and we have $ r > 0 $. 
		Moreover, the finite-$ N $ estimate~\eqref{renewal decay appendix} holds with
		\begin{equation}
			R< \frac{e^{C_q(\gamma)q}}{ r} \left(\frac{1}{1+\frac{e^{-b^{(2)}_N}}{1-e^{-C_q(\gamma)q}}}\right).
		\end{equation}
		\end{cor}
		\begin{proof}
		Keeping $ M $ fixed,  the norm-squared $ C_n \coloneqq \| \Psi_{b,n} \|^2 $ satisfies~\eqref{renewal equation appendix} 
		with $ \alpha_n = \| \widehat \Psi_{0,n} \|^2 $ and $ \beta_n = \| \widehat \Psi_{b,n} \|^2 $. This is the content of~\eqref{eq:renewal}. 
		
		The existence of the limit and its representation as a supremum in~\eqref{eq:pressure} follows from  Fekete's lemma by the supermultiplicativity~\eqref{eq:supermult}. 
		Together with \eqref{renewal equation appendix},  this implies $C_{N+2} \geq (1+\alpha_2)C_N$, and therefore $r^{-1}\geq 1+\alpha_2>1.$ The rate does not depend on $b^{(2)}$, as we get from Theorem~\ref{thm:EGap} that there exists a constant $K>0$ such that for all $N>M$ we have $K^{-1}C_M \Vert \Psi_N\Vert^2 \leq C_N \leq K C_M \Vert \Psi_N\Vert^2$. The claim about the radii of convergence follows from the exponential bound~\eqref{eq:normirrep}. 
		By general renewal theory  \cite[XIII, Theorem 1]{feller1968introduction}, one then has $ r > 0 $ and~\eqref{eq:renewalFeller} holds. The proof is finished by an application of Theorem~\ref{thm:renewal appendix}. 
		\end{proof}

		\section{Comparison with previous constructions} \label{app:B}
		
		In this appendix, we review the construction of the iMPS representation of \cite{estienne2013fractional} and explain how to derive it from our representation rigorously. In particular, we rigorously link the mode operators of~\cite{estienne2013fractional} with the operators defined and studied in Section~\ref{sec:iMPS}. 
		
		\subsection{Formal construction of chiral CFT fields}
		The work~\cite{estienne2013fractional} formally starts with a (pre-)Hilbert space generated by operators $(a_n)_{n\in \mathbb{Z}}$ and $\mathrm{exp}(\pm i \varphi_0)$ acting on some abstract vacuum vector $\vert 0 \rangle$, which satisfy the following commutation relations
		\begin{equation} \label{commutation relation 2}
			[a_n, a_m] = n \delta_{n+m,0}, \qquad [\varphi_0, a_n] = i \delta_{n,0}, \qquad \text{for all } n,m \in \mathbb{Z} ,
		\end{equation}
		and which annihilate the vacuum for positive integers, 
		\begin{equation} \label{cond on annihilation operators}
			a_n \vert 0 \rangle = 0, \qquad \text{for all } n\in \mathbb{N}.
		\end{equation}
		One needs to be cautious with the interpretation of these requirements. 
			As is well known  \cite{vonNeumann1931eindeutigkeit}, all irreducible representations of $[\varphi_0, a_0] =i$, when realized as hermitian operators as is assumed here,  are unitarily equivalent to the position and momentum operator on $L^2(\mathbb{R})$, which do not admit a vacuum as in~\eqref{cond on annihilation operators}. 
			
			This is bypassed by not working with irreducible representations and $ \varphi_0 $ directly. We rather give meaning to a suitable version of $\exp(\pm i \varphi_0)$ only. To do so, 
one takes copies of the Hilbert space $ \mathcal{H} $  defined in Section~\ref{sec:iMPS}, and sets
		\begin{equation}\label{eq:directsumG}
			 \mathcal{G} = \bigoplus_{N\in \mathbb{Z}} \mathcal{H}. 
		\end{equation}
		If one labels the vacuum in the $ N $th copy of $ \mathcal{H} $ by $ | N \rangle $ with $ N \in \mathbb{Z} $, the copies are linked by the action of a unitary group, 
		$$
			\vert N\rangle = \exp\left(i \frac{N}{\sqrt{q}}\varphi_0\right) \vert 0 \rangle ,  
		$$
		and the fact that these vacua are the highest weight states with respect to the operators $(a_n)_{n\in \mathbb{Z}}$:
		$$
		a_0 \vert N \rangle = \frac{N}{\sqrt{q}} \vert N \rangle, \qquad a_n \vert N \rangle = 0 \qquad \text{for all } n\in \mathbb{N} . 
		$$	
		Moreover, the scalar product is chosen in such a way that
		\begin{align*}
			a_n^*= a_{-n}, \qquad \text{for all } n\in \mathbb{Z} . 
		\end{align*}

		Central operators in the construction of~\cite{estienne2013fractional} are the vertex operators $V(z)$ labeled by $ z \in \mathbb{C} $ and acting in $\mathcal{G}$. They are formal Laurent series with coefficients in the space of linear operators $\mathrm{End}(\mathcal{G})$. This means $V(z)\in \mathrm{End}(\mathcal{G})[[z, z^{-1}]]$ can be represented formally as 
		\begin{align*}
			V(z) = \sum_{n\in \mathbb{Z}} b_n z^n, \qquad b_n \in \mathrm{End}(\mathcal{G}).
		\end{align*} 
		Such operators are called a field if for each $v\in \mathcal{G} $ there exists $K_v\in \mathbb{Z}$ such that we have $b_n v=0$ for all $n<K_v$. For a product $a_{n_1} \dots a_{n_m}$,  its normal order to be given by 
		$$:a_{n_1} \dots a_{n_m}:=a_{\tau(n_1)} \dots a_{\tau(n_m)}$$
		where $\tau$ is a permutation such that $\tau(n_j)\geq \tau(n_{k})$ for all $1\leq j\leq k \leq m$.
		In \cite{estienne2013fractional} it is proposed to use the normal-ordered exponential $:\exp(i\sqrt{q} \varphi(z)) :$ as the vertex operator, where $\varphi$ is the free field
		\begin{align*}
			\varphi(z) = \varphi_0 -i a_0 \log(z) + i \sum_{n\in \mathbb{Z}\setminus \{0\}} \frac{1}{n} a_n z^{-n}.
		\end{align*}
		Despite its name, the free field is not a field in the above sense, and the use of $\log(z)$ does not make sense: in order to expand the logarithm into a Laurent series centred around zero it would need to be analytic in an annulus around the origin, which is not the case. One should think of $: \exp(\sqrt{q} a_0 \log(z)):$ as some formal object which one needs to attach meaning to a posteriori.

\subsection{Definition of vertex and mode operators and iMPS representation}
By the definition of the normal ordering, the vertex operators in the above construction are equal to 
\begin{align} \label{vertex operator Estienne et al}
	V(z) = \exp\left(-\sqrt{q} S_-(z)\right) \exp\left(-\sqrt{q} S_+(z)\right) \exp(i \sqrt{q} \varphi_0) \exp(\sqrt{q} a_0 \log(z)),
\end{align}
with $S_-(z), S_+(z)$ as in \eqref{def S+ S-}. To give an analytical meaning to these operators, one proceeds as in Section~\ref{sec:iMPS} and defines the subspaces 
\begin{align*}
		\mathcal{G}_{0,N} \coloneqq \mathrm{span}\left\{ \left(\prod_{k=1}^L (a_{k}^*)^{n_k}\right) \vert N \rangle \ : \ L\in \mathbb{N}_0, n_k \in \mathbb{N}_0 \right\} 
	\end{align*}
	whose closures $ \mathcal{G}_{N} \coloneqq \overline{ \mathcal{G}_{0,N}} $ then yield the full Hilbert space $ \mathcal{G} =  \bigoplus_{N\in \mathbb{Z}} \mathcal{G}_{N} $.
	We denote by $P_N$  the orthogonal projection onto $  \mathcal{G}_{N}  $, and set $ T $ on the dense domain $\mathcal{G}_0 \coloneqq  \bigoplus_{N\in \mathbb{Z}} \mathcal{G}_{0,N} $
	\begin{align*}
		T :  \mathcal{G}_0  \rightarrow  \mathcal{G}_0  \quad T \vert N \rangle \coloneqq \vert N+q \rangle , \; \mbox{and extended through the requirement} \quad [ T, a_n^* ] = 0 \; \mbox{for all $ n \in \mathbb{N} $.}
	\end{align*}
	The operator $T$ then extends to a unitary operator on $\mathcal{G}$. 
	Also extending the operators $W_m$ defined in \eqref{definition Wm} to the direct sum  $\mathcal{G}_0 $, the mode operator for $ m \in \mathbb{Z} $  is then defined on the dense domain $\mathcal{G}_0 $ by
	\begin{equation} \label{definition Vm}
		V_{-m-h} \coloneqq \sum_{N\in \mathbb{Z}} W_{m-N} T P_N .
	\end{equation}
	The integer offset $h$ is referred to as the conformal dimension of the vertex operator
	\begin{equation} \label{mode operators Estienne et al}
	V(z) = \sum_{m\in \mathbb{Z}} V_{-m-h} z^m .
\end{equation}
	when interpreted as a formal Laurent series. Note that
	the operators $T$ and $P_N$ are bounded and both commute with all $W_m$'s. Since $a_0 v_N = (N/\sqrt{q})  v_N$ for all $v_N \in \mathcal{G}_N$, the action of the last factor in~\eqref{vertex operator Estienne et al} is hence
	\begin{align*}
		\exp(\sqrt{q} a_0 \log(z)) v_N = \exp(\sqrt{q} (m/\sqrt{q}) \log(z)) v_N = z^N v_N
		= z^N P_N v_N
	\end{align*}
	for all $z\neq 0$. 
	The second factor in~\eqref{vertex operator Estienne et al} equals $ T $. The first factor agrees with $W(z) $ in~\eqref{eq:informalW}. Consequently, the definition~\eqref{definition Vm} indeed gives an analytical meaning to 
	the normal ordered vertex operator~\eqref{vertex operator Estienne et al}. \\

The iMPS representation of the fractional quantum Hall wavefunctions derived in Section~\ref{sec:iMPS} may be reformulated in terms of the mode operators. To do so, we note that for every $w_N\in \mathcal{G}_N$, we have
	$
		V_{-m-h} w_N = W_{m-N} T w_N $, and hence for any $ \mathbf{k} \in \mathbb{Z}^N $: 
\begin{align}
	\langle qN \vert V_{-k_1-h} \dots V_{-k_N-h} \vert 0 \rangle 
	&= \langle qN \vert W_{k_1-q(N-1)} W_{k_2-q(N-2)} \dots W_{k_{N-1}-q} W_{k_N} \vert qN \rangle \notag \\
	&= \langle 0 \vert W_{k_1-q(N-1)} W_{m_2-q(N-2)} \dots W_{k_{N-1}-q} W_{k_N} \vert 0 \rangle,
\end{align}
where we have used that the $W_m$'s commute with $T$ and that $T$ is unitary.  
Using the permutation property~\eqref{eq:permutations}  of the $ W_m $'s, we also conclude
that for all $ \sigma \in \mathcal{S}_N $:
\begin{equation} \label{permutation V}
	\langle qN \vert V_{-k_{\sigma(1)}-h} \dots V_{-k_{\sigma(N)}-h} \vert 0 \rangle
	= \mathrm{sign}(\sigma)^q  \ \langle qN \vert V_{k_1-h} \dots V_{k_N-h} \vert 0 \rangle . 
\end{equation}

\begin{rem} Despite being repeated throughout the literature,  the mode operators $V_{-m-h}$ do not (anti)commute with each other in general. Indeed, one easily computes 
$
	V_{-m-h} V_{-k-h} \vert \ell \rangle = W_{m-q-\ell} W_k \vert 2q + \ell \rangle $ 
and hence for $m=0, k=-q$ and $\ell=-2q$, we arrive at
\begin{align*}
	V_{0-h} V_{-(-q)-h} \vert 0 \rangle = W_q W_{-q} \vert 0 \rangle = 0 \neq \vert 0 \rangle = W_0 W_0 \vert 0 \rangle = V_{-(-q)-h} V_{0-h} \vert 0 \rangle.
\end{align*}
\end{rem}

\paragraph{Data Availability Statement:} No datasets were generated or analyzed for this
article.
		
\paragraph{Acknowledgements.} This work was supported by the DFG under grant TRR 352--Project-ID 470903074. SW was also supported under EXC-2111 -- 390814868.  S.S. would like to thank Dmytro Bondarenko for helpful discussions.

\paragraph{Conflict of interests} The authors have no competing interests to declare that
	are relevant to the content of this article.

\paragraph{Ethical Statement} Not applicable.

\paragraph{Informed Consent} Not applicable.

		\bibliography{biblio}
		\bibliographystyle{plain}

	\end{document}